\documentclass[acmsmall,authorversion,nonacm]{acmart}

\AtBeginDocument{%
  \providecommand\BibTeX{{%
    \normalfont B\kern-0.5em{\scshape i\kern-0.25em b}\kern-0.8em\TeX}}}

\usepackage{multirow}
\usepackage{booktabs}
\usepackage[subrefformat=parens,labelformat=parens]{subfig}
\usepackage{bm}
\usepackage{balance}
\usepackage{tikz}
\usepackage{colortbl}
\usepackage{subfig}
\usepackage{graphicx}
\usetikzlibrary{positioning}
\usepackage{enumitem}
\usepackage{amsthm}
\usepackage{amsmath}
\usepackage{amsfonts}
\usepackage{bbm}
\usepackage{apptools}
\usepackage{xcolor}
\usepackage[ruled]{algorithm2e}
\usepackage{algorithmic}
\usepackage{placeins}

\usepackage{xcolor}
\definecolor{red}  {rgb}{0.9,0.0,0.0}
\definecolor{green}{rgb}{0.0,0.7,0.0}
\definecolor{blue} {rgb}{0.0,0.0,0.9}

\DeclareMathOperator*{\argmax}{arg\,max}

\newcommand{\sinnamon}{\textsc{Sinnamon}}
\newcommand{\linscan}{\textsc{LinScan}}
\newcommand{\roaringlinscan}{\textsc{LinScan-Roaring}}
\newcommand{\sinnamonplus}{\sinnamon$^+$}

\begin{document}

\title{An Approximate Algorithm for Maximum Inner Product Search over Streaming Sparse Vectors}

\author{Sebastian Bruch}
\affiliation{%
  \institution{Pinecone}
  \city{New York}
  \state{NY}
  \country{USA}
}
\email{sbruch@acm.org}
\orcid{0000-0002-2469-8242}

\author{Franco Maria Nardini}
\affiliation{%
  \institution{ISTI-CNR}
  \city{Pisa}
  \country{Italy}
}
\email{francomaria.nardini@isti.cnr.it}
\orcid{0000-0003-3183-334X}

\author{Amir Ingber}
\affiliation{%
  \institution{Pinecone}
  \city{Tel Aviv}
  \country{Israel}
}
\email{ingber@pinecone.io}
\orcid{0000-0001-6639-8240}

\author{Edo Liberty}
\affiliation{%
  \institution{Pinecone}
  \city{New York}
  \state{NY}
  \country{USA}
}
\email{edo@pinecone.io}
\orcid{0000-0003-3132-2785}




\begin{abstract}
Maximum Inner Product Search or top-$k$ retrieval on sparse vectors is well-understood in information retrieval, with a number of mature algorithms that solve it exactly. However, all existing algorithms are tailored to text and frequency-based similarity measures. To achieve optimal memory footprint and query latency, they rely on the near stationarity of documents and on laws governing natural languages. We consider, instead, a setup in which collections are streaming---necessitating dynamic indexing---and where indexing and retrieval must work with arbitrarily distributed real-valued vectors. As we show, existing algorithms are no longer competitive in this setup, even against na\"ive solutions. We investigate this gap and present a novel \emph{approximate} solution, called \sinnamon{}, that can efficiently retrieve the top-$k$ results for sparse \emph{real valued} vectors drawn from arbitrary distributions. Notably, \sinnamon{} offers levers to trade-off memory consumption, latency, and accuracy, making the algorithm suitable for constrained applications and systems. We give theoretical results on the error introduced by the approximate nature of the algorithm, and present an empirical evaluation of its performance on two hardware platforms and synthetic and real-valued datasets. We conclude by laying out concrete directions for future research on this general top-$k$ retrieval problem over sparse vectors.
\end{abstract}

\maketitle              % typeset the header of the contribution

\section{Introduction}
\label{section:introduction}

Many applications of information retrieval, as the name of the discipline suggests, reduce to or involve the fundamental and familiar question of \emph{retrieval}. In its most general form, it aims to solve the following problem:
\begin{equation}
    \argmax^{(k)}_{x \in \mathcal{X}} \; f(q, x),
    \label{equation:retrieval-problem}
\end{equation}
to find, from a collection $\mathcal{X}$, a subset of $k$ objects
that are the most relevant to a query object $q \in \mathcal{Q}$ according to a similarity function $f: \mathcal{Q} \times \mathcal{X} \rightarrow \mathbb{R}$. In many instances, this manifests as the Maximum Inner Product Search (MIPS) problem where $\mathcal{X}, \mathcal{Q} \subset \mathbb{R}^n$ and $f(\cdot, \cdot)$ is the inner product of its arguments:
\begin{equation}
    \argmax^{(k)}_{x \in \mathcal{X}} \; \langle q \; , \; x \rangle.
    \label{equation:mips}
\end{equation}

As a prominent example, consider a multi-stage ranking system~\cite{asadi2013phd,asadi2013efficiency,yin2016ranking} in the context of text retrieval. 
The cascade of ranking functions often begins with lexical or semantic similarity search which can be formalized using Equation~(\ref{equation:mips}). 

When similarity is based on Term Frequency-Inverse Document Frequency (TF-IDF), for example, $\mathcal{X}$ is made up of high-dimensional vectors, one per document.
Each document vector contains non-zero entries that correspond to terms and their frequencies in that document. 
Here, there is a one to one mapping between dimensions in sparse document vectors and terms in the vocabulary.
Each non-zero entry of a query vector $q$ records the corresponding term's inverse document frequency. BM25~\cite{bm25, bm25original} and
many other popular lexical similarity measures can similarly be expressed in the form above. 

When similarity is based on the semantic closeness of pieces of text, then vectors can be produced by embedding models (e.g.,~\cite{nogueira2020monot5,nogueira2020passage,formal2022splade,reimers-2019-sentence-bert}). 
This formulation trivially extends to joint lexical-semantic search~\cite{wang2021bert,chen2022ecir,Kuzi2020LeveragingSA,bruch2022fusion} too.

This deceptively simple problem is difficult to solve efficiently and effectively in practice. When the coordinates of each vector are almost surely non-zero---a case we refer to as \emph{dense} vectors---then there are volumes of algorithms such as graph-based methods~\cite{malkov2016hnsw,Johnson2021faiss,freshdiskann,NEURIPS2019_0fd7e4f4}, product quantization~\cite{pq2011,pcpq,scann}, and random projections~\cite{fjlt,edo-fastjk-soda,edo-fastjl-acm} that may be used to quickly find an approximate solution to Equation~(\ref{equation:mips}). 
But when vectors are \emph{sparse} in a high-dimensional space (i.e., have thousands to millions of dimensions) with very few non-zero entries, then no general efficient solution exists: Because of the near-orthogonality of most vectors in a sparse high-dimensional regime, algorithms for MIPS do not port over successfully.

It is only by imposing additional constraints on the vectors that the literature approaches this problem at scale and offers solutions that meet certain memory, time, and accuracy constraints. 
Algorithms that rely on sketching cover only binary or categorical-valued vectors~\cite{binsketch,cat_binsketch}. 
Inverted index-based algorithms that are more commonly used in information retrieval, such as WAND and its descendants~\cite{broder2003wand,ding2011bmwand,topk_bmindexes,mallia2019faster-blockmaxwand,mallia2017blockmaxwand_variableBlocks} and JASS~\cite{jass}, as well as signature-based algorithms~\cite{bitfunnel,asadi2013bloom,countminsketch} all make a number of crucial assumptions: that vectors are non-negative and integer-valued; that their non-zero entries follow a Zipfian distribution; that the share of the contribution of entries to the final score is non-uniform (i.e., some entries contribute more heavily to the final score than others); and that query vectors have very few non-zero entries.

Many of these constraints have historically held given the nature of text data and keyword queries in search engines. 
But when the sparse vectors are the output of embedding models~\cite{sparterm,formal2021splade,formal2022splade,zhuang2022reneuir,dai2020sigir,coil,mallia2021learning,zamani2018cikm}, many of these assumptions need not hold. 
For example, query vectors produced by the SPLADE model~\cite{formal2021splade} have, on average, about $43$ non-zero, real entries on the MS MARCO Passage v1 dataset~\cite{nguyen2016msmarco}---far too many for algorithms such as WAND to operate efficiently~\cite{lassance2022sigir} and, without discretization into integers, incompatible with existing algorithms. 
While there are efforts to make such model-generated representations more sparse by way of regularization or pooling~\cite{yang2021sparsifying,lassance2022sigir}, the underlying Sparse MIPS problem (SMIPS) for unconstrained real vectors remains mostly unexplored.

We investigate that handicap in this work because we believe SMIPS to be of increasing importance as evidenced by the examples above.
Efficiently solving the SMIPS problem enables further innovation in text retrieval and other related areas.  
In our search for an algorithm, we pay particular attention to the more difficult \emph{online} SMIPS problem, where we assume no knowledge of the streaming collection $\mathcal{X}$ and require that the algorithm supports online insertions and deletions. 
We introduce this particular challenge to support real-world use-cases where collections change rapidly, as well as emerging research on Retrieval-Enhanced Machine Learning~\cite{zamani2022reml} where a learning algorithm interacts with a retrieval system during the training process, thereby needing to search for, insert, and delete objects in and from a dynamic collection.

Furthermore, we explore the online SMIPS problem in the context of a vector database depicted in Figure~\ref{figure:architecture}. 
In particular, we assume that the system has (and is required to have) an efficient storage system that contains all active vectors. 
In this setup, the (exact or approximate) top-$k$ retrieval engine may access the vector storage during query execution.

\begin{figure}[t]
\begin{center}
\centerline{
\includegraphics[height=1.1in]{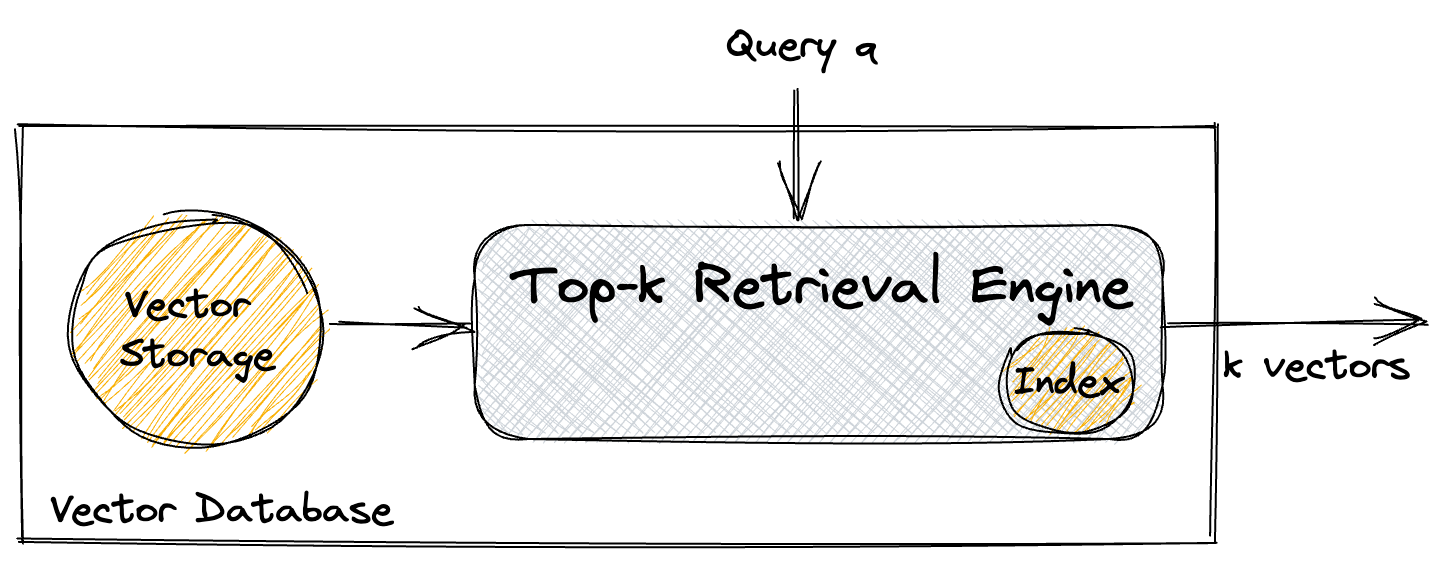}}
\caption{A vector database system consisting of a storage system and an exact or approximate top-$k$ retrieval engine that solves the MIPS/SMIPS problem of Equation~(\ref{equation:mips}).}
\label{figure:architecture}
\end{center}
\end{figure}

Given the setup above, we establish a baseline by revisiting a na\"ive, \emph{exact} algorithm, which we call \linscan{}, to approximate Equation~(\ref{equation:mips}).
Here, $q, x \in \mathbb{R}^n$ and their number of non-zero entries is much fewer than $n$.
That is, $\psi_q, \psi_d \ll n$ where $\psi_q$ and  $\psi_d$ denote the number of non-zero entries in a query vector and the document vector respectively. 
\linscan{} simply stores pairs of vector identifiers and coordinate values in an inverted index that is optionally compressed.
During retrieval, it traverses the index one coordinate at a time to accumulate the inner product scores. 
As we show in this work, \linscan{} proves surprisingly competitive because it takes advantage of instruction-level parallelism and efficient cache utilization available on modern CPUs.

We then build on \linscan{} and propose an online, \emph{approximate} algorithm called \sinnamon{}.
It approximately solves Equation~(\ref{equation:mips}) for sparse vectors, with the implication that 
some of the candidates in the top-$k$ set may be there erroneously. 
As we will show, this tolerance for error in the top-$k$ set makes it possible to tailor the approximate retrieval stage to meet a given set of time, space, and accuracy constraints.

\sinnamon{} makes use of two data structures. 
One, which is more familiar to the reader, is a lean, dynamic inverted index. 
In \sinnamon{}, this index is simply a mapping from a coordinate to the identifier of vectors in $\mathcal{X}$ that have a non-zero value in that coordinate. 
In other words, we maintain inverted lists that contain just vector identifiers. 
This structure allows us to quickly verify if the $j$th coordinate of a vector $x$ is non-zero (i.e., $\mathbbm{1}_{x[j] \neq 0}$) and obtain the set of vectors whose $j$th coordinate is non-zero: $\{ i \;|\; \mathbbm{1}_{x_i[j] \neq 0}\}$.

Coupled with the inverted index is a novel probabilistic sketch data structure.\footnote{We use ``sketch'' to describe a compressed data structure that approximates a high-dimensional vector, and ``to sketch'' to describe the act of compressing a vector into a sketch.} 
A high-dimensional sparse vector $x \in \mathbb{R}^n$ is sketched as $\tilde{x} \in \mathbb{R}^{2m}$ (${2m} \ll n$) using a lossy transformation $e(\cdot): \mathbb{R}^n \rightarrow \mathbb{R}^{2m}$. 
Together with the inverted index, this sketch offers an inverse transformation $e^{-1}(\cdot)$ such that for an arbitrary query vector $q$, we have that $\langle q \; , \; e^{-1}(\tilde{x}) \rangle \geq \langle q \; , \; x \rangle$ and the difference can be tightened by the parameters of the algorithm. 
Another crucial property of our data structure is that, much like the machinery of a Counting Bloom filter~\cite{counting-bloom-filter}, obtaining the value of the $j$th coordinate of $e^{-1}(\tilde{x})$ can be done efficiently and with access to the same coordinates of the sketch regardless of the input, $\tilde{x}$.

As a document vector\footnote{We refer to vectors that are expected to be indexed as ``document vectors'' or simply ``documents,'' and call the input to the retrieval algorithm as the ``query vector,'' ``query point,'' or simply ``query.''} $x_i$ is inserted into $\mathcal{X}$, we record the identifier of its non-zero coordinates in the inverted index and subsequently insert its sketch $\tilde{x}_i$ into the $i$th column of a sketch matrix $\tilde{\mathcal{X}} \in \mathbb{R}^{2m \times |\mathcal{X}|}$.
When we receive a query $q$, we use a coordinate-at-a-time algorithm that efficiently computes the inner product scores by accessing only a single or a fixed group of rows in $\tilde{\mathcal{X}}$ per coordinate. When deleting $x_i$, we simply remove its identifier $i$ from the inverted index and mark the $i$th column in $\tilde{\mathcal{X}}$ as vacant.

In addition to a theoretical analysis of our data structure, we extensively evaluate \linscan{} and \sinnamon{} on benchmark retrieval datasets and demonstrate their many interesting properties empirically. 
We show that due to predictable and regular memory access patterns, both algorithms are fast on modern CPUs. 
We discuss further how we may control the memory usage of \sinnamon{} by adjusting the sketch dimensionality $2m$, 
and tune a knob within the sketching algorithm to control its approximation error and retrieval accuracy.
Moreover, due to their coordinate-at-a-time query processing logic, \linscan{} and \sinnamon{} can be trivially turned into \emph{anytime} algorithms, terminating retrieval once an allotted time budget is exhausted. 
Finally, as we will demonstrate in this work, it is straightforward to parallelize computations in \linscan{} and \sinnamon{}. 
These properties make these methods the first SMIPS algorithms for real vectors that allow one to explore the Pareto frontier of effectiveness and time- and space-efficiency, to trivially scale indexes vertically through parallelization, and to tailor it to the needs of resource-constrained environments and applications.

We begin this work with a review of the relevant literature on this topic in Section~\ref{section:related-work}. 
We then describe the \linscan{} and \sinnamon{} algorithms in Sections~\ref{section:linscan} and~\ref{section:algorithm} respectively. 
That presentation is followed by a detailed error analysis of the data structure and retrieval algorithm in Section~\ref{section:analysis}, and a comprehensive empirical evaluation on two hardware platforms and a variety of sparse vector collections in Section~\ref{section:evaluation}. 
We conclude this work with a discussion in Section~\ref{section:discussion}.

\section{Related Work}
\label{section:related-work}

The information retrieval literature offers numerous algorithms to solve a constrained variant of the SMIPS problem that is specifically tailored to text retrieval and keyword search. The research on that topic has advanced the field considerably over the past few decades, making retrieval one of the most efficient components in a modern search engine. We do not review this vast literature here and refer the reader to existing surveys~\cite{zobel2006inverted_files,tonellotto2018survey} for details. Instead, we briefly review key algorithms and explain what makes them less suitable to operate in the setup we consider in this work.

Among the many algorithms in existence, WAND~\cite{broder2003wand} and its intellectual descendants and incremental optimizations~\cite{ding2011bmwand,topk_bmindexes,mallia2019faster-blockmaxwand,mallia2017blockmaxwand_variableBlocks} have become the \emph{de facto} top-$k$ retrieval solution. The core logic in WAND and other related algorithms centers around a document-at-a-time traversal of the inverted index. By maintaining an upper-bound on the partial score contribution of each coordinate to the final inner product, we can quickly tell if a document may possibly end up in the top $k$ set: if it appears in enough inverted lists whose collective score upper-bound exceeds the current threshold, then it is a candidate to be fully evaluated; otherwise, it has no prospect of ever making it to the top-$k$ set and can therefore be safely rejected without further computation.
 
The excellent performance of this logic rests on a number of important assumptions, however. Like all other existing algorithms, it is designed primarily for non-negative vectors. Due to its irregular memory access pattern, the algorithm operates better when the query has only a few non-zero coordinates. But perhaps its key assumption is the fact that word frequencies in natural languages often follow a Zipfian distribution. Given the role that word frequencies play in relevance measures such as BM25~\cite{bm25original}, the Zipfian shape implies that some words (i.e., coordinates) are inherently far more important than others. That, in turn, boosts or attenuates the contribution of the coordinate to the final inner product score, making the distribution of upper-bounds over coordinates quite skewed. Such skewness contributes heavily to the success of WAND and other dynamic pruning algorithms~\cite{tonellotto2018survey}.

While non-negativity and high query sparsity can be relaxed, the algorithm duly redesigned, and its implementation optimized for a more general regime, the reliance on Zipfian data is less forgiving. When the distribution of non-zero coordinates deviates from the Zipfian curve, the coordinate upper-bounds become more uniform, leading to less effective pruning of the inverted lists, and therefore a less efficient top-$k$ retrieval. That, among other problems~\cite{crane2017wsdm}, renders this particular idea of pruning less suitable for a general purpose top-$k$ retrieval for sparse vectors where coordinates take on a non-zero value (nearly) uniformly at random.

Other competing index traversal techniques process a query in a coordinate-at-a-time or score-at-a-time manner. Both of these approaches rely on sorting inverted lists by term frequencies or their ``impact scores'' (i.e., precomputed partial scores). The machinery within these algorithms, however, has the added disadvantage that it relies on the stationarity of the dataset to compute impact scores or sort postings, making it undesirable for streaming collections that require fast updates to the index~\cite{tonellotto2018survey}.

In contrast to the multitude of data structures and algorithms for stationary datasets, the literature on retrieval in streaming collections is rather slim and limited to a few works~\cite{asadi2012bloom,asadi2013phd,asadi2013bloom}. Notably, Asadi and Lin~\cite{asadi2012bloom,asadi2013bloom} used Bloom filters~\cite{bloom-filter} to speed up postings list intersection in conjunctive and disjunctive queries at the expense of accuracy and memory footprint. These approximate methods proved instrumental in creating an end-to-end algorithm for retrieval and ranking of streaming documents~\cite{asadi2013phd}. While these works are related to the question we investigate in this work, the proposed methods are not directly applicable: We are not interested in set membership tests for which Bloom filters are a natural choice, but rather in approximating real-valued vectors in such a way that leads to arbitrarily accurate inner product with a query vector.

Another relevant topic is the use of signatures for retrieval and inner product approximation~\cite{bitfunnel,binsketch,cat_binsketch}. Pratap et al. propose a simple algorithm~\cite{binsketch} to sketch sparse \emph{binary} vectors in such a way that the inner product of sketches approximates the inner product of original vectors. The core idea is to randomly map coordinates in the original space to coordinates in the sketch. When two or more entries collide, the sketch records the OR of the colliding values. A later work extends this idea to categorical-valued vectors~\cite{cat_binsketch}. Nonetheless, it is not obvious how the proposed sketching mechanisms may be extended to real-valued vectors.

Deviating from the standard inverted index solution to top-$k$ retrieval is the work of Goodwin et al.~\cite{bitfunnel}. As part of what is referred to as the BitFunnel indexing machinery, the authors propose to record and store a bit signature for every document vector in the index using Bloom filters. These signatures are scanned during retrieval to deduce if a document contains the terms of a conjunctive query. While it is encouraging that a signature-based replacement to inverted indexes appears not only viable but very much practical, the query logic BitFunnel supports is limited to ANDs and does not generalize to the setup we are considering in this work. Despite that, we note that the ``bit-sliced signatures'' in BitFunnel inspired the particular transposed layout of the sketch matrix in \sinnamon{}.

For completeness, we also briefly note the literature on sparse-sparse matrix multiplication~\cite{smith2015splatt,li2018hicoo,8327050,fowers2014high,9065579,9251978} and sparse matrix-sparse vector multiplication~\cite{Azad2017AWP}. The main challenge in operations concerning sparse matrices is that the computation involved is often highly memory-bound. As such, much of this literature focuses on developing sparse storage formats with hardware- and cache-aware designs that lead to a more effective utilization of memory bandwidth. We believe, however, that the research on compact storage and memory-efficient structures is orthogonal to the topic of our work and offers solutions that could lead to improvements across all algorithms considered in this work.

\section{\linscan{}: An Exact SMIPS Baseline}
\label{section:linscan}

Let us begin with a na\"ive and exact baseline. 
This algorithm, which we call \linscan{}, constructs an inverted index that maps a coordinate (e.g., $j \in [n]$) to a list of ``postings''. 
Each posting is a pair consisting of the identifier $i$ of a document vector $x_i \in \mathbb{R}^n$ and a value $x_i[j]$.
This adds $2 \times \psi_d$ values to the index for every document. 
In our implementation, we use two parallel arrays to implement an inverted list (also known as non-interleaved inverted lists), one that stores vector identifiers and another that holds values. 
For completeness, we show this indexing logic in Algorithm~\ref{algorithm:linscan:indexing}.

\begin{algorithm}[!t]
\SetAlgoLined
{\bf Input: }{Sparse document vector $x_i \in \mathbb{R}^n$.}\\
{\bf Requirements: }{Inverted index $\mathcal{I}$.}\\

\begin{algorithmic}[1]
    \STATE $nz(x_i) \leftarrow \{ j \;|\; x_i[j] \neq 0 \}$
    \STATE $\mathcal{I}[j].\textsc{Insert}((i, x_i[j]))\;\; \forall j \in nz(x_i)$
 \end{algorithmic}
 \caption{Indexing in \linscan{}}
\label{algorithm:linscan:indexing}
\end{algorithm}

In its most basic variant, we store the inverted index without using any form of compression: That is, the document identifiers are stored as $32$-bit integers and values as $32$-bit floats.
This allows us to quantify the latency of the logic within the algorithm itself and remove other factors related to compression. 
It also enables the algorithm to take advantage of instruction-level parallelism and efficient caching that come for free (using default compiler optimization techniques) with a coordinate-at-a-time retrieval strategy. 
To make the algorithm more practical, we also consider a variant where the list of vector identifiers in each inverted list is compressed using the Roaring~\cite{roaring} dynamic data structure and the values are stored using the \textit{bfloat16} standard ($16$-bit floating points). 
The loss of precision due to the conversion from $32$-bit floats to $16$-bit values is negligible in practice. 
We denote this variant by \roaringlinscan{}.

\begin{algorithm}[t]
\SetAlgoLined
{\bf Input: }{Sparse query vector $q \in \mathbb{R}^n$; Number of vectors to return, $k$.}\\
{\bf Requirements: }{Inverted index $\mathcal{I}$.}\\
\KwResult{Exact solution set of Equation~(\ref{equation:mips})}

\begin{algorithmic}[1]
    \STATE $nz(q) \leftarrow \{ j \;|\; q[j] \neq 0 \}$ \label{algorithm:linscan:retrieval:scoring-begins}
    \STATE $\textit{scores}[i] \leftarrow 0 \;\forall i$
    \FORALL{$j \in nz(q)$} \label{algorithm:linscan:retrieval:scoring}
        \FORALL{$(i, x_i[j]) \in \mathcal{I}[j]$}
            \STATE $\textit{scores}[i] \leftarrow \textit{scores}[i] + x_i[j] \times q[j]$
        \ENDFOR
    \ENDFOR \label{algorithm:linscan:retrieval:scoring-ends}
    \RETURN $\textsc{FindLargest}(\mathit{scores}, k)$ \label{algorithm:linscan:retrieval:find-largest}
 \end{algorithmic}
 \caption{Retrieval in \linscan{}}
\label{algorithm:linscan:retrieval}
\end{algorithm}

\begin{algorithm}[t]
\SetAlgoLined
{\bf Input: }{An array of floating point values, \emph{scores}; $k$}\\
\KwResult{Index of the largest $k$ values.}

\begin{algorithmic}[1]
    \STATE heap $\leftarrow$ \textsc{BinaryHeap}
    \STATE $\theta \leftarrow -\infty$
    \FORALL{$i \in [\textsc{Len}(\textit{scores})]$}
        \IF{$\textit{scores}[i] > \theta$}
            \STATE heap.\textsc{Insert}$((\textit{key}=i, \textit{value}=\textit{scores}[i]))$
        \IF{heap.$\textsc{Size}() > k$}
            \STATE $\theta \leftarrow \text{heap.}\textsc{Pop}().\textit{value}$
        \ENDIF
        \ENDIF
    \ENDFOR
    \RETURN \emph{key}s of items in heap
\end{algorithmic}
\caption{A sample implementation of \textsc{FindLargest}}
\label{algorithm:find-largest}
\end{algorithm}

During retrieval, \linscan{} follows a simple two-step procedure shown in Algorithm~\ref{algorithm:linscan:retrieval}. 
In the \emph{scoring} step (lines~\ref{algorithm:linscan:retrieval:scoring-begins} through~\ref{algorithm:linscan:retrieval:scoring-ends}), it traverses the inverted index one coordinate at a time for every non-zero coordinate in the query vector and accumulates partial scores for all documents. 
At the end of this step, the algorithm will have computed the exact inner product scores for every vector in the collection---document vectors that are not visited in the scoring loop on line~\ref{algorithm:linscan:retrieval:scoring} will have a score of $0$. In the \emph{ranking} step (line~\ref{algorithm:linscan:retrieval:find-largest}), it finds the top $k$ vectors with the largest inner product scores; in our implementation of \textsc{FindLargest}, we use a heap to efficiently identify the top $k$ vectors as shown in Algorithm~\ref{algorithm:find-largest}.

An interesting property of \linscan{}'s retrieval algorithm is that it is trivial to execute its logic in parallel in a dynamic manner. 
While most existing algorithms would require some form of sharding of the index by document ids (i.e., keeping separate index structures for different ranges of document ids), \linscan{} can execute retrieval with as many threads as are available using the very same monolithic data structure. 
It is the combination of the coordinate-at-a-time nature of \linscan{} and the layout of its data structure that lend the algorithm to such a dynamically adjustable level of concurrency. By breaking up an inverted list into contiguous segments on the fly, we can accumulate partial scores for each segment concurrently. Similarly, it is just as trivial to execute $\textsc{FindLargest}(\cdot)$ in parallel. 
We consider this parallel variant of the algorithm in this work and refer to it as $\linscan{}^\parallel$.

\begin{algorithm}[t]
\SetAlgoLined
{\bf Input: }{Sparse query vector $q \in \mathbb{R}^n$; Time budget, $T < \infty$; Number of vectors to re-rank, $k^\prime$; Number of vectors to return, $k$.}\\
{\bf Requirements: }{Inverted index $\mathcal{I}$; Vector storage $\mathcal{S}$.}\\
\KwResult{Approximate solution set of Equation~(\ref{equation:mips})}

\begin{algorithmic}[1]
    \STATE $nz(q) \leftarrow \{ j \;|\; q[j] \neq 0 \}$
    \STATE \textbf{Sort} coordinates in $nz(q)$ in descending order by the absolute value of the vector entries ($|q[j]|$)
    \STATE $\textit{scores}[i] \leftarrow 0 \;\forall i$
    \FORALL{$j \in \textbf{sorted } nz(q)$}
        \FORALL{$(i, x_i[j]) \in \mathcal{I}[j]$}
        \STATE $\textit{scores}[i] \leftarrow \textit{scores}[i] + x_i[j] \times q[j]$
        \ENDFOR
        \STATE \textbf{break} if time budget $T$ exhausted
    \ENDFOR
    \STATE $V \leftarrow \textsc{FindLargest}(\mathit{scores}, k^\prime)$
    \FOR {$v \in V$}
    \item $s_{v} \leftarrow \langle q,\; \mathcal{S}.\textsc{Fetch}(v) \rangle$
    \ENDFOR
    \RETURN $\textsc{FindLargest}(\{ s_v \;|\; v \in V \}, k)$
 \end{algorithmic}
 \caption{Anytime retrieval in \linscan{}}
\label{algorithm:linscan:retrieval:anytime}
\end{algorithm}

Finally, let us introduce an \emph{anytime} but necessarily \emph{approximate} version of \linscan{} by way of a simple modification to the retrieval logic. In this variant, we visit the non-zero query coordinates in order from the coordinate with the largest absolute value to the smallest one. As soon as a given time-budget $T$ is exhausted, we terminate the scoring step of the retrieval; setting $T=\infty$ to give the algorithm an unlimited time budget reduces the logic to the vanilla exact \linscan{}. But because the scores may no longer represent the exact inner products, we find top $k^\prime$ document vectors according to these possibly inexact scores for some $k^\prime \geq k$, and subsequently fetch those vectors from storage to compute their exact scores and finally return the top $k$ elements. This procedure is shown in Algorithm~\ref{algorithm:linscan:retrieval:anytime}.

\subsection{Deletions}
We have so far described the indexing and retrieval algorithms in \linscan{}. In this section, we briefly touch on deletion strategies. We preface this discussion with the note that insertion, deletion, and retrieval procedures are not entirely independent: A particular deletion algorithm may pair better with a particular set of insertion and retrieval algorithms. While we are careful to incorporate this fact into making a choice between different deletion strategies, we acknowledge that more can and should be done to optimize joint insertion-deletion-retrieval efficiency. But as much of the optimization involves heavy engineering (e.g., applying deletions in batch, running a separate background process to reclaim deleted space, etc.), we do not dwell on this point in this work and leave an empirical exploration of this detail to future work.

Throughout this work, when we delete a document vector $x_i$ from a \linscan{} index, we invoke a process that is best described as ``full deletion.'' This strategy simply wipes all postings associated with $x_i$ from the inverted index and frees up the space they occupied. That involves removing a posting from an inverted list, noting that the posting may reside anywhere within the list.

An obvious advantage of this protocol is that it does not produce any waste in memory in the form of zombie postings---space allocated in memory that lingers on after the document has been deleted. Moreover, it is compatible with the insertion and retrieval logic described earlier as it maintains the contiguity of the inverted lists and the alignment between the identifier and value arrays.

An obvious disadvantage of this approach, on the other hand, is that fully removing a posting from an inverted list often leads to a reorganization of the underlying array of data in memory, which is itself a potentially expensive procedure. But we believe that the benefits of the full deletion approach outweigh its pitfalls, especially considering the ramifications of alternative algorithms for the insertion and retrieval procedures.

Consider, for instance, a different method which simply designates a posting as ``deleted'' using a special value without immediately reclaiming its space. Perhaps the space is reclaimed periodically by a background process or recycled when a new posting is inserted into the inverted list. Regardless, it is clear that the insertion and retrieval procedures have to be modified so as to safely handle postings that are designated for deletion. In retrieval, for example, this involves conditioning on the content of each posting, leading to branches in the execution.

Given the discussion above, we believe our choice of full deletion is appropriate for \linscan{} and helps to reduce the overall complexity in insertion and retrieval.

\section{\sinnamon{}: An Approximate SMIPS Algorithm}
\label{section:algorithm}

This section describes the algorithmic details of \sinnamon{}. We begin with a detailed account of the efficient indexing procedure, where we create a sketch matrix and an inverted index. Like \linscan{}, we then present an anytime, coordinate-at-a-time retrieval algorithm that is amenable to parallelism. We subsequently explain how \sinnamon{} supports deletions. In each subsection, we also discuss how this algorithmic framework can be extended in the future and state the research questions that it in turn inspires.

\begin{algorithm}[!t]
\SetAlgoLined
{\bf Input: }{Sparse document vector $x_i \in \mathbb{R}^n$.}\\
{\bf Requirements: }{Inverted index $\mathcal{I}$; Sketch matrix $\tilde{\mathcal{X}} = [\mathcal{U};\; \mathcal{L}] \in \mathbb{R}^{2m \times |\mathcal{X}|}$; $h$ independent random mappings $\pi_o: [n] \rightarrow [m]$.}\\

\begin{algorithmic}[1]
    \STATE $nz(x_i) \leftarrow \{ j \;|\; x_i[j] \neq 0 \}$
    \STATE $\mathcal{I}[j].\textsc{Insert}(i)\;\; \forall j \in nz(x_i)$
    \STATE Let $u, l \in \mathbb{R}^m $
    \FORALL{$k \in [m]$}
        \STATE $u[k] \leftarrow \max_{\{ j \in nz(x_i) \;|\; \exists \;o\; \mathit{s.t.}\; \pi_o(j) = k \}} x_i[j]$ \label{algorithm:indexing:upper-bound-sketch}
        \STATE $l[k] \leftarrow \min_{\{ j \in nz(x_i) \;|\; \exists \;o\; \mathit{s.t.}\; \pi_o(j) = k \}} x_i[j]$
    \ENDFOR
    \STATE $\mathcal{U}^T[i] \leftarrow u$
    \STATE $\mathcal{L}^T[i] \leftarrow l$
 \end{algorithmic}
 \caption{Indexing in \sinnamon{}$_{(2m, \; h)}$}
\label{algorithm:indexing}
\end{algorithm}

\subsection{Indexing}
\label{section:algorithm:indexing}
When a new document vector $x_i$ arrives into the collection $\mathcal{X}$, \sinnamon{} executes an efficient two-step algorithm to index it and make it available for retrieval. The first stage is the familiar procedure of inserting the vector identifier $i$ into an inverted index $\mathcal{I}$. $\mathcal{I}$ is a mapping from coordinates to a list of vectors in which that coordinate is non-zero: $\mathcal{I}[j] \triangleq \{ i \;|\; x_i[j] \neq 0 \}$. When processing $x_i$, for every non-zero coordinate in the set $nz(x_i) \triangleq \{j \;|\; x_i[j] \neq 0 \}$, we insert $i$ into $\mathcal{I}[j]$.

The second, novel step involves populating the column $i$ in the \emph{sketch} matrix $\tilde{\mathcal{X}} \in \mathbb{R}^{2m \times |\mathcal{X}|}$, which has $2m$ rows (the sketch size) and $|\mathcal{X}|$ columns (the collection size). For notational convenience and to simplify prose, we regard $\tilde{\mathcal{X}}$ as a block matrix $\tilde{\mathcal{X}} = [\mathcal{U}_{m \times |\mathcal{X}|};\; \mathcal{L}_{m \times |\mathcal{X}|}]$ with the top half of the matrix denoted by $\mathcal{U}$ and the bottom half by $\mathcal{L}$.

Intuitively, what the sketch of a vector in \sinnamon{} captures is an upper-bound and a lower-bound on the entries of $x_i$ in such a way that its inner product with any query vector can be approximated arbitrarily accurately. 
This sketching step, in effect, can be thought of as a lossy compression of a sparse vector such that the error incurred from losing the original values does not severely degrade the solution set of Equation~(\ref{equation:mips}). We will revisit the effect of this approximation on the final inner product later in this work.

Algorithm~\ref{algorithm:indexing} presents this indexing procedure. The algorithm makes use of $h$ independent random mappings $\pi_o: [n] \rightarrow [m]$, where each $\pi_o(\cdot)$ projects coordinates in the original space to an integer in $[m]$. In the notation of the algorithm, we construct an upper-bound vector $u \in \mathbb{R}^m$ and a lower-bound vector $l \in \mathbb{R}^m$, and insert $u$ and $l$ into the $i$th column of $\mathcal{U}$ and $\mathcal{L}$ respectively. In words, the $k$th coordinate of $u$, $u[k]$ ($l[k]$) records the largest (smallest) value from the set of all entries in $x_i$ that map into $k$ according to at least one $\pi_o$. Figure~\subref*{figure:example:indexing} illustrates the algorithm for an example vector using a single mapping $\pi$.

\begin{figure}[t]
\begin{center}
\centerline{
\subfloat[Indexing]{
\includegraphics[width=0.5\linewidth]{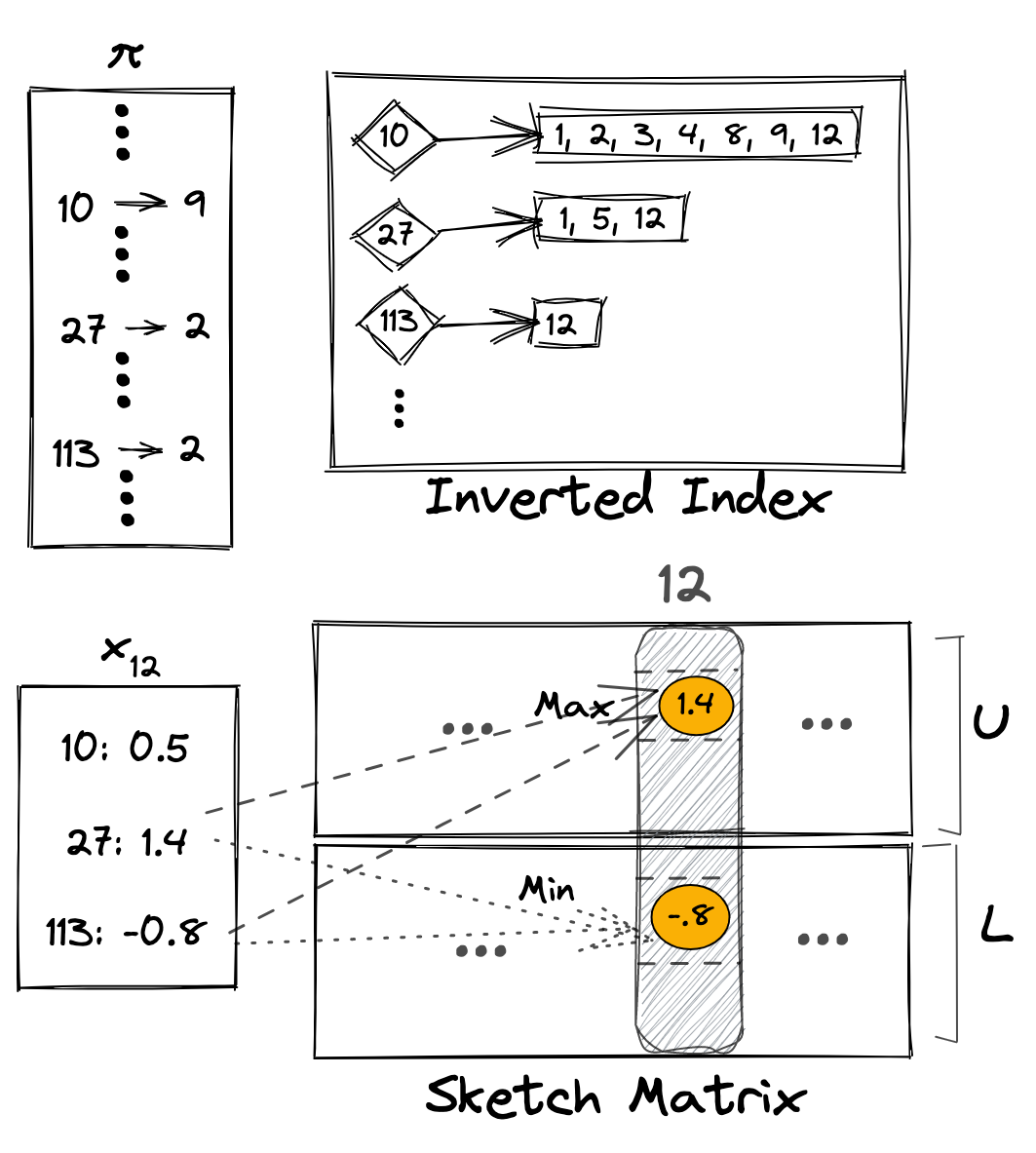}\label{figure:example:indexing}}
\subfloat[Scoring]{
\includegraphics[width=0.5\linewidth]{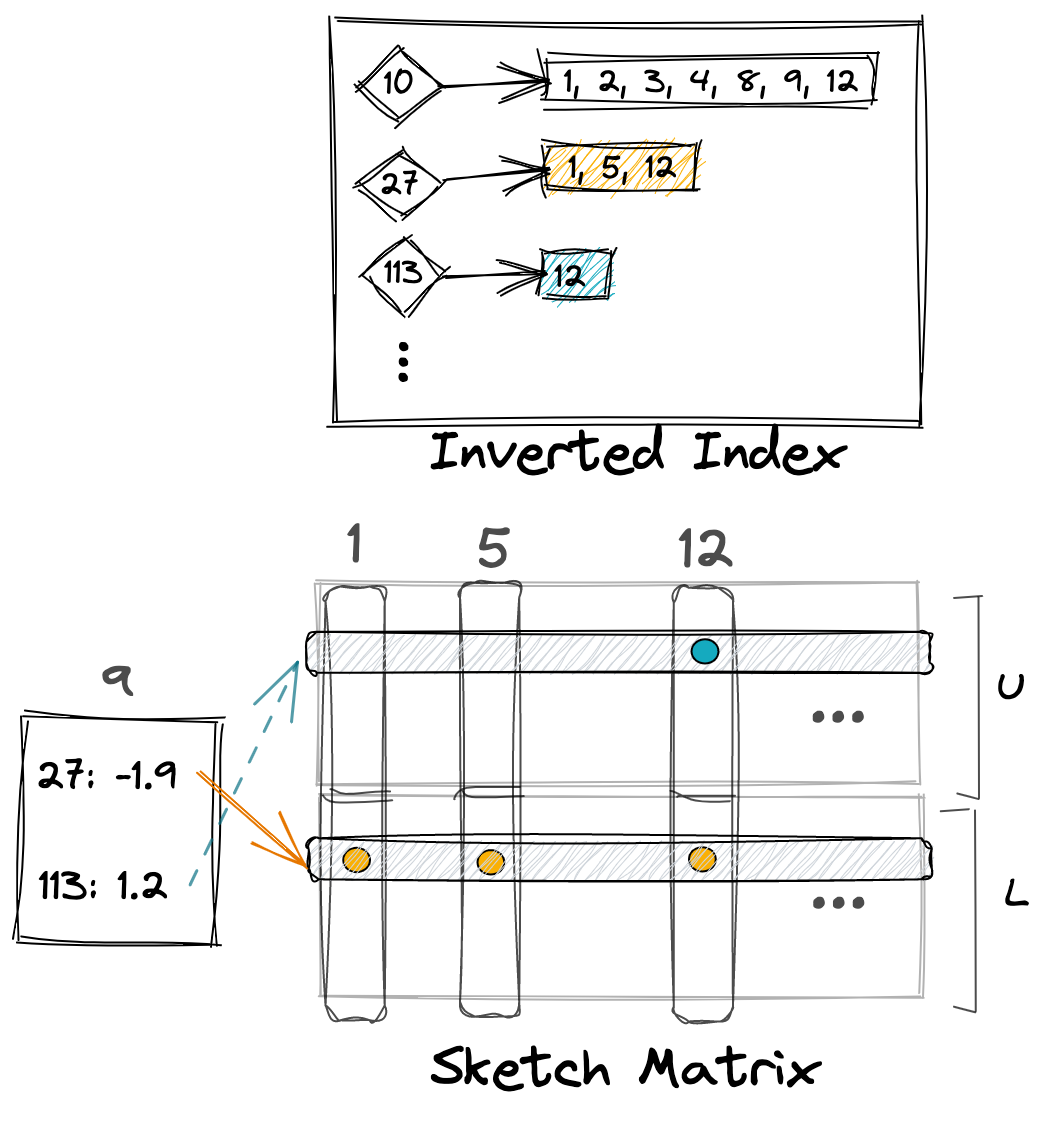}\label{figure:example:scoring}}
}
\caption{Example of (a) indexing and (b) score computation in \sinnamon{}. When inserting the vector $x_{12} \in \mathbb{R}^n$ consisting of $3$ non-zero coordinates $\{ 10, 27, 113\}$, we first populate the inverted index and then insert a sketch of $x$ into the $12$th column of the sketch matrix. The top half of the matrix, $\mathcal{U}$, records the upper-bounds and the bottom half,$\mathcal{L}$, the lower-bounds, with the help of a single random mapping $\pi$ from $n$ to $m$. When computing the approximate inner product of a query vector $q$ with the documents in the collection, we look up the inverted list for one coordinate and traverse its corresponding row in the sketch matrix to accumulate partial scores in a coordinate-at-a-time algorithm.}
\label{figure:example}
\end{center}
\end{figure}

It is instructive to contrast the index structure in \sinnamon{} with the one in \linscan{}. 
The main material difference between the two indexes is that \sinnamon{} allocates a constant amount of memory to store the sketch of each document, whereas \textsc{LinScan} stores all non-zero values of a document within the index. As a result, the amount of memory required to store values with \linscan{} grows linearly in $\psi_d$.

Finally, let us consider the time complexity of the insertion procedure in \sinnamon{}. Let us assume that raw vectors are represented in a ``sparse'' format: rather than a vector being implemented as an array of size $n$ with most entries $0$, we assume a vector is a mapping from a coordinate to a value. Assuming further that inserting a value into an inverted list using \textsc{Insert}$(\cdot)$ can be carried out in constant time, the overall time complexity of Algorithm~\ref{algorithm:indexing} is then $\mathcal{O}(\psi_d h)$ for a vector $x$ with $\psi_d$ non-zero entries. The algorithm stores $\psi_d$ integers in the inverted index and $2m$ real values per vector.

As an important note we add that, when operating in the non-negative orthant where vectors are in $\mathbb{R}_{+}^{n}$, we do not record the lower-bounds (i.e., the sketch $\mathcal{L}$) and only maintain a sketch matrix $\tilde{\mathcal{X}} = \mathcal{U} \in \mathbb{R}^{m \times |\mathcal{X}|}$. We call this variant of the algorithm \sinnamonplus{}.

\subsubsection{Notes on Implementation}
In practice, the inverted index can be compressed using any of the existing bitmap or list compression codes such as Roaring~\cite{roaring}, Simple16~\cite{simple16}, PForDelta~\cite{pfordelta} or others~\cite{pibiri2020compression,wang2017compression}. We use the Roaring codec throughout this work to compress inverted lists because it achieves a reasonable compression ratio and, at the same time, supports very fast insertion and deletion operations.

As with \roaringlinscan{}, we store sketches using \emph{bfloat16}. Finally, we note that the matrix $\tilde{\mathcal{X}}$ is in practice implemented as an array of $2m$ rows, with each row growing as needed to accommodate the sketch of the $i$th vector. This particular data layout is more efficient during retrieval because \sinnamon{} needs access to the same row in order to compute partial scores for documents within a single inverted list. As such, by representing a row as a contiguous region of memory, we improve the overall memory access pattern and make the data structure more cache-friendly.

\subsubsection{Extensions and Future Considerations}
One of the desired properties in our design is that the inverted index $\mathcal{I}$ must offer compression as well as efficient insertion and deletion operations. This is required because the efficiency of the inverted index directly affects the overall efficiency of the indexing, update, and retrieval algorithms in \sinnamon{}. For this study, we settled on Roaring~\cite{roaring} bitmaps as it suits the needs of the algorithm. However, we note that studying inverted indexes with the outlined properties is an orthogonal area of research with many existing studies on the required operations.

Having stated that, existing inverted indexes are deterministic and exact. 
\sinnamon{}, on the other hand, is an approximate algorithm which trades size and speed for error tolerance. 
\sinnamon{} also offers levers, as we will argue, to compensate for the incurred error. 
This enables us to explore \emph{approximate} inverted indexes where each inverted list may be a \emph{superset} rather than an exact set and where multiple entries may share an inverted list. In other words, the research question to investigate is whether there exists an approximate inverted index $\tilde{\mathcal{I}}$ such that $\mathcal{I}[j] \subset \tilde{\mathcal{I}}[j] \;\forall\; j$ and where $|\tilde{\mathcal{I}}| \ll |\mathcal{I}|$ (with $|\cdot|$ denoting the overall size of the index) with a quantifiable effect on the final inner product with arbitrary query vectors. We wish to investigate this question in a future study.

Another research question in this context is the representation of the values. While \sinnamon{} offers a fixed number of dimensions in the sketch, how those entries are represented affects the overall memory usage. We believe that a number of quantization methods can be used to reduce the capacity requirement while maintaining an approximately accurate inner product. This is another area we wish to explore in the future.

\subsection{Retrieval}
\label{section:algorithm:retrieval}

Assuming we have an inverted index $\mathcal{I}$ and a sketch matrix $\tilde{\mathcal{X}}$ as described in Section~\ref{section:algorithm:indexing}, as well as $h$ random mappings $\pi_o(\cdot)_{1 \leq o \leq h}$, we now discuss the question of retrieval: Given a query vector $q \in \mathbb{R}^n$, find the top-$k$ closest vectors in the collection.

\subsubsection{Scoring}
Similar to \linscan{}, \sinnamon{} approaches retrieval in two steps. In the first and most critical step, \sinnamon{} computes an upper-bound on the inner product of the query vector and every vector in the collection. It does so by traversing the inverted list of every non-zero coordinate in the query, one coordinate at a time, and computing and accumulating partial scores (i.e., the product of query value at that coordinate and the document value as encoded in the sketch). Importantly, \sinnamon{} visits non-zero query coordinates in order from the coordinate with the largest absolute value to the smallest one to facilitate an anytime variant. As soon as a given time-budget $T$ is exhausted, it terminates the \emph{scoring} phase. At this point, all partial scores are upper-bounds on the exact inner product of the processed coordinates. This means, for example, when $T=\infty$ (i.e., when time is unlimited) the computed score of a document vector is an upper-bound on its inner product with the query.

\begin{algorithm}[t]
\SetAlgoLined
{\bf Input: }{Sparse query vector $q \in \mathbb{R}^n$; Time budget, $T$.}\\
{\bf Requirements: }{Inverted index $\mathcal{I}$; Sketch matrix $\tilde{\mathcal{X}} = [\mathcal{U};\; \mathcal{L}] \in \mathbb{R}^{2m \times |\mathcal{X}|}$; $h$ independent random mappings $\pi_o: [n] \rightarrow [m]$.}\\
\KwResult{Approximate scores for all $x_i \in \mathcal{X}$}

\begin{algorithmic}[1]
    \STATE $nz(q) \leftarrow \{ j \;|\; q[j] \neq 0 \}$
    \STATE \textbf{Sort} coordinates in $nz(q)$ in descending order by the absolute value of the vector entries ($|q[j]|$)
    \STATE $\textit{scores}[i] \leftarrow 0 \;\forall i$
    \FORALL{$j \in \textbf{sorted } nz(q)$}
        \STATE $K \leftarrow \{ \pi_o(j) \; 1 \leq o \leq h \}$ \label{algorithm:retrieval:rows}
        \FORALL{$i \in \mathcal{I}[j]$} \label{algorithm:retrieval:inverted_list_traversal}
            \IF {$q[j] > 0$} \label{algorithm:retrieval:branching}
                \STATE $s \leftarrow q[j] \times \min_{k \in K} \mathcal{U}[k][i]$ \label{algorithm:retrieval:least_upperbound}
            \ELSE
                \STATE $s \leftarrow q[j] \times \max_{k \in K} \mathcal{L}[k][i]$ \label{algorithm:retrieval:greatest_lowerbound}
            \ENDIF
            \STATE $\textit{scores}[i] \leftarrow \textit{scores}[i] + s$
        \ENDFOR
        \STATE \textbf{break} if time budget $T$ exhausted
    \ENDFOR
    \RETURN $\textit{scores}$
 \end{algorithmic}
 \caption{Scoring in \sinnamon{}$_{(2m, \; h)}$}
\label{algorithm:retrieval}
\end{algorithm}

Algorithm~\ref{algorithm:retrieval} presents the scoring procedure in \sinnamon{}. Intuitively, when the sign of a query entry at coordinate $j$ is positive, we find the least upper-bound on the value of $x_i[j]$ for a document $x_i \in \mathcal{I}[j]$ (line~\ref{algorithm:retrieval:least_upperbound} in Algorithm~\ref{algorithm:retrieval}). When $q[j] < 0$, we find the greatest lower-bound on $x_i[j]$ (line~\ref{algorithm:retrieval:greatest_lowerbound}). In this way, \sinnamon{} guarantees that the partial score is always an upper-bound on the actual partial score. This is illustrated for an example query vector in Figure~\subref*{figure:example:scoring}.

We note that, the expected time complexity of the scoring algorithm is $\mathcal{O}(\psi_q \log \psi_q + \psi_q h |\mathcal{X}| \psi_d/n)$, typically dominated by the second term, where the term $|\mathcal{X}| \psi_d/n$ represents the expected number of vectors that have a non-zero value in a particular coordinate (with the assumption that non-zero coordinates are uniformly distributed).

\subsubsection{Ranking}
At the end of the scoring stage, \sinnamon{} gives us approximate scores for every document in the collection. In the second stage, which we refer to as ``ranking,'' we must find the top-$k$ vectors that make up the (approximate) solution to Equation~(\ref{equation:mips}).

\begin{algorithm}[t]
\SetAlgoLined
{\bf Input: }{Sparse query vector $q \in \mathbb{R}^n$; Number of vectors to re-rank, $k^\prime$; Number of vectors to return, $k$.}\\
{\bf Requirements: }{Vector storage $\mathcal{S}$; Scores from Algorithm~\ref{algorithm:retrieval}, $\mathit{scores}$.}\\
\KwResult{Approximate solution set of Equation~(\ref{equation:mips})}

\begin{algorithmic}[1]
    \STATE $V \leftarrow \textsc{FindLargest}(\mathit{scores}, k^\prime)$
    \FOR {$v \in V$}
    \item $s_{v} \leftarrow \langle q,\; \mathcal{S}.\textsc{Fetch}(v) \rangle$
    \ENDFOR
    \RETURN $\textsc{FindLargest}(\{ s_v \;|\; v \in V \}, k)$
 \end{algorithmic}
 \caption{Ranking in \sinnamon{}$_{(2m, \; h)}$}
\label{algorithm:retrieval:ranking}
\end{algorithm}

To do that, we first find the $k^\prime \geq k$ vectors with the largest approximate score using a heap with time complexity $\mathcal{O}(|\mathcal{X}| \log k^\prime)$. The reason the initial pass selects a set that has more than $k$ vectors is that by doing so we compensate for the approximation error of the scoring algorithm. We later review the relationship between $k$ and $k^\prime$ empirically. We subsequently execute Algorithm~\ref{algorithm:retrieval:ranking} to compute the exact inner product between the query and the set of $k^\prime$ vectors, and eventually return the top-$k$ subset according to the exact scores.

Much like \linscan{}, \sinnamon{}'s retrieval algorithm (i.e., Algorithms~\ref{algorithm:retrieval} and~\ref{algorithm:retrieval:ranking}) is trivially amenable to dynamic parallelism. By breaking up an inverted list into non-overlapping segments in line~\ref{algorithm:retrieval:inverted_list_traversal} of Algorithm~\ref{algorithm:retrieval}, we can accumulate partial scores concurrently without the need to break up or shard the sketch matrix. In the parallel version of the algorithm, we also execute $\textsc{FindLargest}(\cdot)$ and the exact computation of inner products in Algorithm~\ref{algorithm:retrieval:ranking} using multiple threads. We refer to the parallel variant of the algorithm as $\sinnamon{}^\parallel$.

We note that, while a mono-CPU variant of \sinnamon{} offers a consistent analysis of the trade-offs within the algorithm and sheds light onto its behavior in comparison with other algorithms, we believe that the ease by which \sinnamon{} (and \linscan{}) can be run concurrently on a per-query basis renders the algorithm suitable for production systems that operate on large (dynamic) collections. In particular, because the index structure remains monolithic within each machine, it need not be rebuilt or re-assembled when porting the index to another machine with a different configuration.

\subsubsection{Notes on Implementation}
As is clear from line~\ref{algorithm:retrieval:rows} in Algorithm~\ref{algorithm:retrieval}, for a query coordinate $j$, we need only probe a fixed set of $h$ rows in the sketch matrix. When $h=1$, as a typical example, this implies that we need only visit a single row which is stored as a contiguous array in memory. Due to this property and the predictability of the memory access pattern, it is often possible to cache a few upcoming memory locations in advance of the computation so as to speed up scoring. We observe the effect of the cache-friendliness of the sketch matrix in practice by enabling default compiler optimizations. We also note that, it is possible to further optimize the implementation through explicit instruction-level parallelism where the compiler fails to do so itself, though we do not explicitly use this technique in this work.

In our implementation, we further optimize the efficiency of the algorithm by re-arranging the logic so as to remove the branching on line~\ref{algorithm:retrieval:branching} in Algorithm~\ref{algorithm:retrieval}. This is possible if the programming language offers ``function pointers,'' to point to the $\min$ and $\max$ operators depending on the sign of the query entry.

\subsubsection{Extensions and Future Considerations}
A requirement that is often faced in practice is for a retrieval algorithm to support constrained search. For example, a user may only ask for the top-$k$ set of songs whose genre matches a set of desired genres. One way to formalize this is to require that the retrieval algorithm enforce arbitrarily many binary constraints on the solution space. In other words, the vectors in the solution set of Equation~(\ref{equation:mips}) must pass an arbitrary set of functions $\mathcal{G} = \{ g_i: \mathcal{X} \rightarrow \{0, 1\} \}$. This transforms the problem to the following constrained retrieval problem:
\begin{align}
\begin{split}
    \argmax^{(k)}_{x \in \mathcal{X}} \; &\langle q \; , \; x \rangle, \\
    \textit{s.t.} \;\; \wedge_{g \in \mathcal{G}} \; &g(x) = 1,
    \label{equation:sinnamon-problem}
\end{split}
\end{align}
where $\wedge$ denotes the \textsc{And} operator.

\sinnamon{} naturally supports this mode of search because, by default, it computes the scores of all documents in the collection in its scoring stage---the same is true of \linscan{}. It is therefore possible to enforce arbitrary constraints by masking out those columns in $\tilde{\mathcal{X}}$ that do not satisfy the given conditions. We defer an examination of this setup to future studies.

\subsection{Deletions}

When deleting a vector $x_i$ from \linscan{}'s index, we committed to a ``full deletion,'' wiping all postings associated with $x_i$ from inverted lists. That strategy, we argued, fits \linscan{} well as it simplifies its insertion and retrieval logic. \sinnamon{}, in contrast, provides us with an alternative deletion mechanism.

When deleting $x_i$ from the collection, we simply remove all instances of $i$ from inverted lists in $\mathcal{I}$, much like in \linscan{}. However, we do not clean up the sketch of $x_i$ from the sketch matrix $\tilde{\mathcal{X}}$. Instead, we add $i$ to the set of available document identifiers so that the next vector that is inserted into the index may reuse $i$ as its identifier and recycle column $i$ in $\tilde{\mathcal{X}}$ to store its sketch.

This protocol is efficient because it only involves the removal of an integer from compressed inverted lists---an instruction that is often very fast to execute. Contrast this with \linscan{} where the two parallel arrays within an inverted list must remain aligned at all times. Because the values array must be cut at the same spot as the array that holds document identifiers, we have the overhead of having to find the index of the posting that holds $i$, then proceed to delete the $i$th entry in the values array.

We do note that, for applications that receive far more delete requests than new insertions, \sinnamon{}'s deletion logic may prove suboptimal. This is because, by virtue of not freeing up a column upon deletion, the sketch matrix could grow to be as large as the maximal number of vectors that exist in the dataset at the same time. For such applications, a different deletion strategy may be required that may involve defragmenting the matrix and reclaiming the underlying space. In practice, however, we find this particular scenario to be a mere hypothetical; in reality, deletions are dwarfed by insertion and update requests, where \sinnamon{}'s default deletion strategy leads to minimal to no waste.

\section{Analysis}
\label{section:analysis}

Recall that \sinnamon{} uses a sketch of size $2m$ to record upper- and lower-bounds on the values of active\footnote{In the rest of this work we refer to coordinates in a sparse vector as either zero or non-zero. In this section, to make the exposition more accurate, we adopt a more formal terminology and say a coordinate is \emph{inactive} when it is not present in the sparse vector, and \emph{active} when it is. Note that, the value of an active coordinate is \emph{almost surely} non-zero; that leaves room for the unlikely event that it may draw the value $0$ from its value distribution.} coordinates in a vector. Consider now Line~\ref{algorithm:indexing:upper-bound-sketch} of Algorithm~\ref{algorithm:indexing} where, given a vector $x \in \mathbb{R}^n$ and $h$ independent random mappings $\pi_o: [n] \rightarrow [m]$ ($1 \leq o \leq h$), we construct the upper-bound sketch $u \in \mathbb{R}^m$ where the $k$th dimension is assigned the following value:
\begin{equation}
    u[k] \leftarrow \max_{\{ i \in nz(x) \;|\; \exists \;o\; \mathit{s.t.}\; \pi_o(i) = k \}} x[i].
\end{equation}
The lower-bound sketch is filled in a symmetric manner, in the sense that the algorithmic procedure is the same but the operator changes from $\max(\cdot)$ to $\min(\cdot)$.

When a query coordinate is positive, to reconstruct the value $x[i]$ of a document vector, we take the least upper-bound from $u$, as captured on Line~\ref{algorithm:retrieval:least_upperbound} of Algorithm~\ref{algorithm:retrieval}, restated below for convenience:
\begin{equation}
    \tilde{x}[i] \leftarrow \min_{k \in \{ \pi_o(i) \; 1 \leq o \leq h\}} u[k].
\end{equation}
When the query coordinate is negative, on the other hand, it is the greatest lower-bound that is returned instead.

Given the above, it is easy to see that the following proposition is always true:

\begin{theorem}
    The score returned by Algorithm~\ref{algorithm:retrieval} of \sinnamon{}$_{T=\infty}$ is an upper-bound on the inner product of query and document vectors.
\end{theorem}

The fact above implies that \sinnamon{}'s approximation error is always non-negative. But what can we say about the probability that such an error occurs when approximating a document value? How large is the overestimation error of a single value? How does that error affect the final score of a query-document pair? These are some of the questions we examine in the remainder of this section.

\subsection{Notation and Probabilistic Model}
Denote by $X \in \mathbb{R}^n$ a random sparse vector that is constructed as follows:
The $i$th coordinate of $X$ is \emph{inactive} with probability $1 - p_i$.
Otherwise, it is \emph{active} and its value is a random variable, $X_i$,
drawn from some distribution with probability density function (PDF)
$\phi$ and cumulative distribution function (CDF) $\Phi$.
We assume throughout this work that $X_i$s are independent.

We have that the active coordinates of $X$
are encoded into the upper-bound and lower-bound sketches of \sinnamon{}, $\mathcal{U}$ and $\mathcal{L}$.
Denote by $U_k$ and $L_k$ random variables that represent the $k$th coordinate of
$\mathcal{U}$ and $\mathcal{L}$, where $1 \leq k \leq m$.

For every active variable $X_i$, we can obtain an upper-bound on its value from $\mathcal{U}$ and a lower-bound from $\mathcal{L}$. We denote these decoded upper- and lower-bounds by $\overline{X}_i$ and $\underline{X}_i$, respectively. Clearly, $\underline{X}_i \leq X_i \leq \overline{X}_i$ always holds. When our statement is agnostic to whether a decoded value is an upper-bound or a lower-bound on $X_i$, we simply write $\tilde{X}_i$ for the decoded value.

\subsection{Sketching Error}
In this section, we focus on the approximation error of a single active document coordinate:
What is the probability and expected magnitude of error when recovering the value of a single
active coordinate from a document sketch? More formally, we are interested in finding
$\mathbb{P}[\tilde{X}_i \neq X_i]$ and the distribution of error in the form of its
CDF: $\mathbb{P}[|\tilde{X}_i - X_i | < \delta]$, when $X_i$ is active.

We state the following result on the probability of error of the upper-bound sketch.

\begin{theorem}[Probability of Error of the Upper-Bound Sketch] 
    For large values of $m$ and an active $X_i$,
    \begin{equation}
        \mathbb{P}\big[ \overline{X}_i > X_i \big] \approx \int {\big[ 1 - e^{-\frac{h}{m} (1 - \Phi(\alpha)) \sum_{j \neq i} p_j } \big]^h \phi(\alpha) d\alpha},
        \label{equation:analysis:prob-of-error}
    \end{equation}
    where $\phi(\cdot)$ and $\Phi(\cdot)$ are the PDF and CDF of $X_i$s.
\end{theorem}

Extending this result to the lower-bound sketch involves replacing $1 - \Phi(\alpha)$ with $\Phi(\alpha)$.
When the distribution defined by $\phi$ is symmetric, the probabilities of error
too are symmetric for the upper-bound and lower-bound sketches.

\begin{proof}
Consider the random value $X_i$. Suppose $k = \pi_o(i)$ for some $1 \leq o \leq h$. We thus need to probe $U_k$ as part of producing $\overline{X}_i$. The event that $U_k > X_i$ happens only when there exists another active coordinate $X_j$ such that $X_j > X_i$ and $\pi_o(i)=\pi_o(j)=k$. Consider the probability $\mathbb{P}[U_k > X_i]$.

To derive that probability, it is easier to think in terms of complementary events: $U_k = X_i$ if every other active coordinate whose value is larger than $X_i$ maps to a sketch coordinate except $k$. Clearly the probability that any arbitrary $X_j$ maps to a sketch coordinate other than $U_k$ is simply $1 - 1/m$. Therefore, given a vector $X$, the probability that no active coordinate $X_j$ larger than $X_i$ maps to $U_k$ is:
\begin{equation}
    \mathbb{P}\big[ \underbrace{\nexists j \neq i \;s.t.\; \pi_o(j) = \pi_o(i) = k \;\textit{for some}\; 1 \leq o \leq h}_{\text{Event A}}  \;|\; X \big] = 1 - (1 - \frac{1}{m})^{h\sum_{j \neq i} \mathbbm{1}_{X_j \text{ active}} \mathbbm{1}_{X_j > X_i}}.
\end{equation}
Because $m$ is large by assumption, we can approximate $e^{-1} \approx (1 - 1/m)^m$ and rewrite the expression above as follows:
\begin{equation}
    \mathbb{P}\big[\text{Event A}  \;|\; X \big] \approx 1 - e^{-\frac{h}{m}\sum_{j \neq i} \mathbbm{1}_{X_j \text{ active}} \mathbbm{1}_{X_j > X_i}}.
\end{equation}
Finally, we marginalized the expression above over $X_j$s for $j \neq i$ to remove the dependence on all but the $i$th coordinate of $X$.
To simplify the expression, however, we take the expectation over the first-order Taylor expansion of the right hand side
around $0$. This results in the following approximation:
\begin{equation}
    \mathbb{P}\big[\text{Event A}  \;|\; X_i \big] \approx 1 - e^{-\frac{h}{m} (1 - \Phi(\alpha)) \sum_{j \neq i} p_j}.
\end{equation}

For $\overline{X}_i$ to be larger than $X_i$, event $A$ must take place for all $h$ sketch coordinates corresponding to $i$.
That probability, by the independence of random mappings, is:
\begin{equation}
    \mathbb{P}\big[ \overline{X}_i > X_i \;|\; X_i \big] \approx \big[ 1 - e^{-\frac{h}{m} (1 - \Phi(\alpha)) \sum_{j \neq i} p_j } \big]^h.
\end{equation}
In deriving the expression above, we conditioned the event on the value of $X_i$.
Taking the marginal probability leads us to the following expression for the event that $\overline{X}_i > X_i$ for any $i$:
\begin{equation}
    \mathbb{P}\big[ \overline{X}_i > X_i \big] \approx \int {\big[ 1 - e^{-\frac{h}{m} (1 - \Phi(\alpha)) \sum_{j \neq i} p_j} \big]^h d\mathbb{P}(\alpha)} \approx \int {\big[ 1 - e^{-\frac{h}{m} (1 - \Phi(\alpha)) \sum_{j \neq i} p_j} \big]^h \phi(\alpha) d\alpha}.
\end{equation}
That concludes our proof.
\end{proof}

Equation~(\ref{equation:analysis:prob-of-error}) is unwieldy in an analytical sense, but can be computed numerically.
Nonetheless, it offers insights into the behavior of the upper-bound sketch. Our first observation is that the sketching mechanism presented here is more suitable for distributions where larger values occur with a smaller probability such as sub-Gaussian variables. In such cases, the larger the value is the smaller its chance of being overestimated by the upper-bound sketch. Regardless of the underlying distribution, empirically, the largest value in a vector is always recovered \emph{exactly}.

The second insight is that there is a sweet spot for $h$ given a particular value of $m$: using more random mappings helps lower the probability of error until the sketch starts to saturate, at which point the error rate increases. This particular property is similar to the behavior of a Bloom filter.

In addition to these general observations, it is often possible to derive a closed form expression for special distributions and obtain further insights. For example, when active $X_i$s are drawn from a zero-mean, unit-variance Gaussian distribution, the probability of error can be expressed as in the following corollary.

\begin{corollary}
    Suppose the probability that a coordinate is active, $p_i$, is equal to $p$ for all coordinates of vector $X \in \mathbb{R}^n$.
    When active $X_i$s are drawn from $Gaussian(\mu=0, \sigma=1)$, the probability of error is:
    \begin{equation}
    \mathbb{P}\big[ \overline{X}_i > X_i \big] \approx 1 + \sum_{k=1}^{h} {h \choose k} (-1)^k \frac{m}{kh(n-1)p} \left(1 - e^{-\frac{kh(n-1)p}{m}} \right).
    \label{equation:analysis:prob-error-gaussian}
\end{equation}
\end{corollary}
We provide a complete proof in Appendix~\ref{appendix:proof:prob-error-gaussian}. By expanding the above expression, it becomes clear, for example, that when vectors are Gaussian and $m$ is less than half the average number of active coordinates, using more than one random mapping leads to an increase in the probability of error.

\begin{table*}[t]
\caption{Probability of error as expressed in Equation~(\ref{equation:analysis:prob-of-error}) for sample distributions. In all cases, $\psi_d = n p_i = np = 120$ and the distribution over the support is as indicated in the first column. We select the support arbitrarily for the purposes of this demonstration. For the Zeta distributions, we quantize the interval $[-1, 1]$ into $2^{10}$ discrete values and define the distribution over the discrete support.}
\label{table:analysis:prob-of-error:discrete}
\begin{center}
\begin{sc}
\begin{tabular}{c|ccc|ccc|ccc}
& \multicolumn{3}{c}{$m=\psi_d/2$} & \multicolumn{3}{c}{$m=\psi_d$} & \multicolumn{3}{c}{$m=2 \psi_d$} \\
\toprule
$\phi$ & $h=1$ & $h=2$ & $h=3$ & $h=1$ & $h=2$ & $h=3$ & $h=1$ & $h=2$ & $h=3$ \\
\midrule
$\text{Uniform}_{[-1, 1]}$ & 0.57 & 0.63 & 0.69 & 0.37 & 0.38 & 0.43 & 0.21 & 0.17 & 0.17 \\
$\text{Zeta}(s=2.5)$ & 0.34 & 0.33 & 0.38 & 0.19 & 0.12 & 0.11 & 0.10 & 0.04 & 0.02 \\
$\text{Zeta}(s=4.0)$ & 0.13 & 0.03 & 0.04 & 0.07 & 0.02 & 0.008 & 0.03 & 0.005 & 0.001 \\
$\text{Gaussian}(\mu=0, \sigma=1)$ & 0.57 & 0.63 & 0.69 & 0.37 & 0.38 & 0.43 & 0.21 & 0.17 & 0.17 \\
$\text{Gaussian}(\mu=0, \sigma=0.1)$ & 0.57 & 0.63 & 0.69 & 0.37 & 0.38 & 0.43 & 0.21 & 0.17 & 0.17 \\
\bottomrule
\end{tabular}
\end{sc}
\end{center}
\end{table*}

Let us attempt to develop an empirical understanding of the error and verify the observations above. 
We consider special distributions for which we can either solve Equation~(\ref{equation:analysis:prob-of-error}) exactly or approximate it numerically. Table~\ref{table:analysis:prob-of-error:discrete} shows this probability for five different distributions. We observe that when $m < np$, for most distributions, using more than one random mapping leads to an increase in the likelihood of error. This is because a larger $h$ leads to the sketch saturating quickly. When $m$ is larger, however, using more mappings often translates to better accuracy up to a point.

What is strange at first glance, however, is that the probability that a value is overestimated is approximately the same for uniform and Gaussian distributions, which appears to contradict our statement that \sinnamon{}'s sketching is more suitable for sub-Gaussian rather than uniform variables. But to understand the differences between these distributions, we must contextualize the probability of error with the distribution of error and understand its concentration around zero.

\begin{table*}[t]
\caption{Expected value of error for sample distributions. The setup is the same as in Table~\ref{table:analysis:prob-of-error:discrete}.}
\label{table:analysis:expected_error:discrete}
\begin{center}
\begin{sc}
\begin{tabular}{c|ccc|ccc|ccc}
& \multicolumn{3}{c}{$m=\psi_d/2$} & \multicolumn{3}{c}{$m=\psi_d$} & \multicolumn{3}{c}{$m=2 \psi_d$} \\
\toprule
$\phi$ & $h=1$ & $h=2$ & $h=3$ & $h=1$ & $h=2$ & $h=3$ & $h=1$ & $h=2$ & $h=3$ \\
\midrule
$\text{Uniform}$ & 0.43 & 0.46 & 0.52 & 0.26 & 0.24 & 0.27 & 0.15 & 0.09 & 0.09 \\
$\text{Zeta}(s=2.5)$ & 0.001 & 5$e$-4 & 5$e$-4 & 7$e$-4 & 2$e$-4 & 2$e$-4 & 4$e$-4 & 6$e$-5 & 3$e$-5 \\
$\text{Gaussian}(\mu=0, \sigma=1)$ & 0.40 & 0.43 & 0.48 & 0.24 & 0.22 & 0.25 & 0.14 & 0.09 & 0.08 \\
$\text{Gaussian}(\mu=0, \sigma=0.1)$ & 0.07 & 0.07 & 0.08 & 0.05 & 0.04 & 0.04 & 0.02 & 0.01 & 0.01 \\
\bottomrule
\end{tabular}
\end{sc}
\end{center}
\end{table*}

To that end, we state the following result for the upper-bound sketch. We denote by $\overline{Z}_i$ the decoding error $\overline{X}_i - X_i$.

\begin{theorem}[CDF of Error in the Upper-Bound Sketch]
    For an active $X_i$ whose values are drawn from a distribution with PDF and CDF $\phi$ and $\Phi$,
    the probability that $\overline{Z}_i = \overline{X}_i - X_i$ is less than $\delta$ is:
    \begin{equation}
        \mathbb{P}[\overline{Z}_i \leq \delta; \; X_i \text{ is active}] \approx 1 - \int \big[ 1 - e^{-\frac{h}{m} (1 - \Phi(\alpha + \delta)) \sum_{j \neq i} p_j} \big]^h \phi(\alpha) d \alpha.
        \label{equation:analysis:dist-error}
    \end{equation}
    \label{theorem:cdf-upperbound-sketch}
\end{theorem}
\begin{proof}
We begin by quantifying the conditional probability $\mathbb{P}[\overline{Z}_i \leq \delta \;|\; X_i]$. Conceptually, the event $\overline{Z}_i \leq \delta$ for a given $X_i$ happens when all values that collide with $X_i$ are less than or equal to $X_i + \delta$. This event can be characterized as the complement of the event that all $h$ sketch coordinates that contain $X_i$ collide with values greater than $X_i + \delta$. Using this complementary event, we can write the conditional probability as follows:
\begin{equation}
    \mathbb{P}[\overline{Z}_i \leq \delta \;|\; X_i] = 1 - \big[ 1 - (1 - \frac{1}{m})^{h (1 - \Phi(\alpha + \delta)) \sum_{j \neq i} p_j} \big]^h \approx 1 - \big[ 1 - e^{-\frac{h}{m} (1 - \Phi(\alpha + \delta)) \sum_{j \neq i} p_j} \big]^h.
\end{equation}
By computing the marginal distribution over the support, we arrive at the following expression for the CDF of $\overline{Z}_i$:
\begin{equation}
    \mathbb{P}[\overline{Z}_i \leq \delta; \; X_i \text{ is active}] \approx 1 - \int \big[ 1 - e^{-\frac{h}{m} (1 - \Phi(\alpha + \delta)) \sum_{j \neq i} p_j} \big]^h d \mathbb{P}(\alpha),
\end{equation}
which completes the proof.
\end{proof}

Given the CDF of $\overline{Z}_i$ and the fact that $\overline{Z}_i \geq 0$, it follows that its expected value conditioned on $X_i$ being active is:

\begin{lemma}[Expected Value of Error in the Upper-Bound Sketch]
\begin{equation}
    \mathbb{E}[\overline{Z}_i; \; X_i \text{ is active}] = \int_{0}^{\infty} \mathbb{P}[\overline{Z}_i \geq \delta] d\delta \approx \int_{0}^{\infty} \int \big[ 1 - e^{-\frac{h}{m} (1 - \Phi(\alpha + \delta)) \sum_{j \neq i} p_j} \big]^h \phi(\alpha)\; d\alpha \; d\delta.
\end{equation}
\end{lemma}

We present the CDF of $\overline{Z}_i$ in Figure~\ref{figure:analysis:cdf-error} for uniform and Gaussian-distributed vectors, and report the expected value for other distributions in Table~\ref{table:analysis:expected_error:discrete}. Examining Tables~\ref{table:analysis:prob-of-error:discrete} and~\ref{table:analysis:expected_error:discrete} together shows that, while some distributions can result in a similar probability of error given the same sketch configuration, they differ greatly in terms of the expected magnitude of error.

We also find it interesting to study Theorem~\ref{theorem:cdf-upperbound-sketch} for special distributions. As we show in Appendix~\ref{appendix:proof:dist-error-gaussian}, for example, when vectors are drawn from a zero-mean Gaussian distribution with standard deviation $\sigma$, the CDF of the variable $\overline{Z}_i$ can be derived as follows.

\begin{corollary}
    Suppose the probability that a coordinate is active, $p_i$, is equal to $p$ for all coordinates of vector $X \in \mathbb{R}^n$.
    When active $X_i$s are drawn from $Gaussian(\mu=0, \sigma)$, the CDF of error is:
    \begin{equation}
        \mathbb{P}[\overline{Z}_i \leq \delta; \; X_i \text{ is active}] \approx 1 - \big[ 1 - e^{-\frac{h (n - 1) p}{m} (1 - \Phi^\prime(\delta)) } \big]^h,
        \label{equation:analysis:dist-error-gaussian}
    \end{equation}
    where $\Phi^\prime(\cdot)$ is the CDF of a zero-mean Gaussian with standard deviation $\sigma\sqrt{2}$.
    \label{corollary:analysis:dist-error-gaussian}
\end{corollary}

This expression enables us to find a particular sketch configuration given a desired bound on the probability of error. It is straightforward, for instance, to show the following result.

\begin{lemma}
    Suppose the probability that a coordinate is active, $p_i$, is equal to $p$ for all coordinates of vector $X \in \mathbb{R}^n$.
    Suppose active $X_i \sim Gaussian(0, \sigma)$. Given a choice of $0 < \delta, \epsilon < 1$ and the number of random mappings $h$, we have that $\mathbb{P}[\overline{Z}_i > \delta] < \epsilon$ when:
    \begin{equation}
        m > - \frac{h (n - 1)p (1 - \Phi^\prime(\delta))}{\log (1 - \epsilon^{1/h})}.
    \end{equation}
    \label{lemma:analysis:m-and-h}
\end{lemma}

To get a sense of what this expression entails, we have plotted $m/(n - 1)p$ as a function of $h$ given particular configurations of $\sigma$, $\delta$, and $\epsilon$ in Figure~\ref{figure:analysis:m-and-h}. As one may expect, when we utilize more random mappings to form the sketch, we require a smaller sketch size to maintain the same guarantee on the concentration of error. That is true only up to a certain point: using more than three mappings, for example, translates to an increase in $m$ to keep the magnitude of error within a given bound.

\begin{figure}[t]
\begin{center}
\centerline{
\subfloat[$\sigma=0.1$]{
\includegraphics[width=0.37\linewidth]{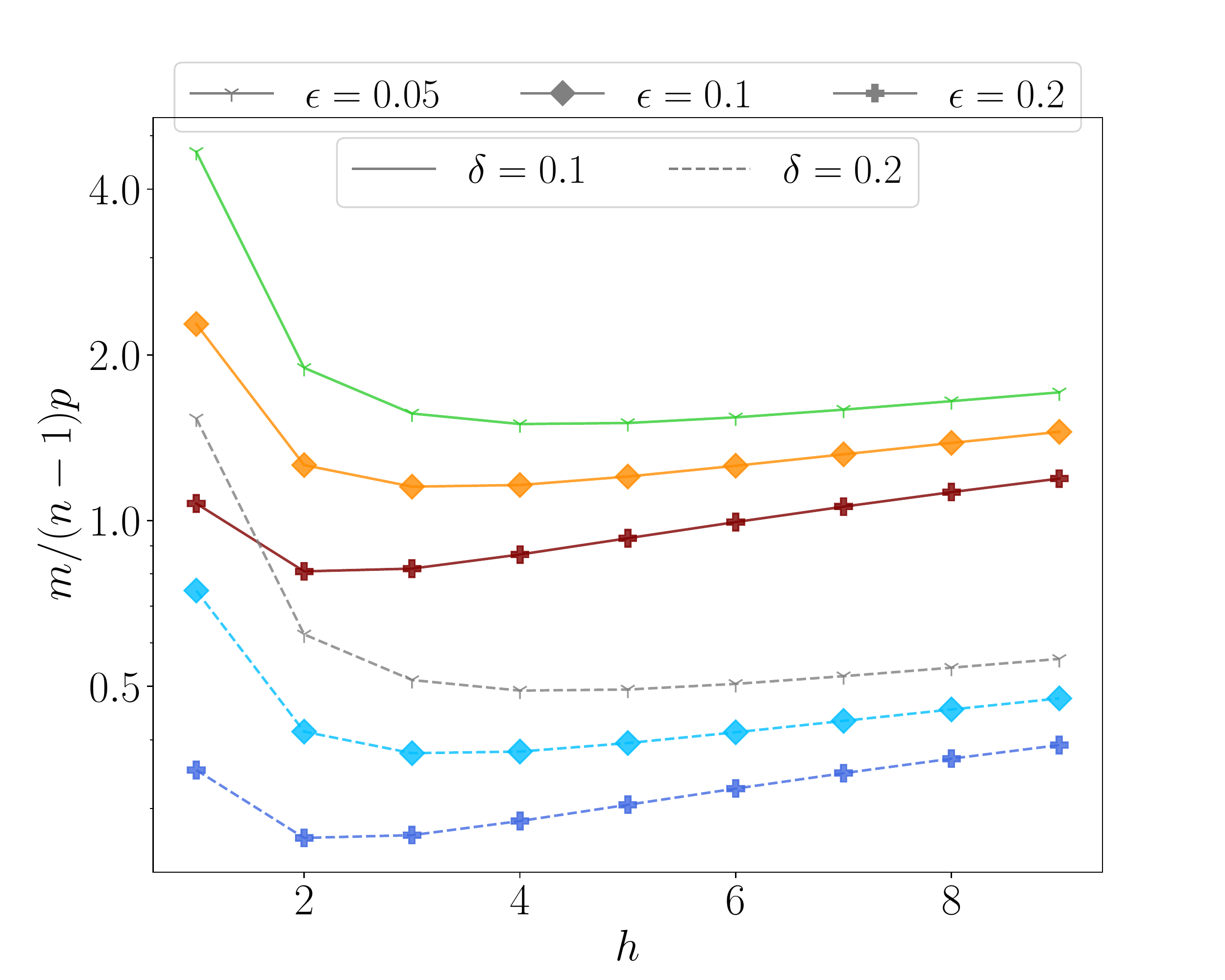}}
\subfloat[$\sigma=1$]{
\includegraphics[width=0.37\linewidth]{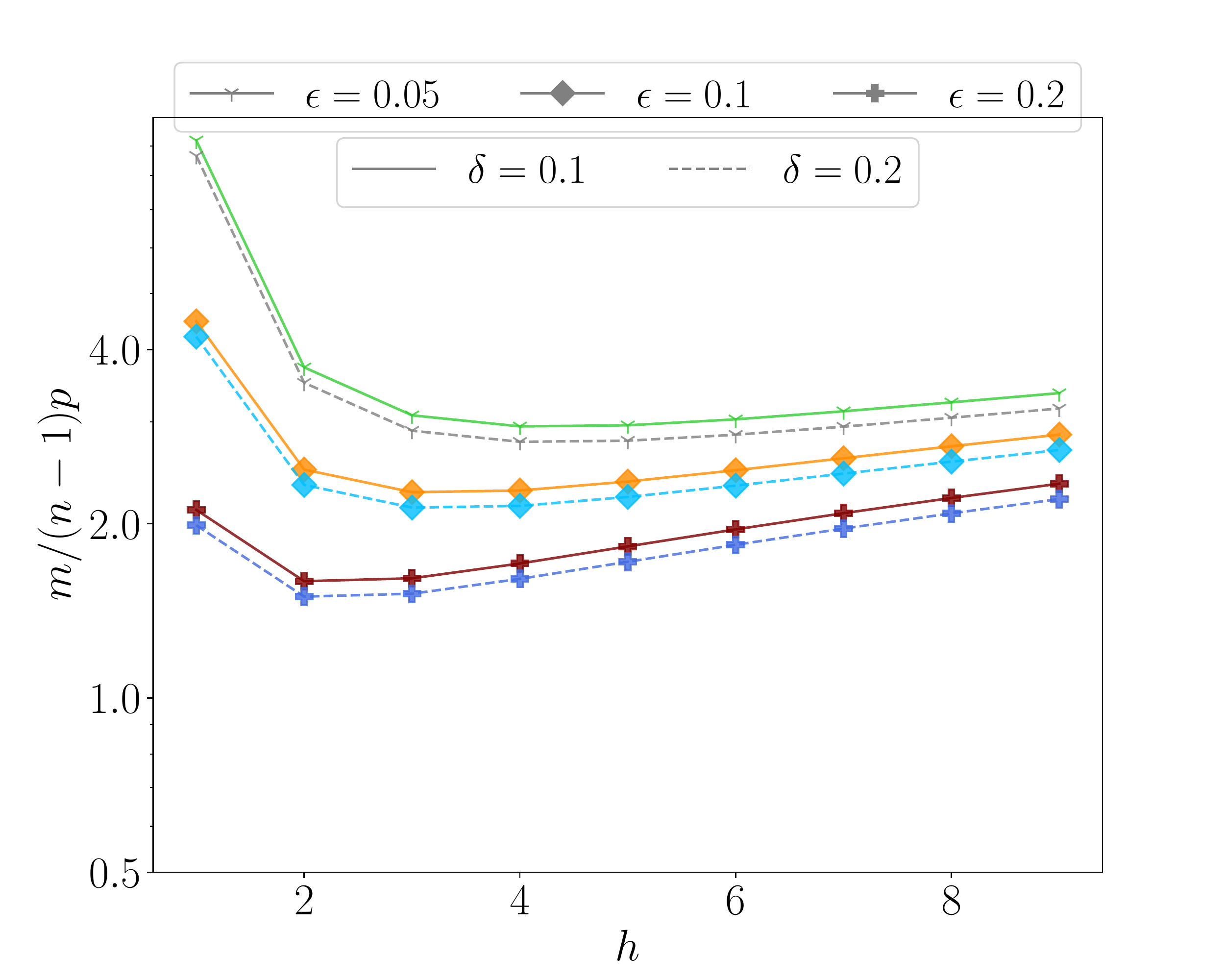}}
}
\caption{Visualization of Lemma~\ref{lemma:analysis:m-and-h} in terms of $m/(n-1)p$ as a function of $h$ for different values of $\sigma$, $\delta$, and $\epsilon$.}
\label{figure:analysis:m-and-h}
\end{center}
\end{figure}

\begin{figure}[t]
\begin{center}
\centerline{
\subfloat[Uniform]{
\includegraphics[width=0.37\linewidth]{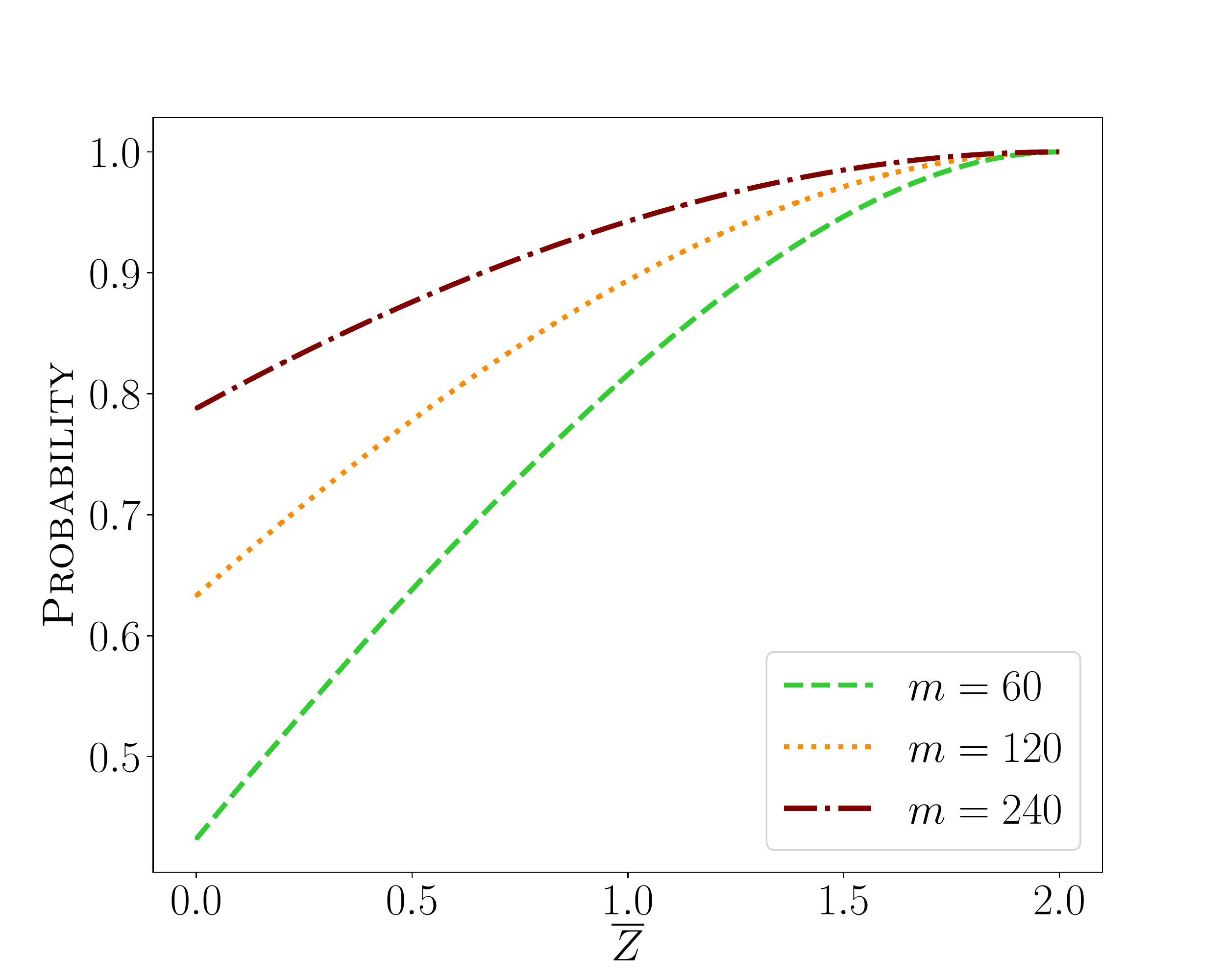}}
\subfloat[Gaussian$(\mu=0, \sigma=0.1)$]{
\includegraphics[width=0.37\linewidth]{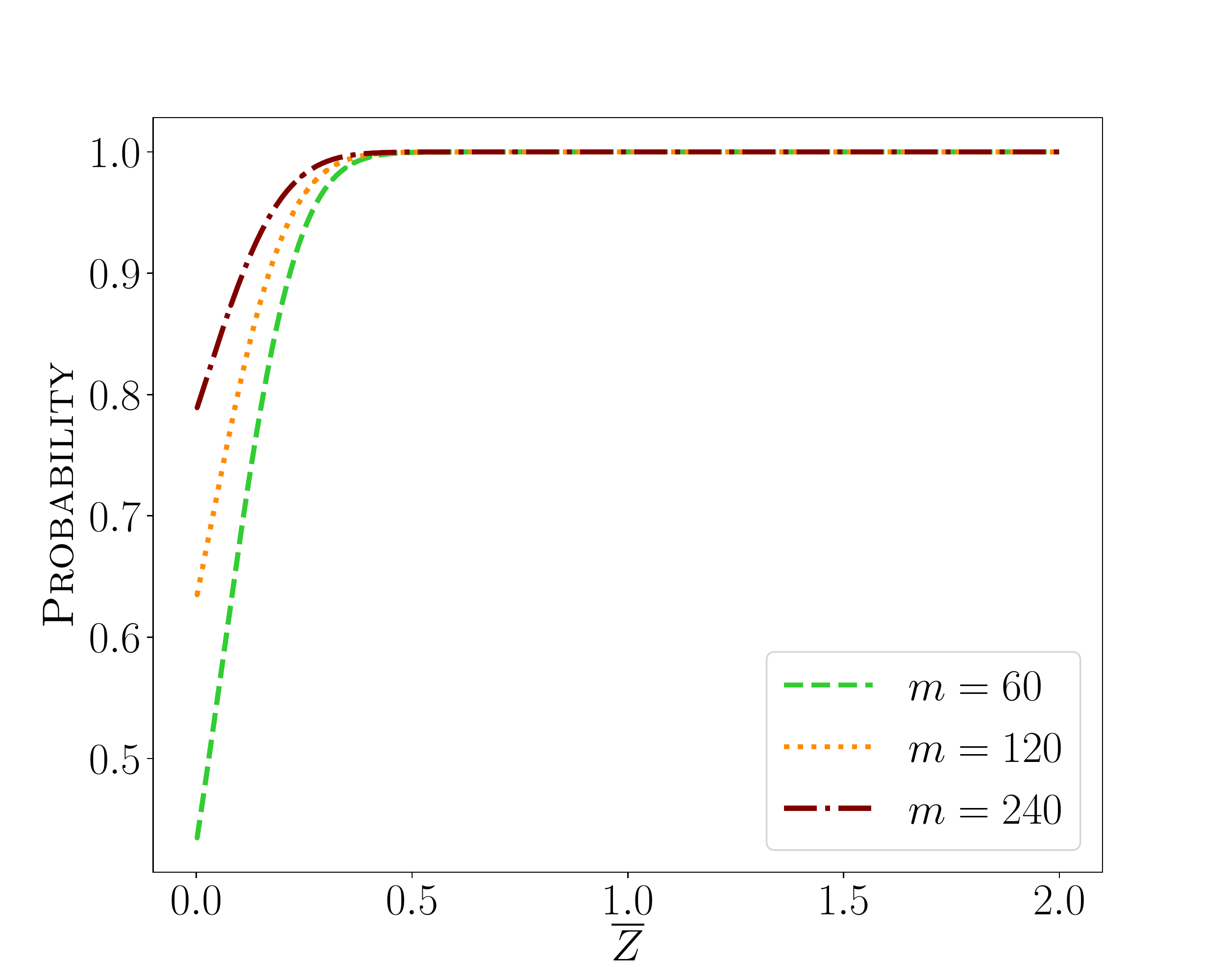}}
}
\caption{Cumulative distribution of the overestimation error, $\overline{Z}$, of the upper-bound sketch for vectors that draw values from the uniform distribution over $[-1, 1]$ and the Gaussian distribution. These curves represent different settings of $m$ with $h=1$ and $np_i = np =120$.}
\label{figure:analysis:cdf-error}
\end{center}
\end{figure}

Before we move on, let us consider the general form of $\overline{Z}_i$ without the assumption that $X_i$ is active.
Denote by $\mu_i$ the expected value of $\overline{Z}_i$ conditioned on $X_i$ being active: $\mu_i = \mathbb{E}[\overline{Z}_i; X_i \text{ is active}]$.
Similarly denote by $\sigma_i^2$ its variance when $X_i$ is active.
Given that $X_i$ is active with probability $p_i$ and inactive with probability $1 - p_i$,
it is easy to show that $\mathbb{E}[\overline{Z}_i] = p_i \mu_i$
and that its variance $\text{Var}(\overline{Z}_i) = p_i \sigma_i^2 + p_i (1 - p_i) \mu_i^2$.

\subsection{Inner Product Error}
In the previous section, we quantified the probability that a value decoded from the upper-bound sketch overestimates the original value of a randomly chosen coordinate. We also characterized the distribution of the overestimation error for a single coordinate and derived expressions for special distributions. In this section, we extend our analysis from a single coordinate to the inner product of a decoded vector with a fixed vector $q$.

We state the following result:

\begin{theorem}
    Suppose that $q \in \mathbb{R}^n$ is a sparse vector with $nz(q) = \{ i \;|\; q[i] \neq 0 \}$ denoting its set of active coordinates. Suppose in a random sparse vector $X \in \mathbb{R}^n$, a coordinate $X_i$ is active with probability $p_i$ and, when active, draws its value from some well-behaved distribution (i.e., expectation, variance, and third moment exist). Let $\tilde{X}$ be the reconstruction of $X$ by \sinnamon{}: when $X_i$ is active $\tilde{X}_i=\overline{X}_i$ if $q[i] > 0$ and $\tilde{X}_i = \underline{X}_i$ if $q[i] < 0$, otherwise $\tilde{X}_i$ is $0$. If $Z_i = \tilde{X}_i - X_i$, then the random variable $Z$ defined as follows:
    \begin{equation}
        Z \triangleq \frac{1}{\sqrt{\sum_{i \in nz(q)} \text{Var}[Z_i] q[i]^2}} \big( \langle q,\, \tilde{X} - X \rangle - \sum_{i \in nz(q)} p_i \mathbb{E}[Z_i] q[i] \big),
    \end{equation}
    approximately tends to a standard normal distribution as $\psi_q = |nz(q)|$ grows.
    \label{theorem:inner-product-error}
\end{theorem}

\begin{proof}
    
Let us expand the inner product between $q$ and $\tilde{X} - X$ as follows:
\begin{equation}
    \langle q, \tilde{X} - X \rangle = \sum_{i \in nz(q)} q[i] (\tilde{X}_i - X_i) = \sum_{i \in nz(q)} q[i] Z_i.
\end{equation}
We have already studied the variables $\overline{Z}_i = \overline{X}_i - X_i$ in the previous section and
noted that the analysis for $\underline{Z}_i$ is symmetrical. Therefore, we already know the properties of
$Z_i$, as it is a well-defined random variable that is $\overline{Z}_i$ when $q[i] > 0$ and $\underline{Z}_i$ otherwise.

As we only care about the order induced by scores, we are permitted to translate and scale the sum above by a constant as follows to arrive at $Z$:
\begin{align}
    \langle q, \tilde{X} - X \rangle &\overset{rank}{=} \sum_{i \in nz(q)} \underbrace{q[i] Z_i}_{Z^{\prime}_i} - \underbrace{q[i] \mathbb{E}[Z_i]}_{\mathbb{E}[Z_i^{\prime}]} = \sum_{i \in nz(q)} Z_i^{\prime} - \mathbb{E}[Z_i^{\prime}] \\
    &\overset{rank}{=} \frac{1}{\sqrt{\sum_{i \in nz(q)} \text{Var}[Z_i] q[i]^2}} \sum_{i \in nz(q)} Z_i - \mathbb{E}[Z_i] \\
    &= Z.
\end{align}
$Z$ is thus the sum of centered random variables $Z^{\prime}_i$. Because we assumed that the distribution of $X_i$ is well-behaved,
we can conclude that $\text{Var}[Z_i] > 0$ and that $\mathbb{E}[|Z_i|^3] < \infty$.
If we operated on the assumption that $Z_i$s are independent---in reality, they are weakly
dependent---albeit not identically distributed, we can appeal to the Berry-Esseen theorem to complete our proof.
\end{proof}

We note that, the conditions required by the result above are trivially satisfied when the random vectors $X_i$s are drawn from a distribution with bounded support.

\begin{figure}[t]
\begin{center}
\centerline{
\subfloat[Uniform]{
\includegraphics[width=0.37\linewidth]{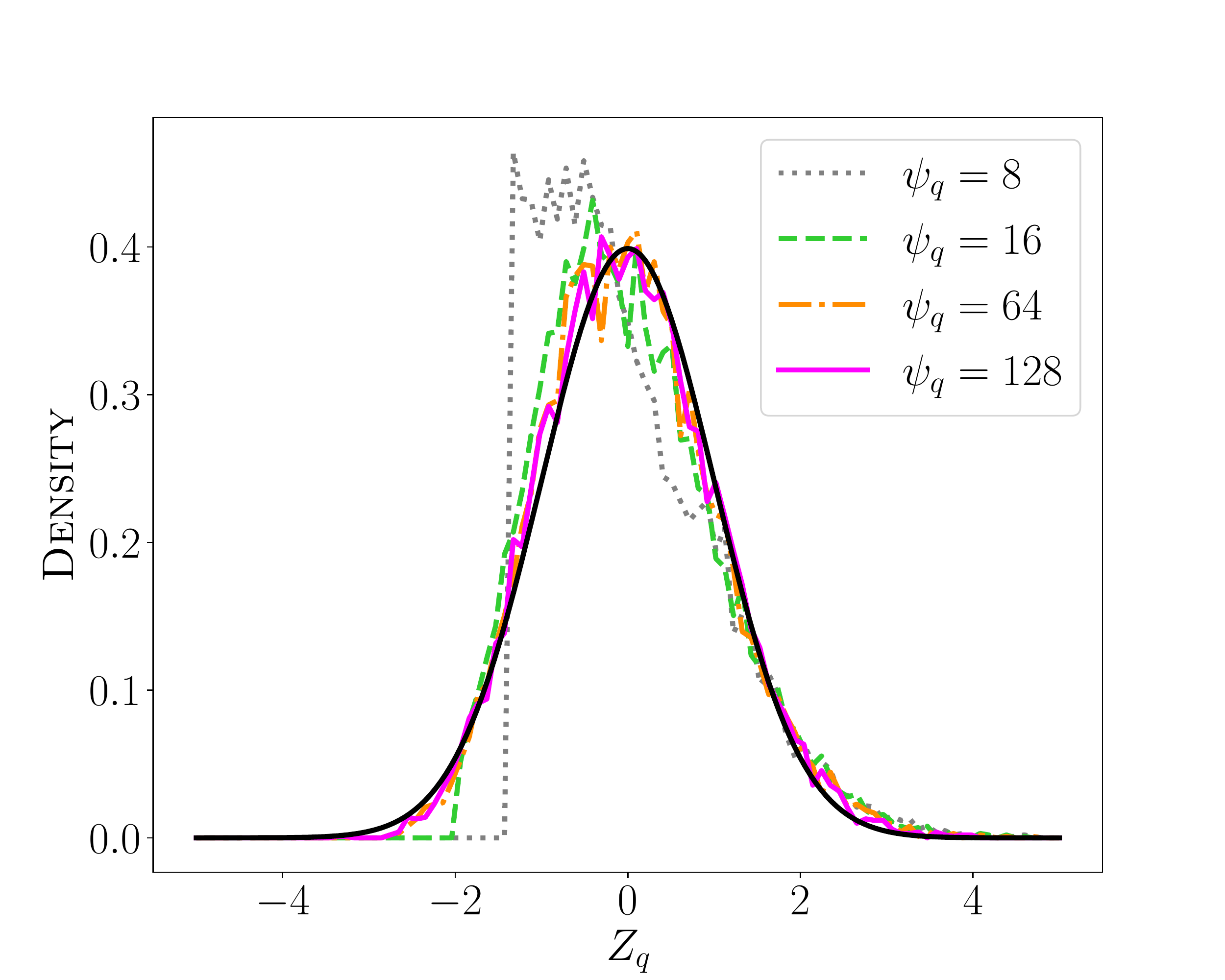}}
\subfloat[Gaussian$(\mu=0, \sigma=0.1)$]{
\includegraphics[width=0.37\linewidth]{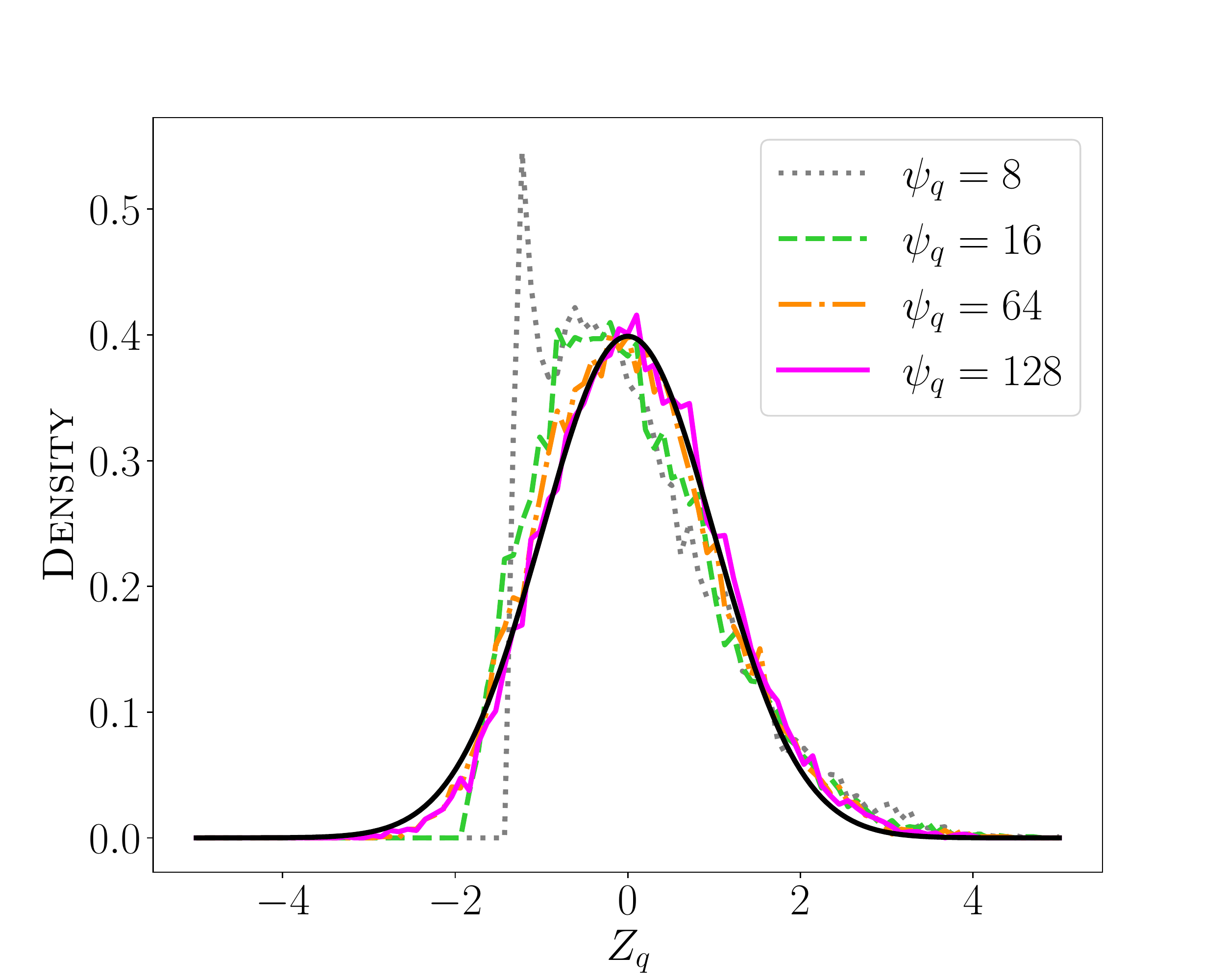}}
}
\caption{Distribution of the transformed inner product error $Z$ for vectors that draw values from the uniform distribution over $[-1, 1]$ and the Gaussian distribution. These curves represent different settings of the number of non-zero coordinates in the query ($\psi_q$) for $m=60$ with $h=1$ and $\psi_d = np_i = 120$. We also plotted the standard normal distribution for reference.}
\label{figure:analysis:inner-product-error}
\end{center}
\end{figure}

We verify the result above by simulating the following experiment. We draw a query with a given number of non-zero coordinates ($\psi_q$) from the standard normal distribution. We then draw a vector of errors $Z_i$ from a distribution defined by its CDF per Equation~(\ref{equation:analysis:dist-error}), and compute the transformed inner product error $Z$, and repeat this procedure $10,000$ times. We then plot the distribution of $Z$ in Figure~\ref{figure:analysis:inner-product-error} for two families of random vectors, one where we assume vectors are drawn from a Uniform distribution over $[-1, 1]$ and another from the Gaussian distribution with standard deviation $\sigma=0.1$. As the figures illustrate, $Z$ tends to a standard normal distribution even when the query has just a handful of active coordinates.

We conclude this section with a remark on what Theorem~\ref{theorem:inner-product-error} enables us to do. Let us assume that we know the distribution of the random variables $X_i$, or that we can estimate $p_i$, $\mathbb{E}[Z_i]$ and $\text{Var}[Z_i]$ from data. Armed with these statistics, we have all the ingredients to form $Z$ for a given query $q$ and any document vector. As $Z$ tends to a standard normal distribution, we can thus compute confidence bounds on the accuracy of the approximate inner product score returned by \sinnamon{}. This information can in theory be used to dynamically adjust $k^\prime$ in Algorithm~\ref{algorithm:retrieval:ranking} on a per-query basis. We leave an exploration of this particular result to future work.

\section{Evaluation}
\label{section:evaluation}

This section presents our empirical evaluation of \sinnamon{} and its properties on real-world data. We begin with a description of our empirical setup. We then verify the theoretical results of Section~\ref{section:analysis} on real datasets. That is followed by a discussion of the retrieval performance of \sinnamon{} where we examine the trade-offs the algorithm offers to configure memory, time, and accuracy. We finally turn to a review of insertions and deletions in \sinnamon{} and showcase its stable, online behavior.

\subsection{Setup}

\subsubsection{Sparse Vector Datasets}
We conduct our experiments on MS MARCO Passage Retrieval v1~\cite{nguyen2016msmarco}. This question-answering dataset is a collection of 8.8 million short passages in English with about $56$ terms ($39$ unique terms) per passage on average. We use the smaller ``dev'' set of queries for retrieval, consisting of $6,980$ questions with an average of about $5.8$ unique terms per query.

We demonstrate the utility of the algorithms with several different methods of encoding text into sparse vectors. In particular, we process MS MARCO passages and queries using BM25~\cite{bm25original,bm25}, SPLADE~\cite{formal2022splade}, efficient SPLADE~\cite{lassance2022sigir}, and uniCOIL~\cite{unicoil}.

\begin{table*}[t]
\caption{The number of non-zero entries in documents ($\psi_d$) and queries ($\psi_q$) for various vector datasets.}
\label{table:evaluation:datasets}
\begin{center}
\begin{sc}
\begin{tabular}{c|cccc}
& BM25 & SPLADE & Efficient Splade & uniCOIL \\
\toprule
\midrule
$\psi_d$ & 39 & 119 & 181 & 68 \\
$\psi_q$ & 5.8 & 43 & 5.9 & 6 \\
\bottomrule
\end{tabular}
\end{sc}
\end{center}
\end{table*}

It is a well-known fact that \textbf{BM25} can be translated into an inner product of two vectors, where each non-zero entry in the query vector has the IDF of the corresponding term in the vocabulary, and the document vectors encode BM25's term importance weight. This requires that the average document length and the hyperparameters $k_1$ and $b$ be fixed, but that is a reasonable assumption for the purposes of our experiments. We set $k_1=0.82$ and $b=0.68$ by tuning using a grid search. We drop the $(k_1+1)$ factor from the numerator of BM25's term importance weight so that document entries are bounded to $[0, 1]$; this is a rank-preserving change. We pre-process the text of the collection and queries using the default word tokenizer and Snowball stemmer of the open-source \textit{NLTK}~\footnote{Available at \url{https://github.com/nltk/nltk}} library. Note that, we include BM25 simply to provide a reference point and emphasize that \linscan{} and \sinnamon{} are designed as general-purpose solutions suitable for real-valued sparse vectors of which BM25 is a special case.

\begin{figure}[t]
\begin{center}
\centerline{
\subfloat[Values]{
\includegraphics[width=0.37\linewidth]{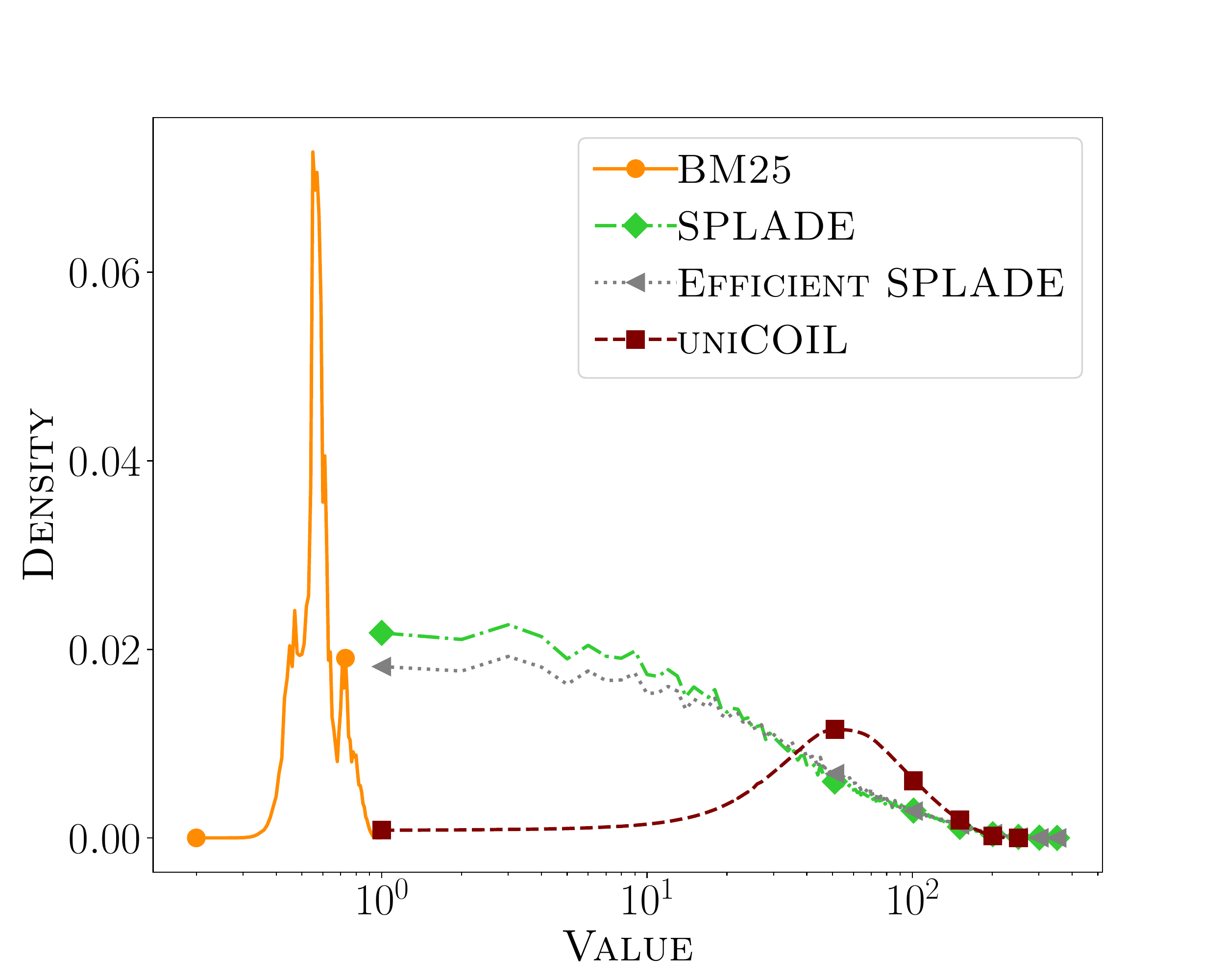}\label{figure:sparse-vector-datasets:distributions:value}}
\subfloat[Non-zero Coordinates]{
\includegraphics[width=0.37\linewidth]{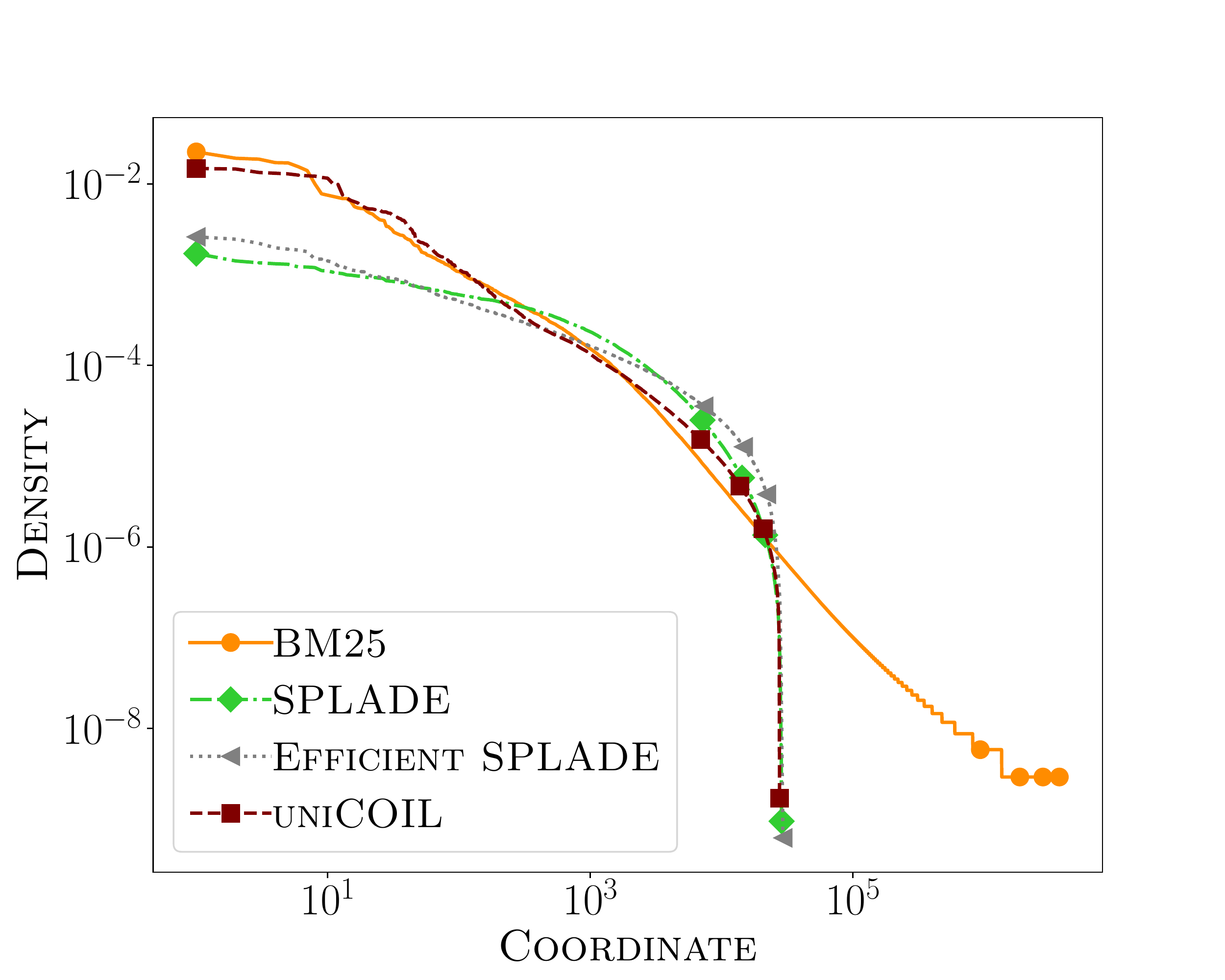}\label{figure:sparse-vector-datasets:distributions:nnz}}
}
\caption{Distributions of values and non-zero coordinates in various vector datasets. For each dataset, the figure on the left shows the likelihood of a non-zero coordinate taking on a particular value. The figure on the right shows the likelihood of a particular coordinate being non-zero.}
\label{figure:sparse-vector-datasets:distributions}
\end{center}
\end{figure}

As the second model, we use \textbf{SPLADE\footnote{Pre-trained checkpoint from HuggingFace available at \url{https://huggingface.co/naver/splade-cocondenser-ensembledistil}}}~\cite{formal2022splade}, a deep learning model that produces sparse representations for a given piece of text, where each non-zero entry is the importance weight of a term in the BERT~\cite{devlin2019bert} WordPiece~\cite{wordpiece} vocabulary comprising of $30,000$ terms. When encoded with this version of SPLADE, the MS MARCO passage vectors contain an average of $119$ non-zero entries and the query vectors $43$ non-zero entries.

We include SPLADE because it enables us to test the behavior of retrieval algorithms on query vectors with a large number of non-zero entries. However, we also create another vector dataset from MS MARCO using a more efficient variant of SPLADE, called \textbf{Efficient SPLADE}\footnote{Pre-trained checkpoints for document and query encoders were obtained from \url{https://huggingface.co/naver/efficient-splade-V-large-doc} and \url{https://huggingface.co/naver/efficient-splade-V-large-query}, respectively}~\cite{lassance2022sigir}. This model produces queries that have far fewer non-zero entries than the original SPLADE model but documents that have a larger number of non-zero entries. More concretely, the mean $\psi_d$ of document vectors is $181$ and that of query vectors is about $5.9$. Due to the larger size of the document collection, this dataset helps us examine the memory footprint of \sinnamon{} in a relatively more extreme scenario.

Similar to SPLADE, \textbf{uniCOIL}~\cite{unicoil} produces impact scores for terms in the vocabulary, resulting in sparse representations for text documents. We use the checkpoint provided by the authors\footnote{Available at \url{https://github.com/castorini/pyserini/blob/master/docs/experiments-unicoil.md}.} to obtain vectors for the MS MARCO collection. This results in document vectors that have on average $68$ non-zero entries, and query vectors with $\psi_q \approx 6$.

We would be remiss if we did not note that all vector datasets produced from the MS MARCO dataset are non-negative. This is a limitation of BM25 and existing embedding models that generate sparse vectors for text. However, the results presented in this section are generalizable to real vectors. To support that claim, we discuss this topic further and present evidence on synthetic data at the end of this section. We further hope that our algorithmic contribution inspires sparse embedding models that can leverage the whole real line including negative weights---an area that is hitherto unexplored due to a lack of efficient SMIPS algorithms.

As a way to compare and contrast the various vector datasets above, we examine some of their pertinent characteristics in Table~\ref{table:evaluation:datasets} and Figure~\ref{figure:sparse-vector-datasets:distributions}. In Figure~\subref*{figure:sparse-vector-datasets:distributions:value}, we plot a histogram of the coordinate values, showing the likelihood that a non-zero coordinate takes on a particular value. One notable difference between the datasets is that SPLADE and its other variant, Efficient SPLADE, appear to have a very different value distribution than BM25 and uniCOIL: in the former smaller values are more likely to occur.

The datasets are different in one other way: the likelihood of a coordinate being non-zero. We plot this distribution in Figure~\subref*{figure:sparse-vector-datasets:distributions:nnz} (in log-log scale, with smoothing to reduce noise in our visualization). We notice that the distributions for (Efficient) SPLADE have a heavier tail than the shape of BM25 and uniCOIL distributions.

\subsubsection{Evaluation Metrics}
There are four metrics that will be our focus in the remainder of this work. First is the \textbf{index size} measured in GB. We rely on the programming language---which, in this work, is Rust\footnote{More information about the language is available at \url{https://www.rust-lang.org/}}---to calculate the amount of space held by a data structure and estimate the memory footprint of the overall index. In \sinnamon{}, this measurement includes the size of the inverted index as well as the sketch matrix.

We also report \textbf{latency} in milliseconds. When reporting the latency of retrieval, this figure includes the time elapsed from the moment a query vector is presented to the algorithm to the moment the algorithm returns the requested top $k$ document vectors. In \sinnamon{}, for example, this includes the scoring time of Algorithm~\ref{algorithm:retrieval} as well as the ranking time of Algorithm~\ref{algorithm:retrieval:ranking}. We note that, because this work is a study of retrieval of generic vectors, we do not report the latency incurred to vectorize a given piece of text.

As another metric of interest, we report the \textbf{accuracy} of approximate algorithms in terms of their recall with respect to exact retrieval. By measuring the recall of an approximate set with respect to an exact top-$k$ set, we can study the impact of the different levers in the algorithm on its overall accuracy as a retrieval engine.

Finally, we also evaluate the algorithms according to task-specific metrics. Because the task in MS MARCO is to rank passages according to a given query, we measure Normalized Discounted Cumulative Gain (\textbf{NDCG})~\cite{jarvelin2000ir} at a deep rank cutoff ($1000$) and Mean Reciprocal Rank (\textbf{MRR}) at rank cutoff $10$. In this way, we examine the impact of \sinnamon{}'s levers on the quality of its solution from the perspective of the end task.

\subsubsection{Hardware}
We conduct experiments on two different commercially available platforms. One is an Intel Xeon (Ice Lake) processor with a clock rate of $2.60$GHz with $8$ CPU cores and $64$GB of main memory. Another is an Apple M1 Pro processor with the same core count ($8$) and main memory capacity ($64$GB). Not only do these processors have different characteristics and, as such, shed light on \sinnamon{}'s behavior in the context of different architectures, but they also represent two different use-cases. The first of these platforms represents a typical server in a production environment---in fact, we rented this machine from the Google Cloud Platform---while the second represents a vast number of end-user devices such as laptops, tablets, and phones. We believe given that \sinnamon{} can be tailored to different memory and latency configurations, it is important to understand its performance both in a production setting as well as on edge devices.

\subsubsection{Algorithms}
We begin, of course, with \textbf{\linscan{}} and \textbf{\sinnamon{}}. We implement the two algorithms and their variants in Rust. This implementation includes support for $\sinnamon{}^+$, $\sinnamon{}^\parallel$, \roaringlinscan{}, $\linscan{}^\parallel$, and their anytime variants. We note that, because the vectors produced by existing models are all non-negative, we only experiment with $\sinnamon{}^+$ and its parameter $m$ (sketch size) on the MS MARCO vector datasets.

To facilitate a more clear discussion of \sinnamon{} and manage the size of the experimental configurations, we remove one variable from the equation in a subset of our experiments. In particular, as we note later, we fix \sinnamon{}'s parameter $h$ to $1$ and study $h>1$ only in a limited number of experiments. We believe, however, that in totality, our experiments along with theoretical results sufficiently explain the role that $h$ plays in the algorithm. Despite that, we do encourage practitioners to explore the trade-offs \sinnamon{} offers for their own use case and application, including tuning $h$ to reach a desired balance between latency, accuracy, and space.

As we noted earlier, to the best of our knowledge, \sinnamon{} is the first approximate algorithm for MIPS over general sparse vectors. In the absence of other general purpose systems designed for sparse vectors, we resort to modifying and implementing in Rust existing algorithms from the top-$k$ retrieval literature so they may operate on real-valued vectors. In particular, we take the popular \textbf{\textsc{Wand}}~\cite{broder2003wand} algorithm as a representative example. In our examination of \textsc{Wand}, we emphasize the logic of the algorithm itself---its document-at-a-time process with a pruning mechanism that skips over the less-promising documents. To that end, our implementation of \textsc{Wand} makes use of an uncompressed inverted index (where a posting is a pair consisting of a document identifier and a $32$-bit floating point value). In this way, we isolate the latency of the algorithm itself without factoring in costs incurred by decompression or other operations unrelated to the retrieval logic itself. As such, we advise against a comparison of the index size in \textsc{Wand} with that in \sinnamon{}.

\subsection{Analysis}

\begin{figure}[t]
\begin{center}
\centerline{
\subfloat[Cumulative distribution of sketching error for $h=1$ (left) and $h=2$ (right)]{
\includegraphics[width=0.37\linewidth]{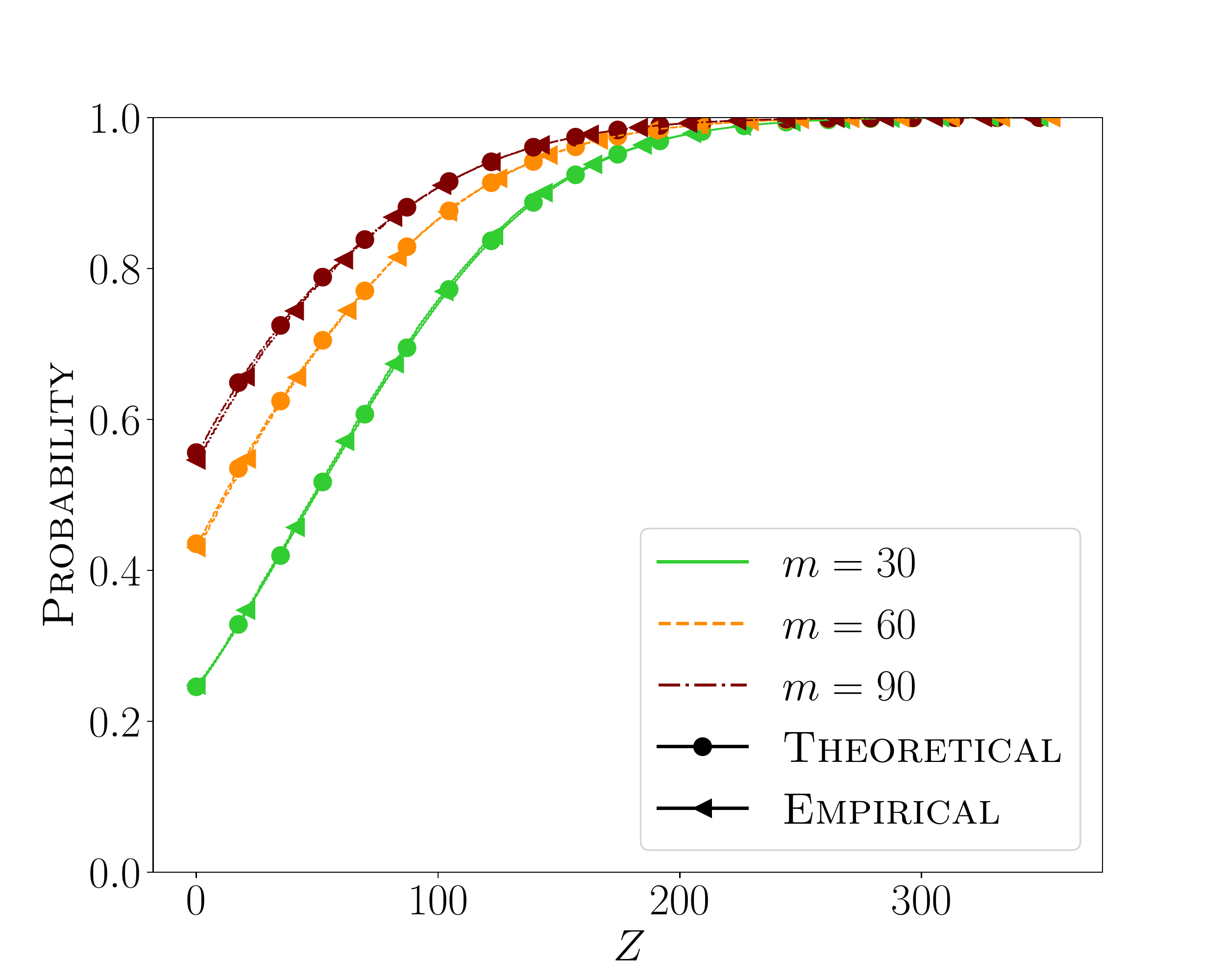}
\includegraphics[width=0.37\linewidth]{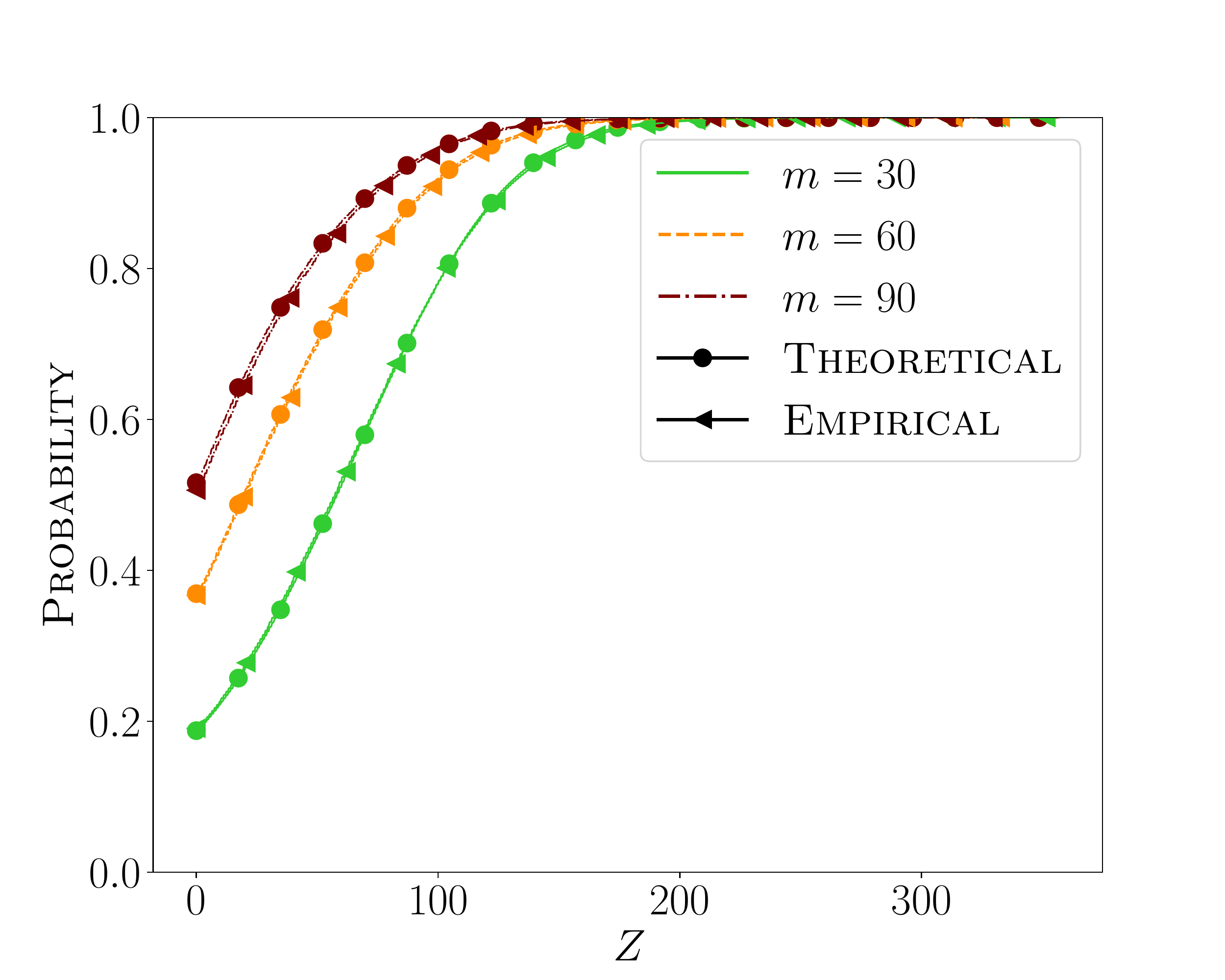}
\label{figure:evaluation:analysis:splade:cdf}}
}
\centerline{
\subfloat[Translated and scaled inner product error for $h=1$ (left) and $h=2$ (right)]{
\includegraphics[width=0.37\linewidth]{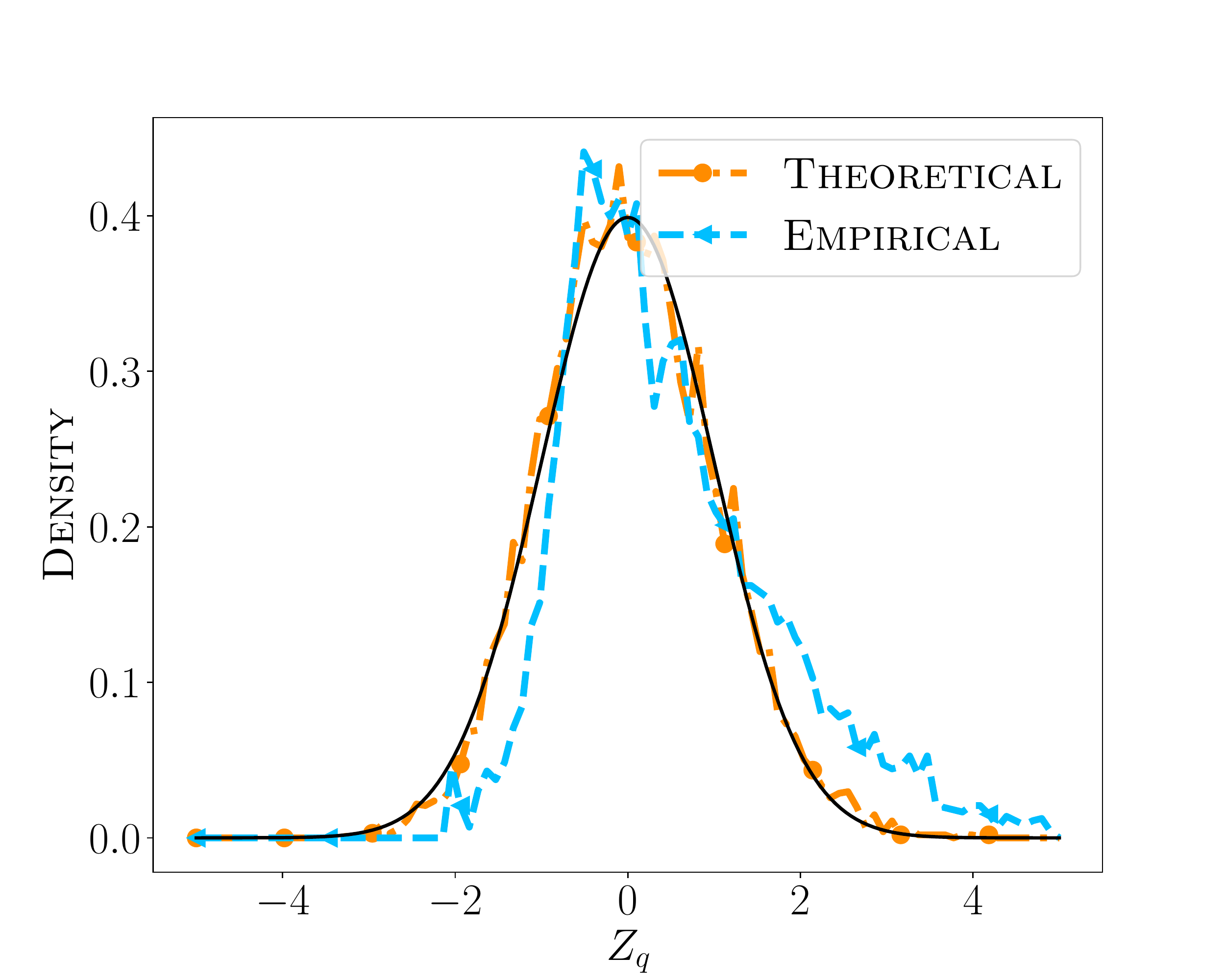}
\includegraphics[width=0.37\linewidth]{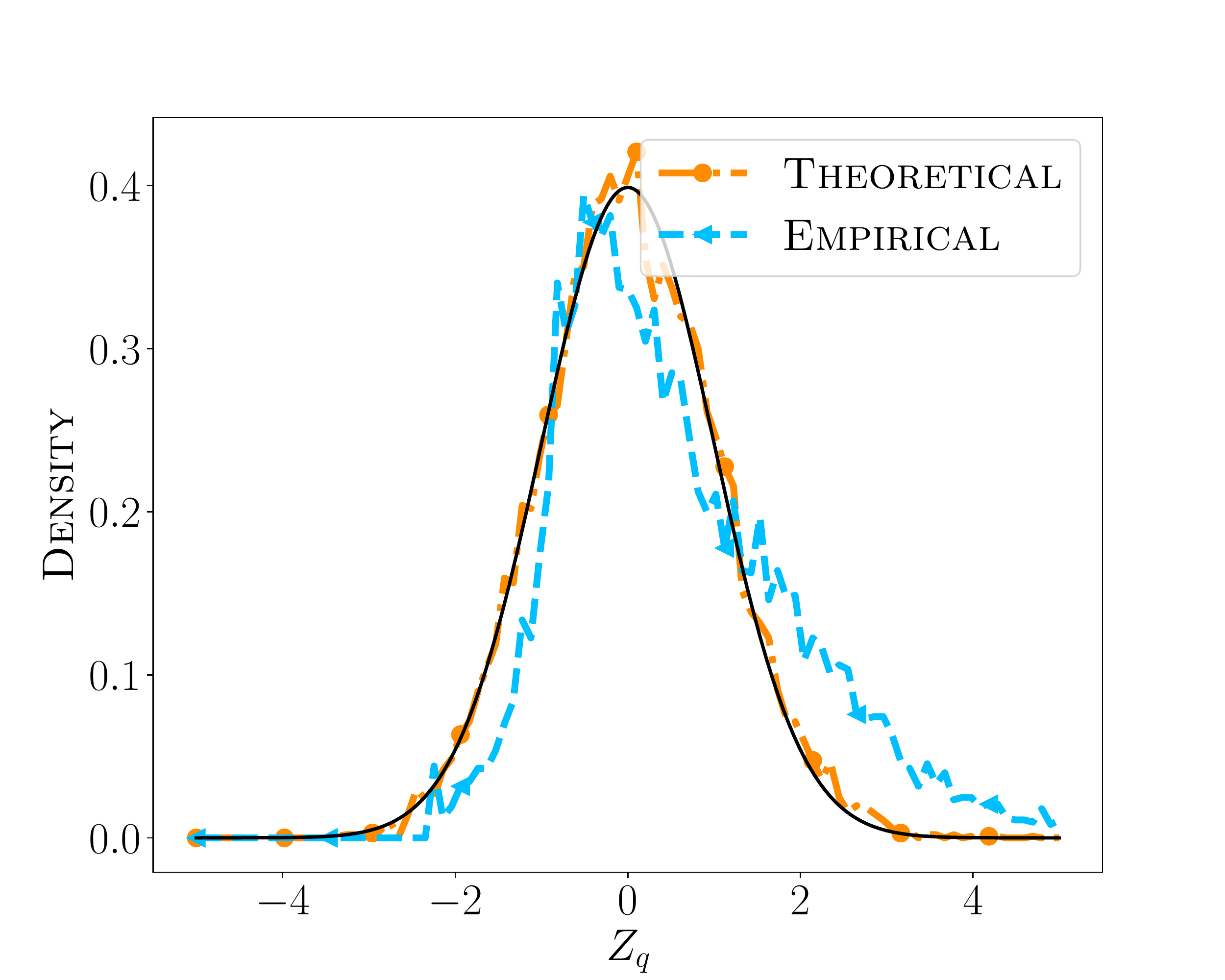}
\label{figure:evaluation:analysis:splade:zq}}
}

\caption{Sketching and inner product error distributions for the SPLADE dataset predicted by theory and inferred empirically. In (a), the theoretical and empirical CDF of the sketching error are almost indistinguishable---theoretical and empirical curves for $m=30$ overlap, so do curves for other values of $m$. In (b), the solid dark curve plots a standard normal distribution for reference.}
\label{figure:evaluation:analysis:splade}
\end{center}
\end{figure}

In this section, we examine the theoretical results of Section~\ref{section:analysis} on the vector datasets generated from MS MARCO. In particular, we infer from the data the empirical probability that coordinate $i$ is non-zero ($p_i$) and the empirical distribution of non-zero values (i.e., the random variables $X_i$).

We then use these statistics to theoretically predict the distribution of sketching error using Theorem~\ref{theorem:cdf-upperbound-sketch} and the inner product error from Theorem~\ref{theorem:inner-product-error}. We refer to statistics predicted by the theorems as ``theoretical'' errors in the remainder of this section and in accompanying figures.

In addition to theoretical predictions, we compute the distribution of sketching and inner product errors directly from the data, which we refer to as ``empirical'' errors. To infer the sketching error, for example, we first sketch a document vector, and then decode the value of its non-zero coordinates from the sketch. We then compute the error between the actual value and the decoded value. By repeating this process for every document in the collection, we can empirically construct the CDF of sketching error.

We find the distribution of the random variable $Z$ from Theorem~\ref{theorem:inner-product-error} by taking a query-document pair, computing their inner product using an exact algorithm as well as \sinnamon{}, and recording the difference between the two values. We then use the query $q$, the non-zero probability of document coordinates $p_i$, and statistics from $X_i$ to compute and plot $Z$.

Now that we have the theoretical and empirical predictions of the distribution of error, we compare the two to verify that the theory holds in practice for arbitrary data distributions. We show this comparison in Figure~\ref{figure:evaluation:analysis:splade} for the SPLADE dataset. We included a similar comparison for the remaining datasets in Appendix~\ref{appendix:analysis-empirical}.

As a general observation, the predictions from theory are accurate. We observe, for example, that $Z$ takes on a Gaussian shape. We further observe that the predicted CDF of sketching error reflects the empirical error. It is also worth noting that, increasing the number of random mappings from $1$ to $2$ results in an increase in the probability of sketching error but a decrease in the expected value of error---the sketching error concentrates closer to $0$. For example, when $m=90$, the probability of sketching error changes from $0.45$ to $0.48$, but the expected value of error improves from $28$ to $21$.

\subsection{Retrieval}

We begin with an examination of index size, accuracy, and latency as a function of the sketch size $m$ and time budget $T$ in \sinnamon{} on the different vector collections described previously. In conjunction with the mono-CPU results, we present the parallel version of the algorithm denoted by $\sinnamon{}^\parallel$ run on $8$ threads. In all the experiments in this section, we set $k$ (in top $k$ retrieval) to $1,000$.

\subsubsection{Latency, Memory, and Retrieval Accuracy}
Our objective is to illustrate the Pareto frontier that \sinnamon{} explores. However, because there are three objectives involved, instead of rendering a three-dimensional image, we show in one plot a pictorial presentation of the trade-off between latency and index size, and couple it with another plot that depicts the interplay between latency and accuracy (i.e., recall with respect to exact retrieval). In each figure, we distinguish between different configurations of $m$ using shapes and colors, and between the different time budgets $T$ for a given $m$ by labeling the points in the figure. We also show results from baseline algorithms for reference, including \linscan{} and its compressed, parallel, and anytime flavors.

Let us now turn to Figure~\ref{figure:evaluation:msmarco-passage-v1-icelake} where we plot the trade-offs between latency, memory, and accuracy on the Intel platform. As we observe similar trends on the M1 platform, we do not include those figures here and refer the reader to Appendix~\ref{appendix:retrieval-m1} for results; we note briefly, however, that all algorithms run substantially faster on M1 due to architectural differences---as memory is mounted on the processing chip in M1, we observe a higher memory throughput, leading to a significant speed-up. Finally, we note that to generate these figures we run \sinnamon{} with $k^\prime=5,000$ and study the impact of $k^\prime$ on latency and accuracy later in this section. All runs of the anytime version of \linscan{} too use $k^\prime=5,000$.

Turning our attention to \textsc{Wand} first---and ignoring its memory footprint as discussed previously---we note its excellent performance on the collection of BM25 vectors. This is, after all, as expected. Because term frequencies---the main ingredient of BM25 encoding---follows a Zipfian distribution; because queries are short; and because the importance of query terms is non-uniform, \textsc{Wand}'s pruning mechanism is able to narrow the search space substantially, leading to a very low latency.

As we move to other collections, however, we lose some of the properties that make \textsc{Wand} fast. For example, on the Efficient SPLADE collection, queries have few non-zero terms and we observe that \textsc{Wand} demonstrates a reasonable latency. But queries are an order of magnitude longer in the SPLADE collection, leading to a dramatic increase in the latency of the algorithm.

Besides the evidence above, we also believe \textsc{Wand} is not suitable for the general setting because its core pruning idea is designed specifically for Zipfian-distributed values. If each coordinate instead has more or less the same likelihood of being non-zero in any given vector or when non-zero entries have a Gaussian distribution, then \textsc{Wand} fails to prune documents effectively and its logic of finding a pivot suddenly becomes the dominant term in its computational complexity.

Interestingly, \linscan{}, which traverses the postings list exhaustively but in a coordinate-at-a-time manner, is often much faster than \textsc{Wand}. As we intuited before, this is due to better cache utilization and instruction-level parallelism that takes advantage of wide CPU registers, both made possible by the algorithm's predictable, sequential scan of arrays in memory. In effect, \linscan{} represents a lower-bound of sorts on \sinnamon{}'s mono-CPU performance, because they share the same index traversal logic but where \sinnamon{}'s retrieval logic involves heavier computation.

Things change when the inverted index in \linscan{} is compressed, turning the algorithm to \roaringlinscan{}. As a general observation, latency tends to increase substantially. This rise can be attributed to the cost of decompression, which inevitably makes the data structure and index traversal logic less friendly to the CPU's caches and registers. Nonetheless, assuming vector values cannot be quantized and must at least occupy $16$ bits, the index size in \roaringlinscan{} represents a ceiling of sorts for \sinnamon{}; that is, if a configuration of \sinnamon{} leads to a higher memory usage than \roaringlinscan{}, then there is no real advantage to utilizing \sinnamon{} in that setup.

Now consider the curves for \sinnamon{}. Naturally, by reducing the scoring time budget $T$, we observe a decrease in overall latency---which includes ranking with unlimited time budget. We also observe, as one would expect, a decrease in retrieval accuracy. This effect is milder in the parallel version of \sinnamon{}. We observe the same trend as we tighten the scoring time budget in the anytime version of \roaringlinscan{}.

In our experiments, we set $m$ to be roughly $25\%$, $50\%$, and $75\%$ of the average $\psi_d$ of the document vector collection (see Table~\ref{table:evaluation:datasets}); this, for example, translates to $10$, $20$, and $30$ for the BM25-encoded dataset. Again, as anticipated, by reducing the sketch size $m$, \sinnamon{} allocates less and less memory to store document sketches. The gap between the different configurations of $m$ and \roaringlinscan{} is not so large where document vectors have few non-zero entries (e.g., BM25) but it widens on collections with a larger $\psi_d$.

\begin{figure}[ht]
\begin{center}
\centerline{
\subfloat[BM25]{
\includegraphics[width=0.45\linewidth]{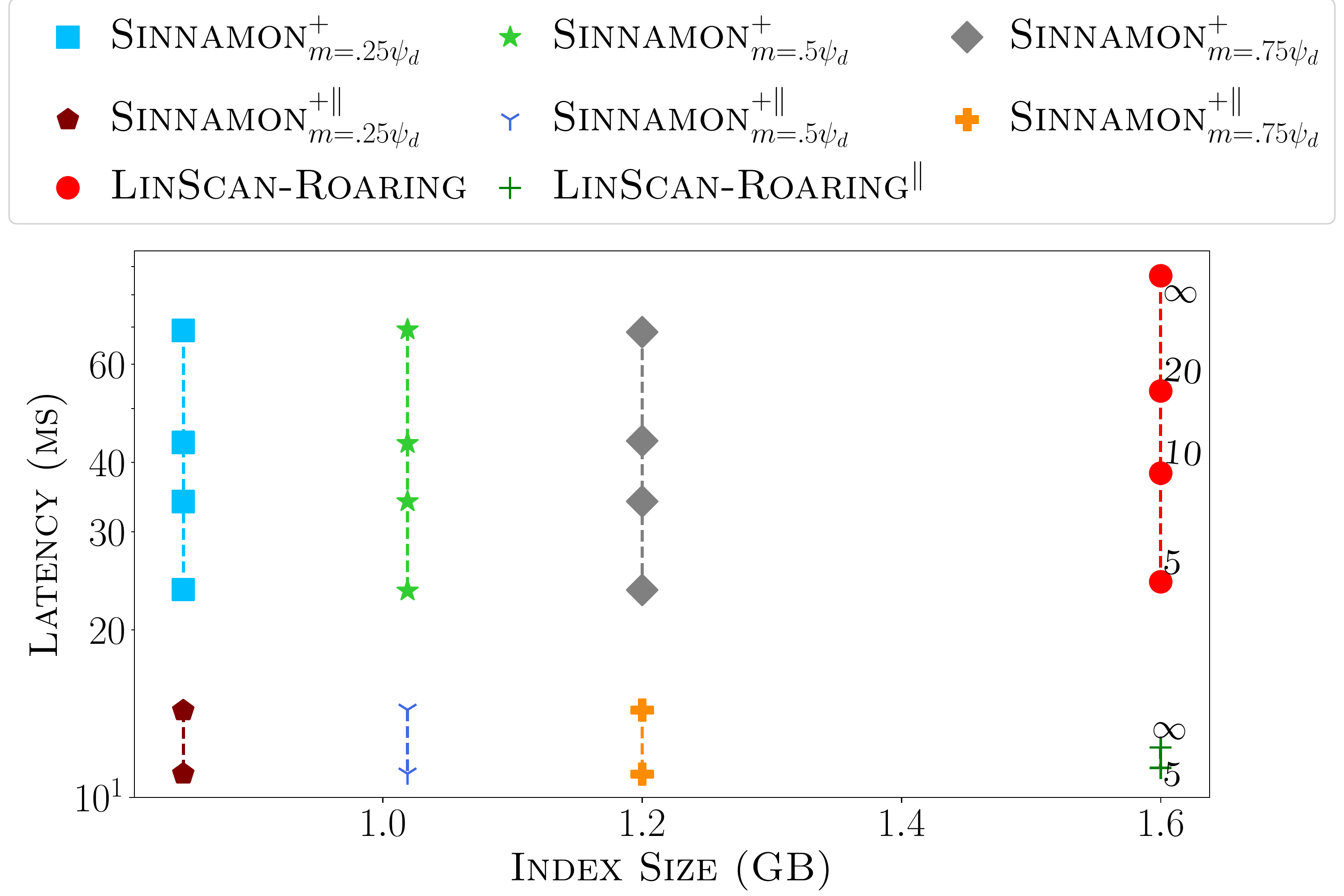}
\includegraphics[width=0.45\linewidth]{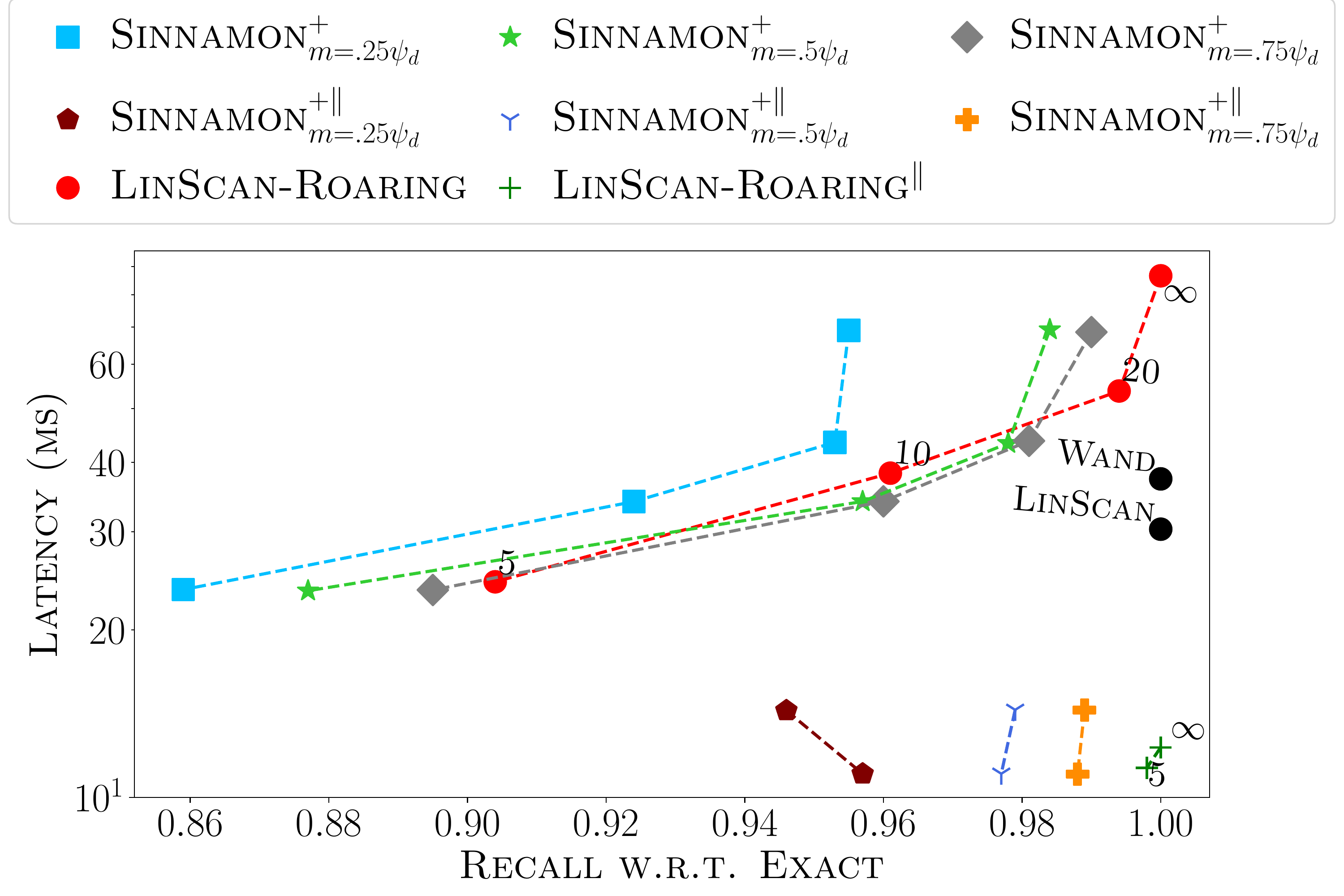}}
}

\centerline{
\subfloat[SPLADE]{
\includegraphics[trim={0 0 0 5.2cm},clip,width=0.45\linewidth]{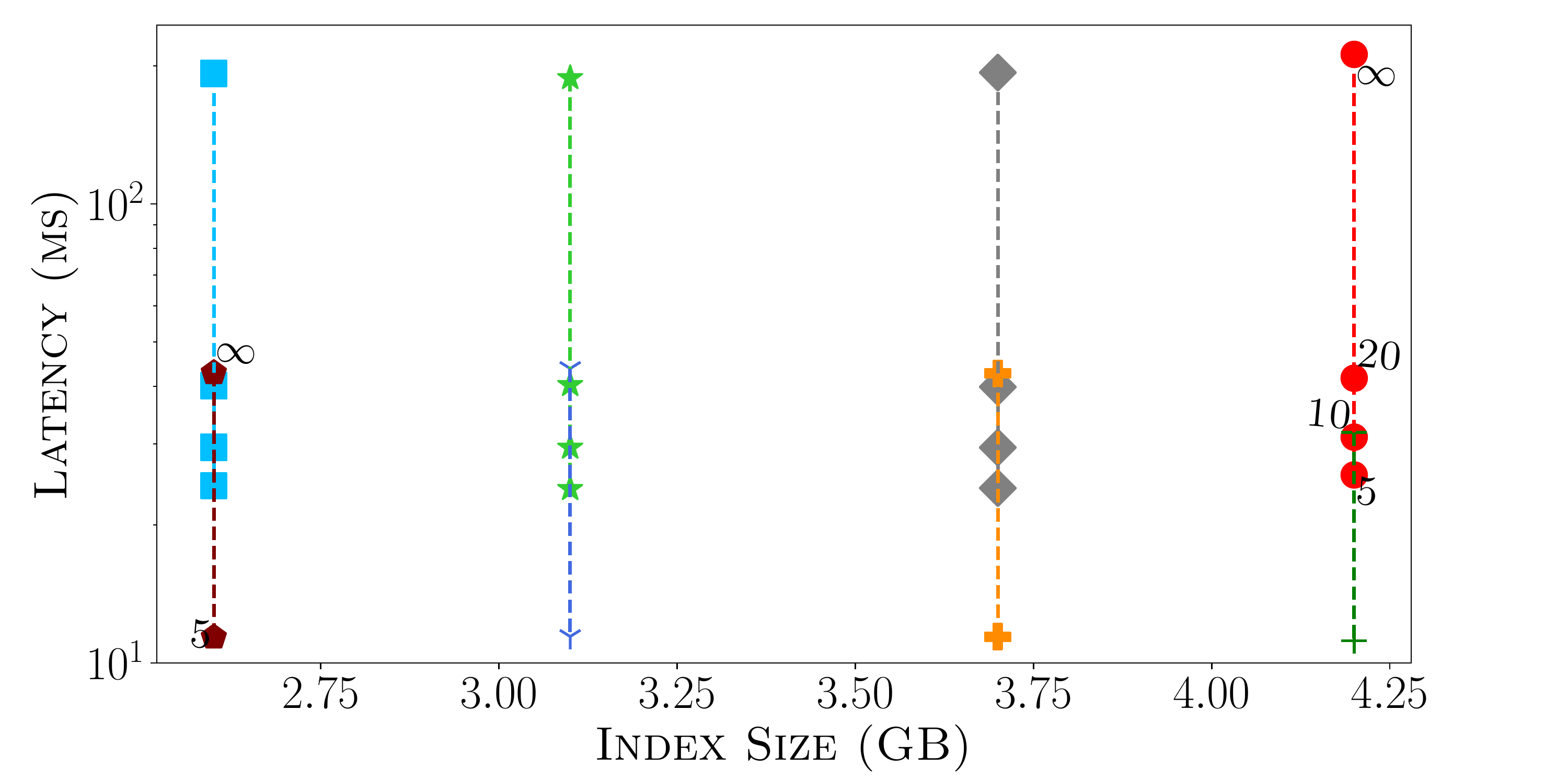}
\includegraphics[trim={0 0 0 5.2cm},clip,width=0.45\linewidth]{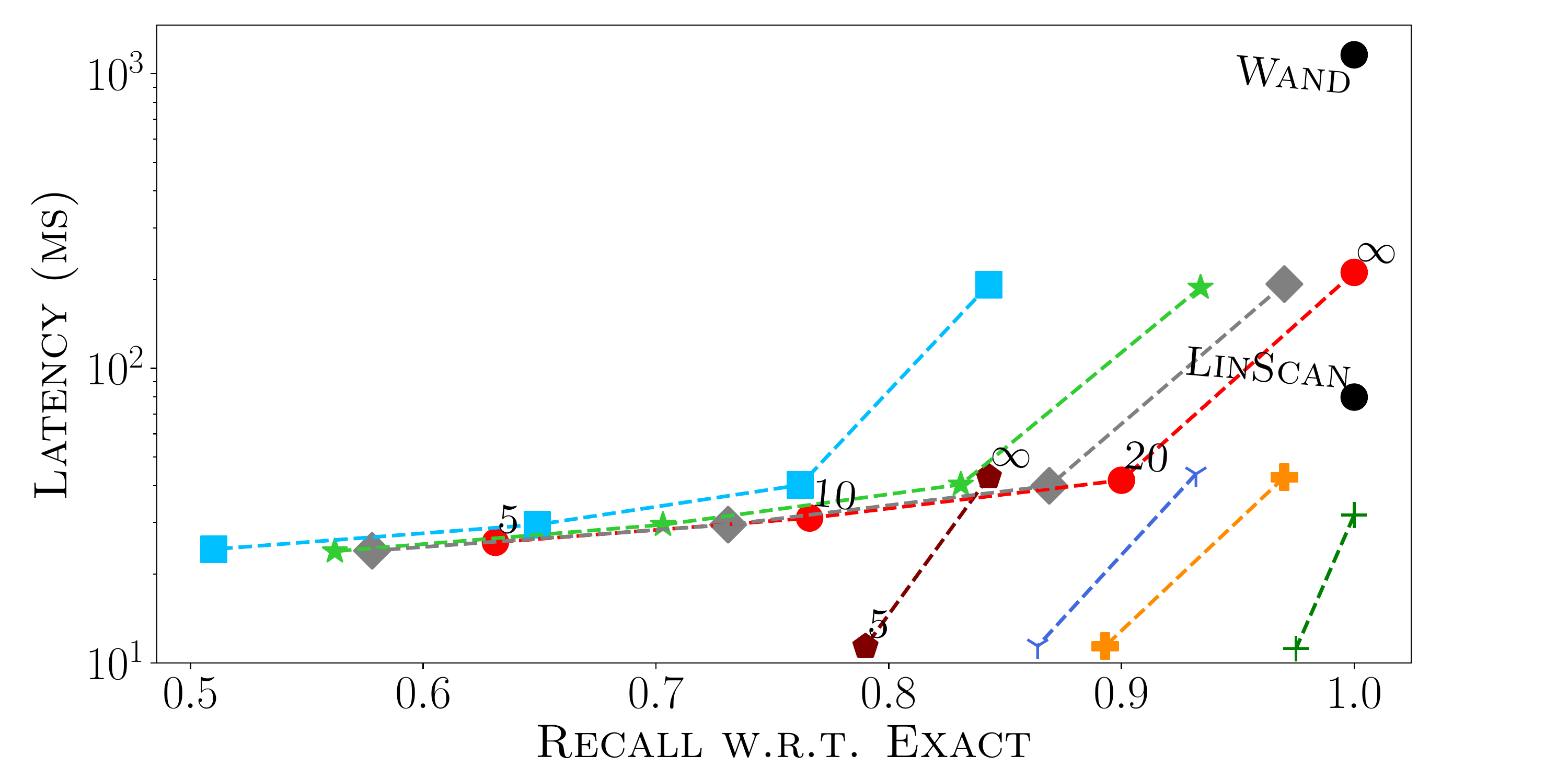}}
}

\centerline{
\subfloat[Efficient SPLADE]{
\includegraphics[trim={0 0 0 5.2cm},clip,width=0.45\linewidth]{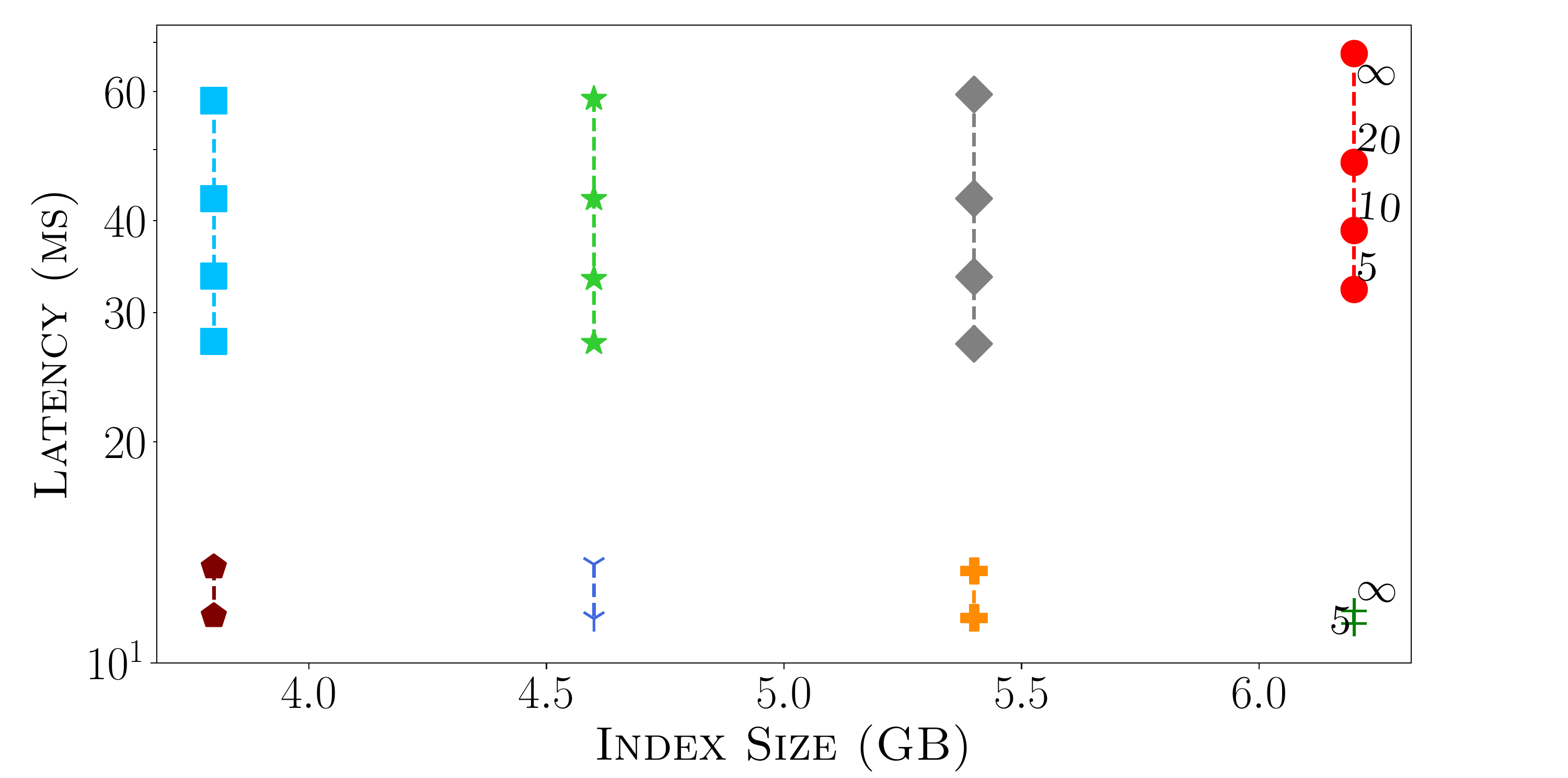}
\includegraphics[trim={0 0 0 5.2cm},clip,width=0.45\linewidth]{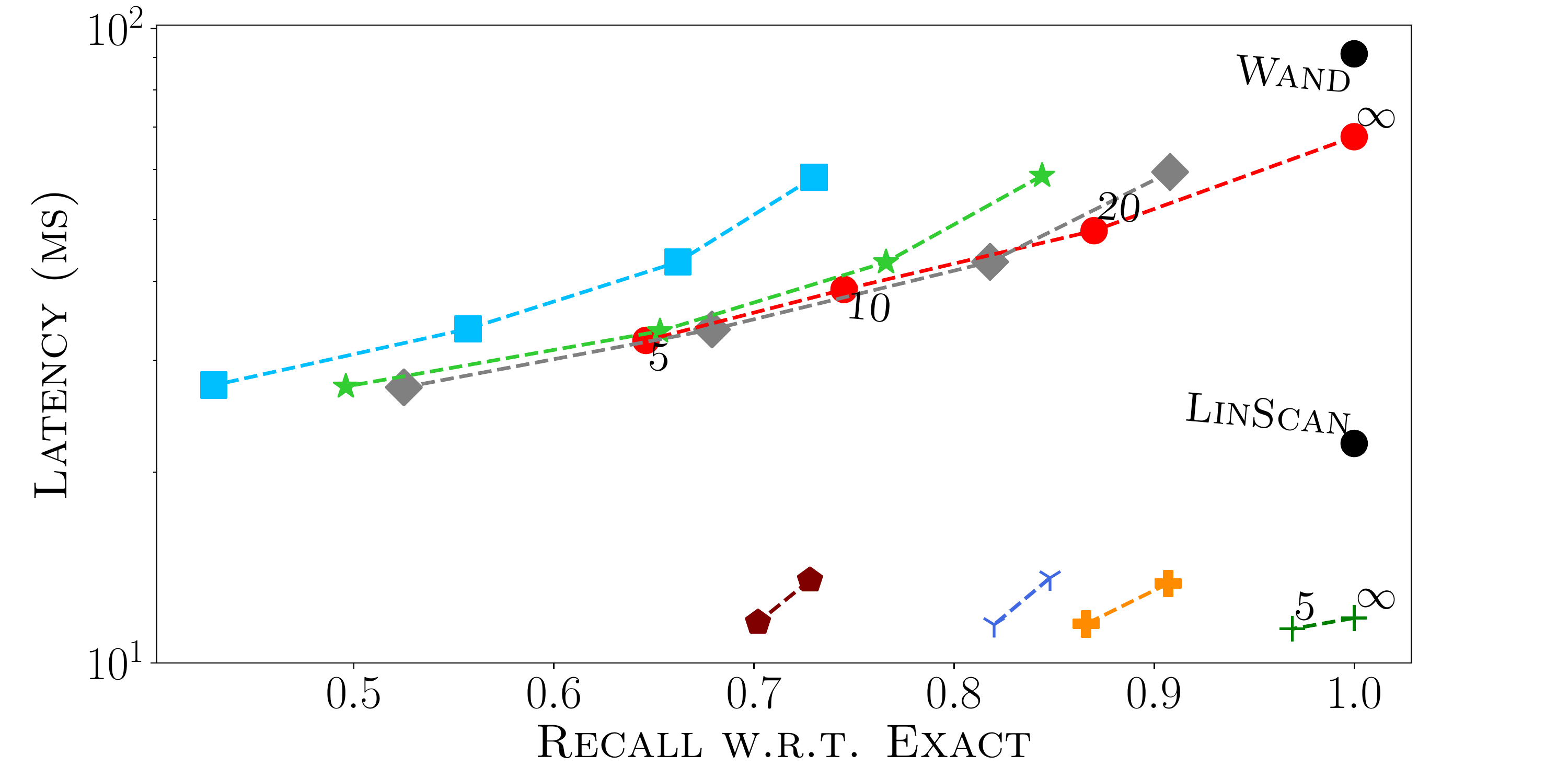}}
}

\centerline{
\subfloat[uniCOIL]{
\includegraphics[trim={0 0 0 5.2cm},clip,width=0.45\linewidth]{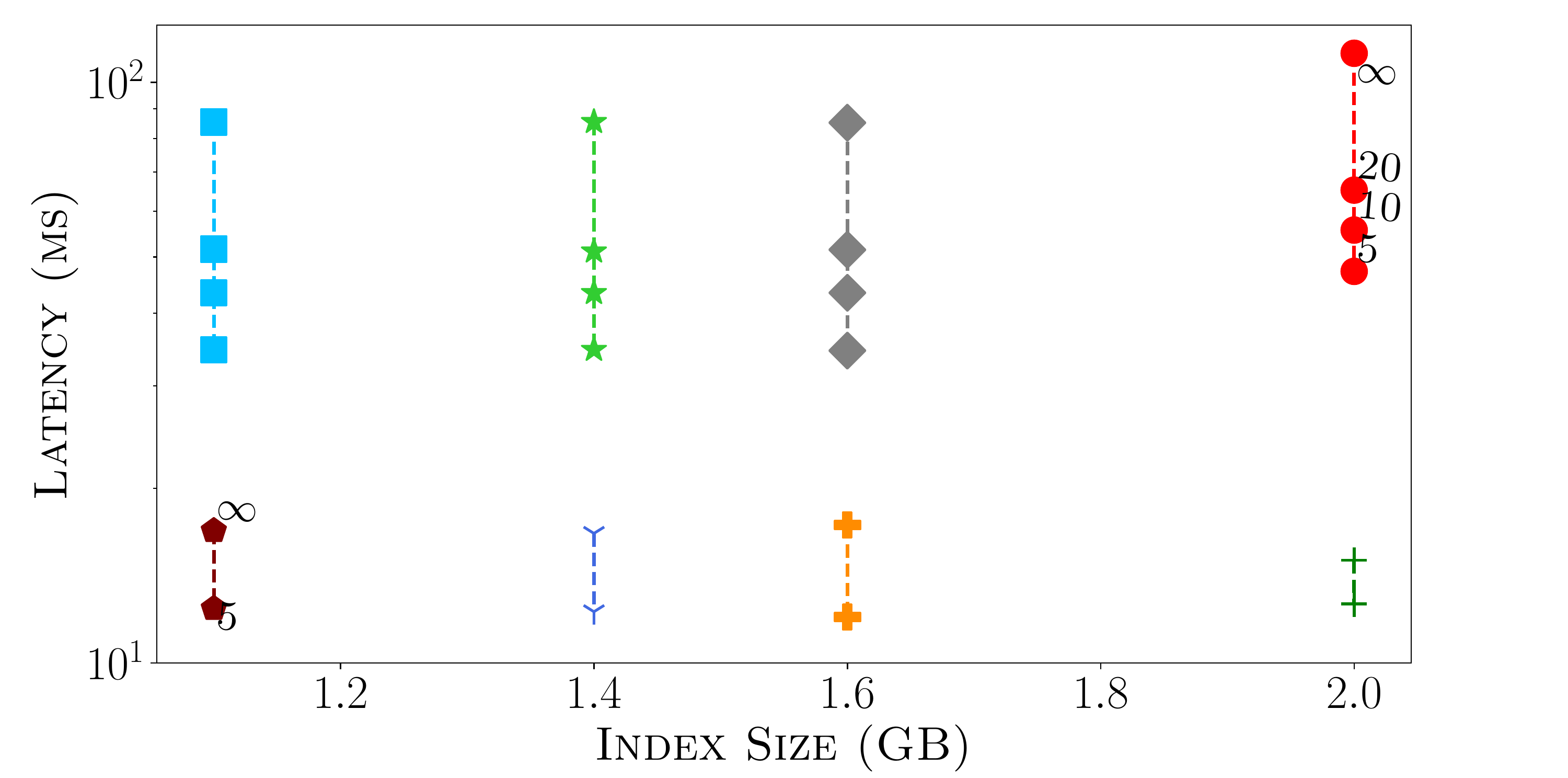}
\includegraphics[trim={0 0 0 5.2cm},clip,width=0.45\linewidth]{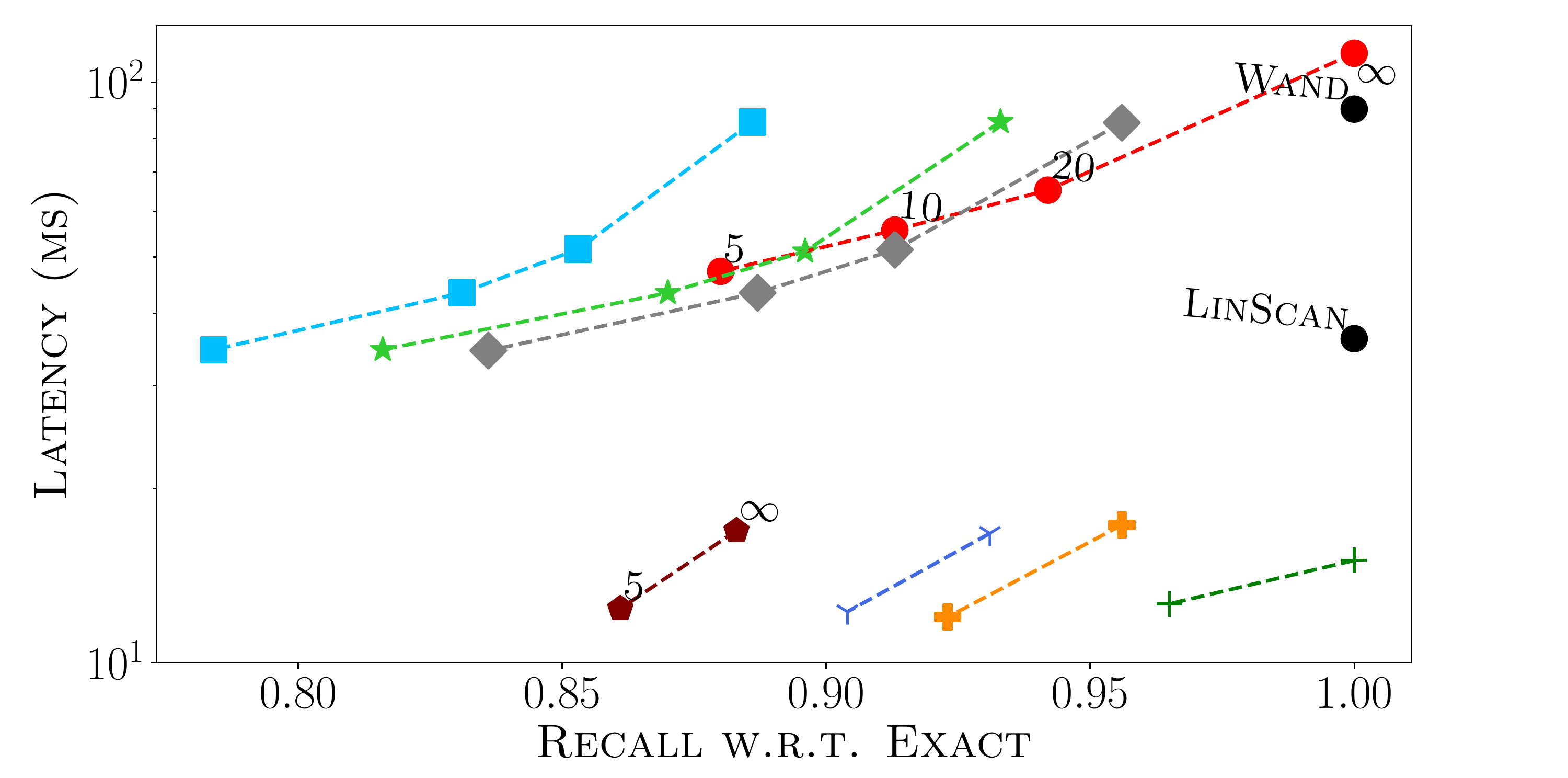}}
}

\caption{Trade-offs on the Intel processor between latency and memory (left column), and latency and accuracy (right column) for various vector collections (rows). Shapes (and colors) distinguish between different configurations of \sinnamon{}, and points on a line represent different time budgets $T$ (in milliseconds).}
\label{figure:evaluation:msmarco-passage-v1-icelake}
\end{center}
\end{figure}

\FloatBarrier

Moving on to the parallel algorithms $\roaringlinscan{}^\parallel$ and $\sinnamon{}^\parallel$, we observe generally strong performance in terms of latency---we will study speed-up in the number of threads later in this section. This enormous gain in latency that is achievable by straightforward parallelization while keeping the index monolithic is, as we argued before, one of the stronger properties of the two algorithms. While the curves for \sinnamon{} may suggest small differences in retrieval accuracy between the mono-CPU and parallel versions, we note that the changes are due to the inherent randomness of the algorithm but that they are not statistically significant according to a paired two-tailed $t$-test with $p$-value $<0.01$.

\subsubsection{Effect on End Metrics}
We have so far established that a smaller $m$ in \sinnamon{} and a reduced scoring time budget $T$ in \sinnamon{} and \roaringlinscan{} lead to lower retrieval accuracy, but how does the loss in retrieval accuracy translate to a loss in the task-specific accuracy? For MS MARCO, we measure accuracy as NDCG$@1000$ and MRR$@10$ and present the results on the Intel platform in Figure~\ref{figure:evaluation:msmarco-passage-v1-icelake-end-metric}, with equivalent figures for M1 in Appendix~\ref{appendix:retrieval-end-metric-m1}.

The pattern that emerges is that, with the exception of BM25, the end metrics are often affected by a reduction in $m$ and $T$. Interestingly the drop in \sinnamon{}'s performance from $T=\infty$ to $T=20$ is not statistically significant on the SPLADE collection according to a paired two-tailed $t$-test, though nonetheless we do observe a reduction in quality as a general trend. This effect is less severe on the M1 chip as the scoring phase runs faster to begin with, and imposing a tighter limit on the time budget does not degrade quality as substantially as it does on the Intel chip. We again note the impressive performance of $\sinnamon{}^\parallel$ across collections.

\begin{figure}[!ht]
\begin{center}
\centerline{
\subfloat[BM25]{
\includegraphics[width=0.45\linewidth]{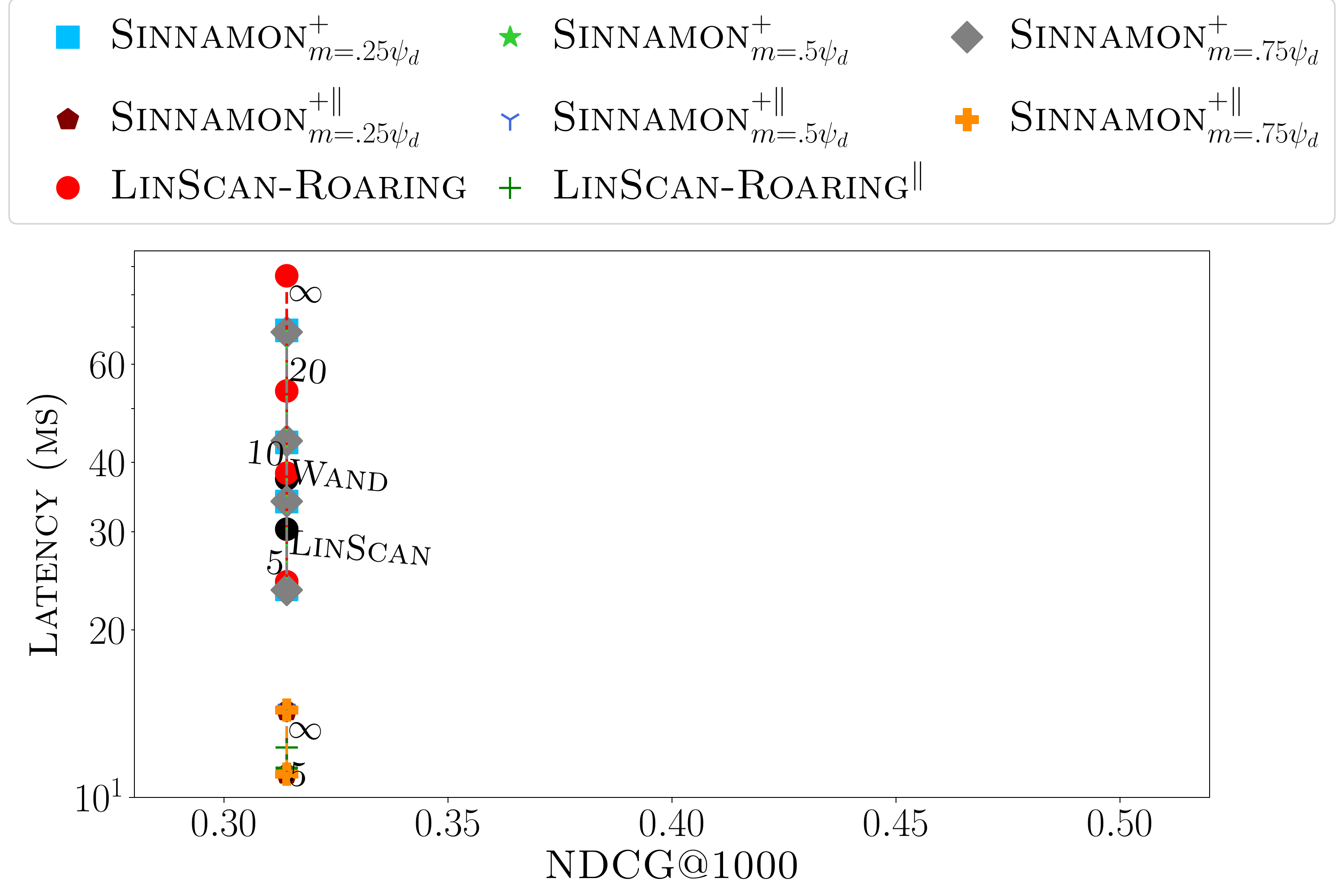}
\includegraphics[width=0.45\linewidth]{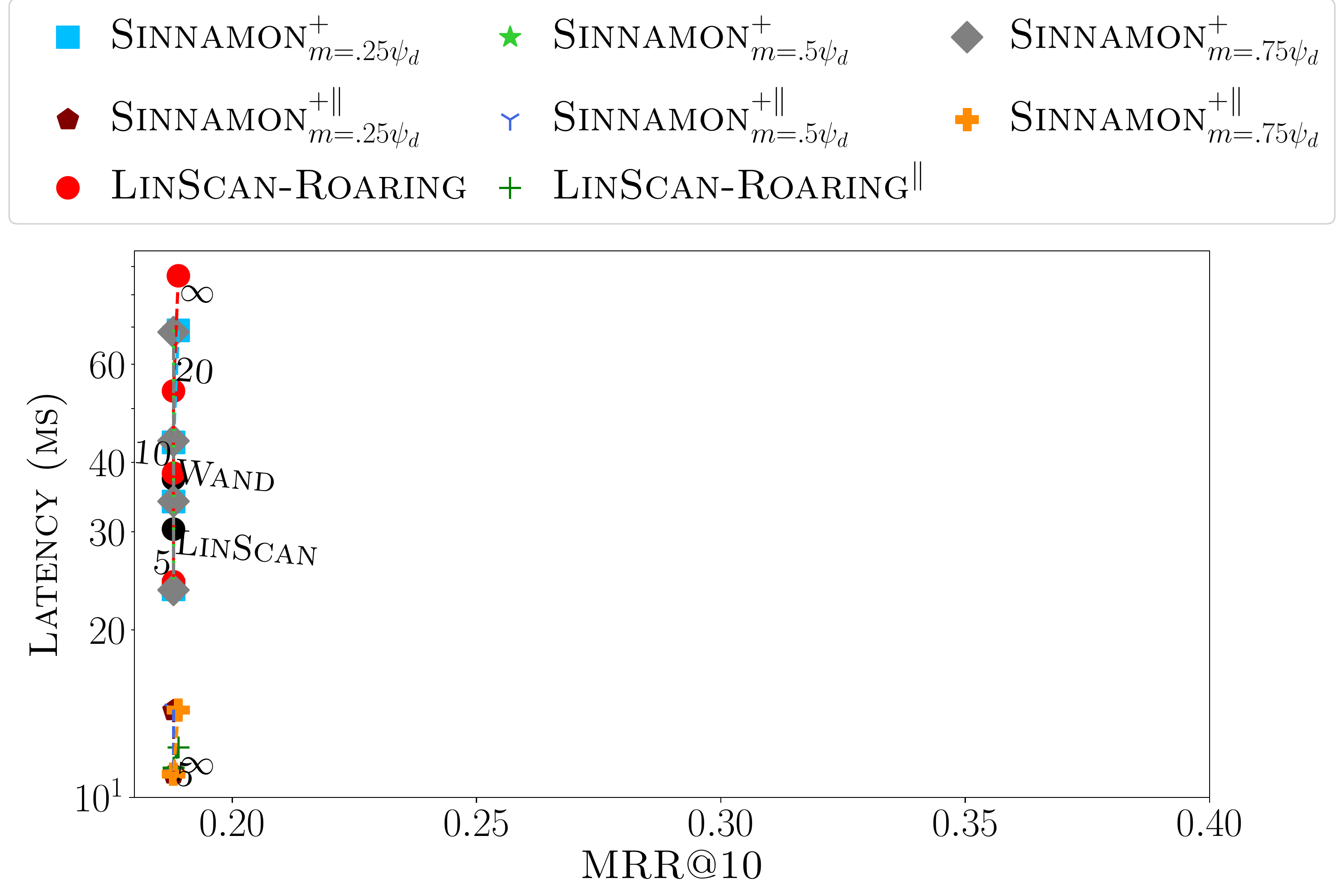}}
}

\centerline{
\subfloat[SPLADE]{
\includegraphics[trim={0 0 0 5.2cm},clip,width=0.45\linewidth]{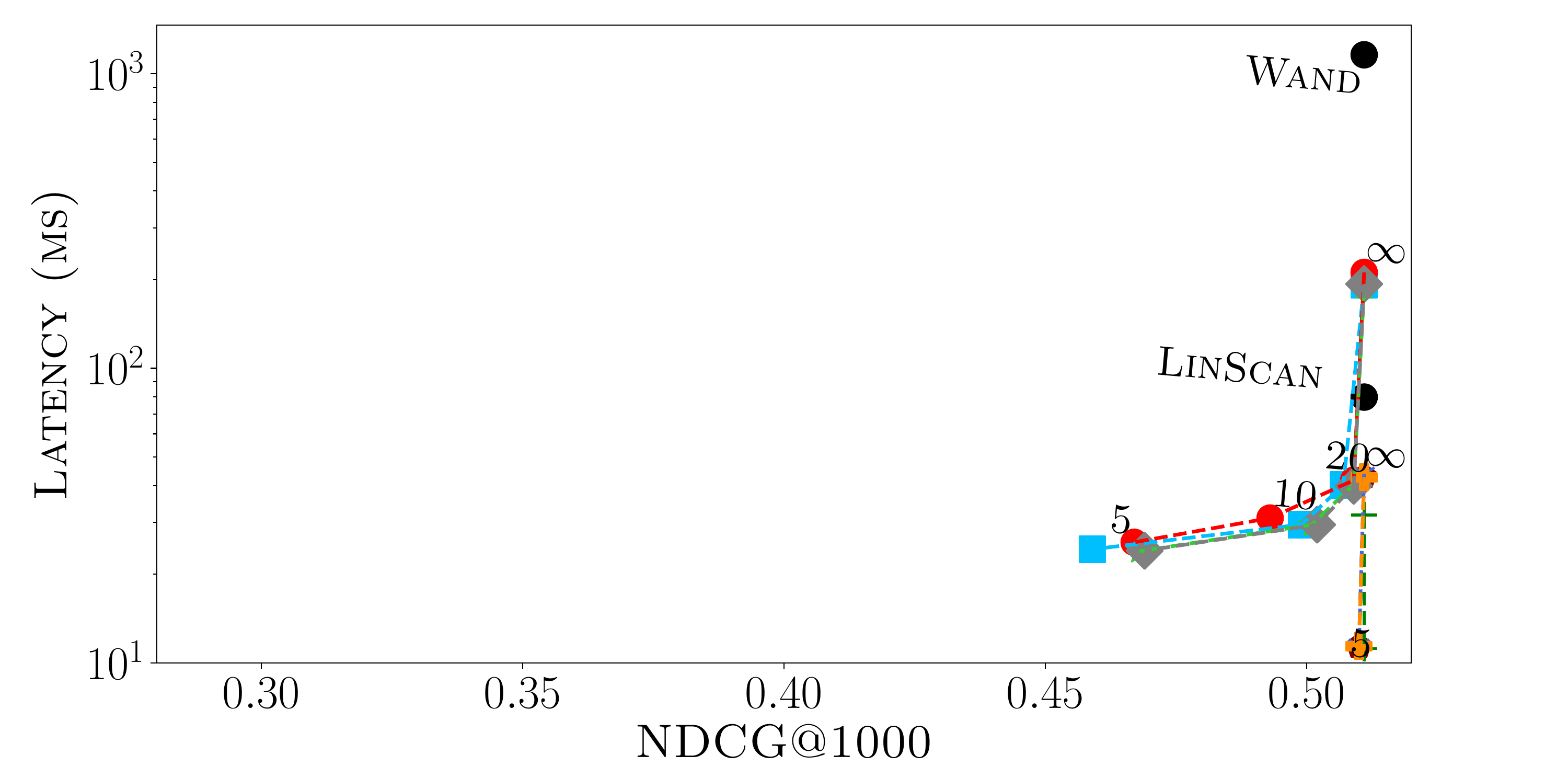}
\includegraphics[trim={0 0 0 5.2cm},clip,width=0.45\linewidth]{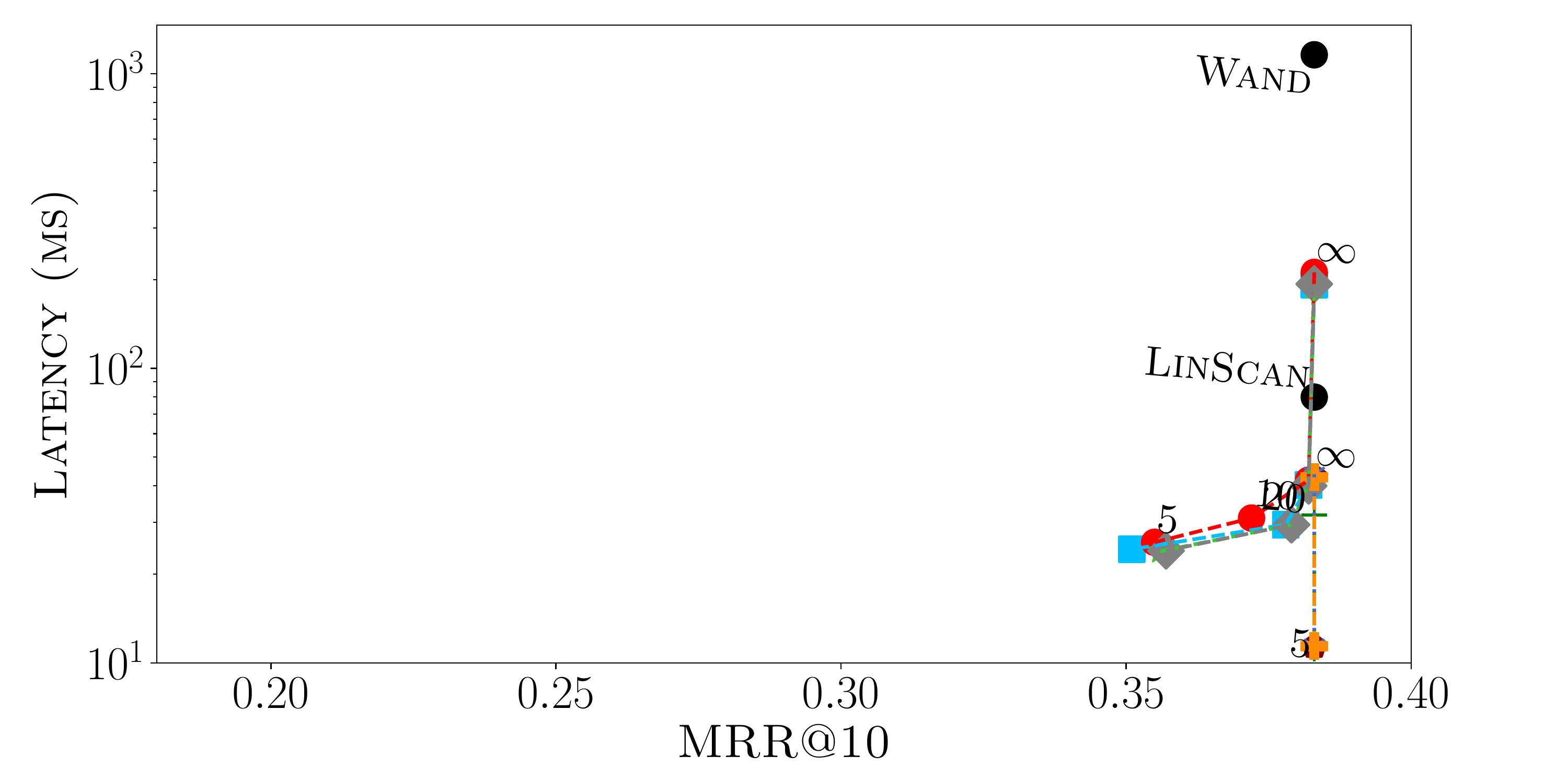}}
}

\phantomcaption
\end{center}
\end{figure}
\begin{figure}[!ht]
\begin{center}
\ContinuedFloat

% \end{center}
% \end{figure}
% \begin{figure}[!ht]
% \begin{center}
% \ContinuedFloat

\centerline{
\subfloat[Efficient SPLADE]{
\includegraphics[trim={0 0 0 5.2cm},clip,width=0.45\linewidth]{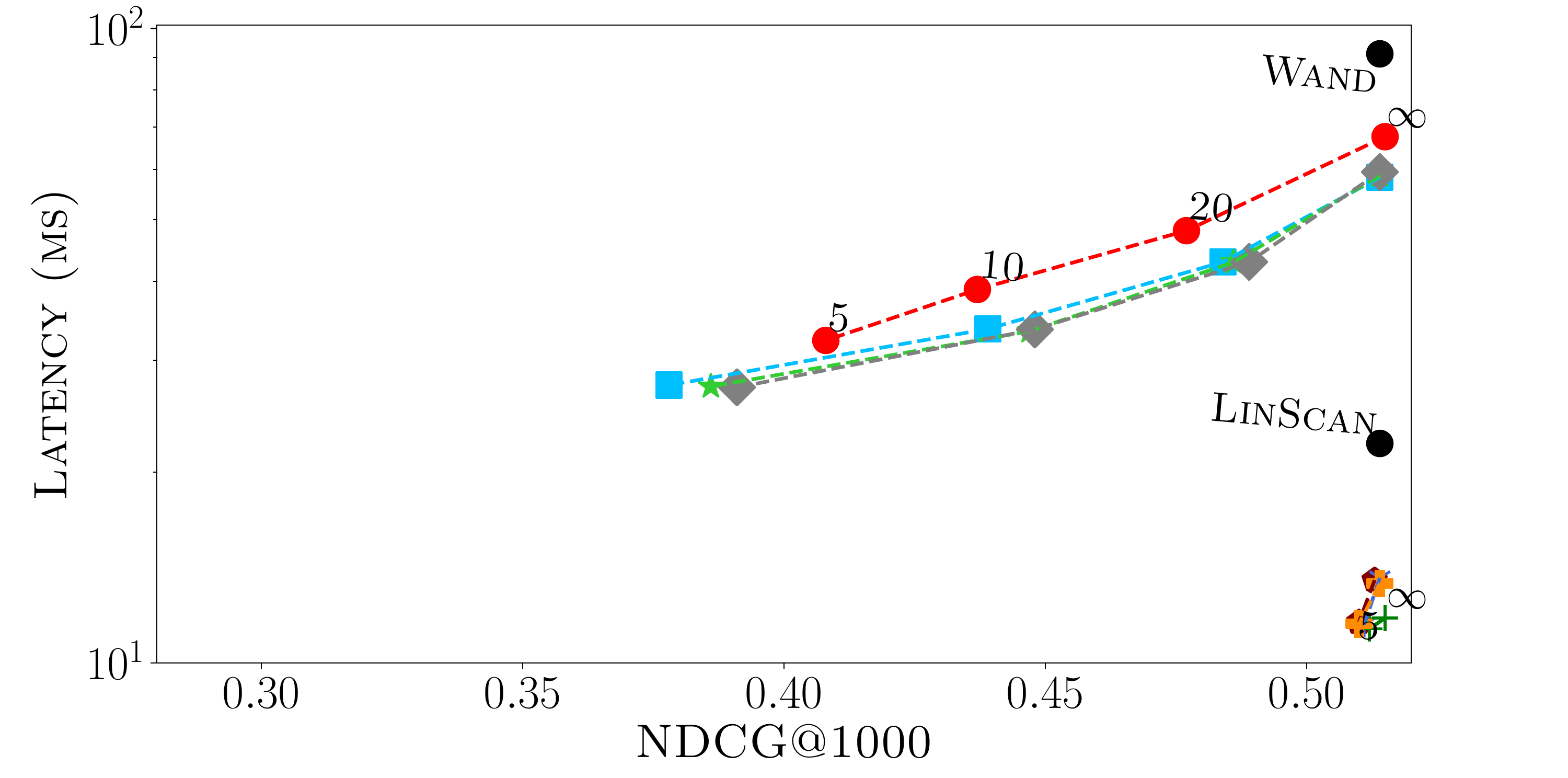}
\includegraphics[trim={0 0 0 5.2cm},clip,width=0.45\linewidth]{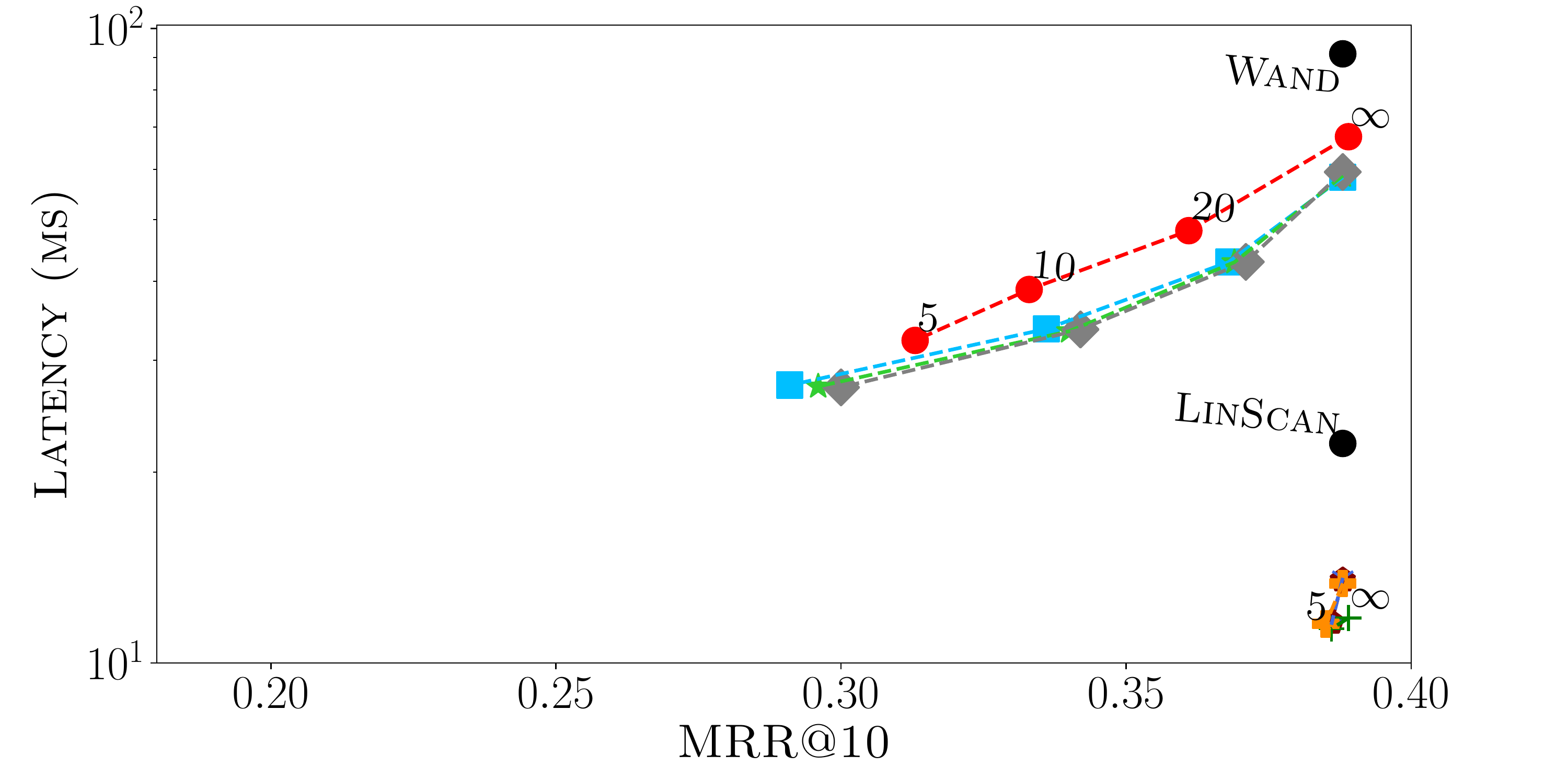}}
}

% \end{center}
% \end{figure}
% \begin{figure}[!ht]
% \begin{center}
% \ContinuedFloat

\centerline{
\subfloat[uniCOIL]{
\includegraphics[trim={0 0 0 5.2cm},clip,width=0.45\linewidth]{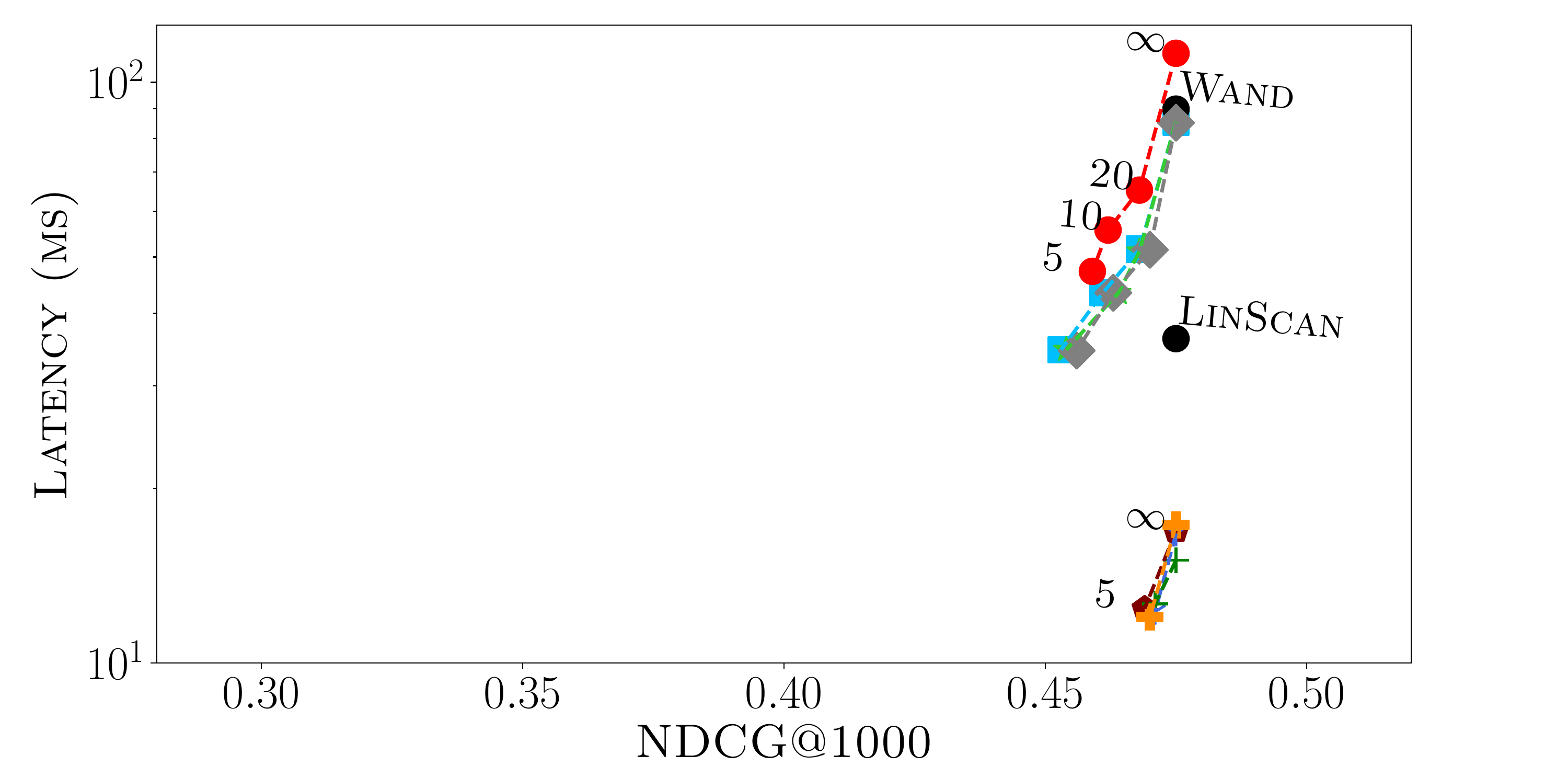}
\includegraphics[trim={0 0 0 5.2cm},clip,width=0.45\linewidth]{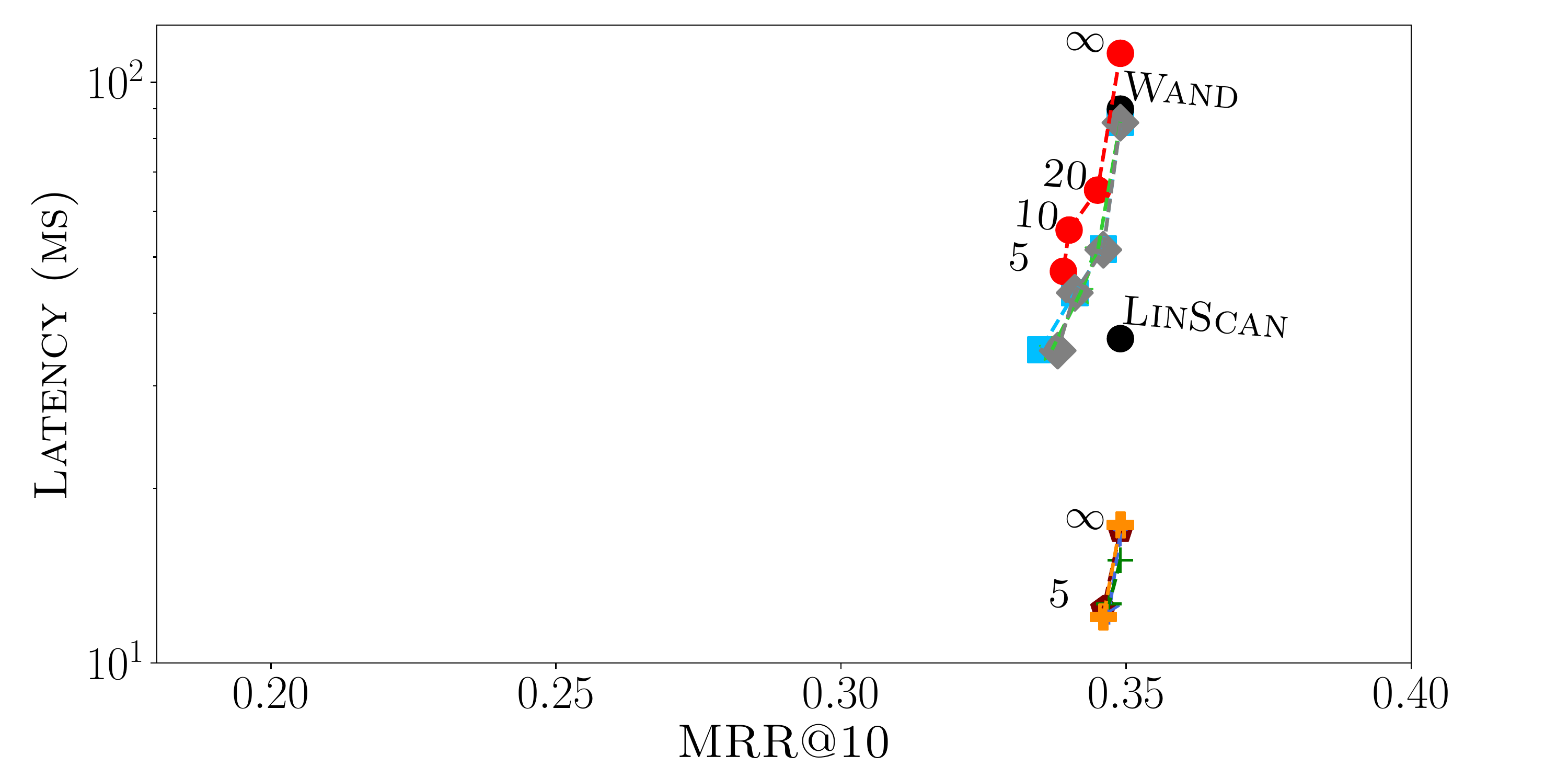}}
}

\caption{Trade-offs on the Intel processor between latency and NDCG$@1000$ (left column), and latency and MRR$@10$ (right column) for various vector collections (rows). As before, shapes (and colors) distinguish between different configurations of \sinnamon{}, and points on a line represent different time budgets $T$ (in milliseconds).}
\label{figure:evaluation:msmarco-passage-v1-icelake-end-metric}
\end{center}
\end{figure}
\FloatBarrier

\begin{figure}[th]
\begin{center}
\centerline{
\subfloat[$m \approx 25\% \psi_d$]{
\includegraphics[width=0.5\linewidth]{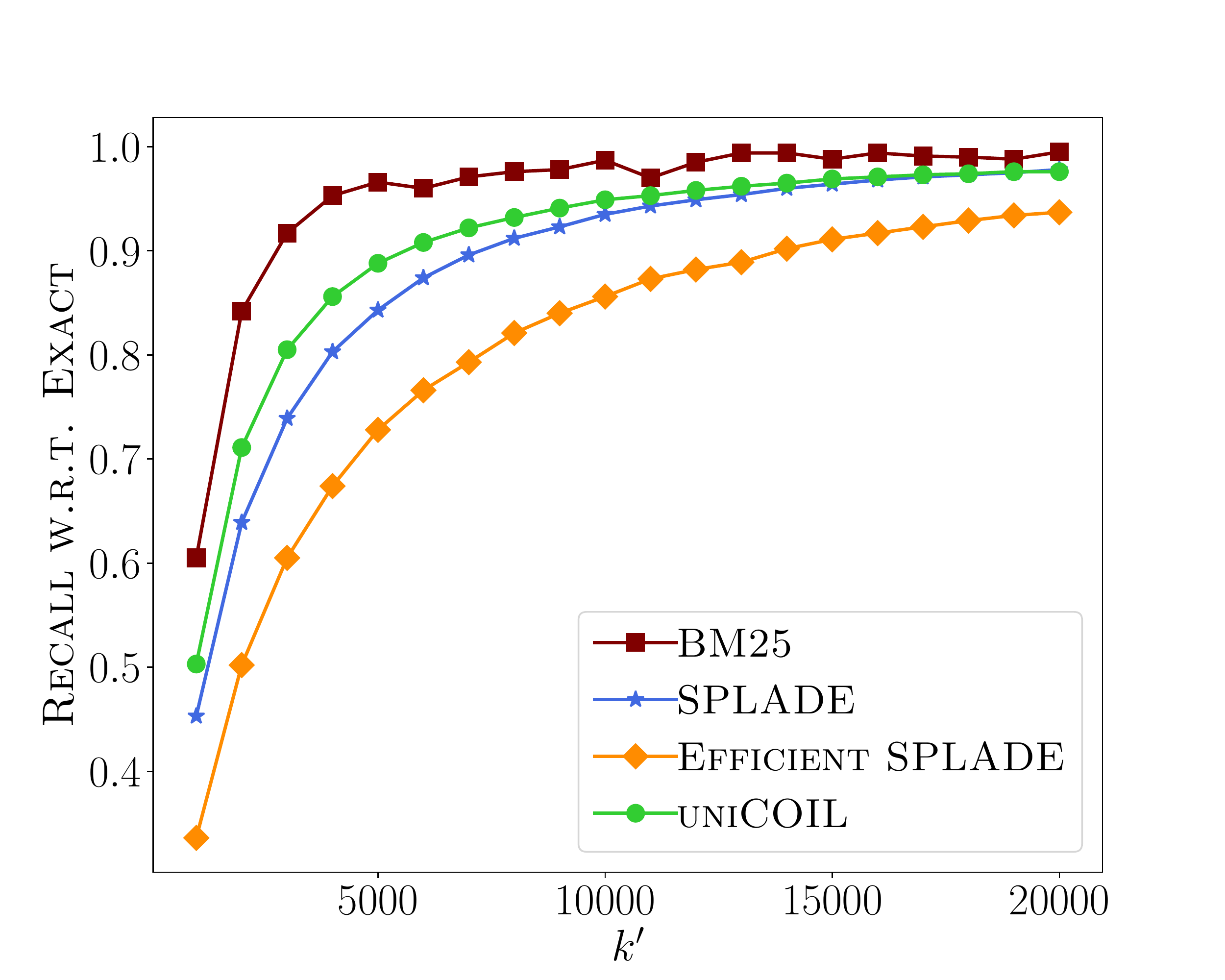}}
\subfloat[$m \approx 75\% \psi_d$]{
\includegraphics[width=0.5\linewidth]{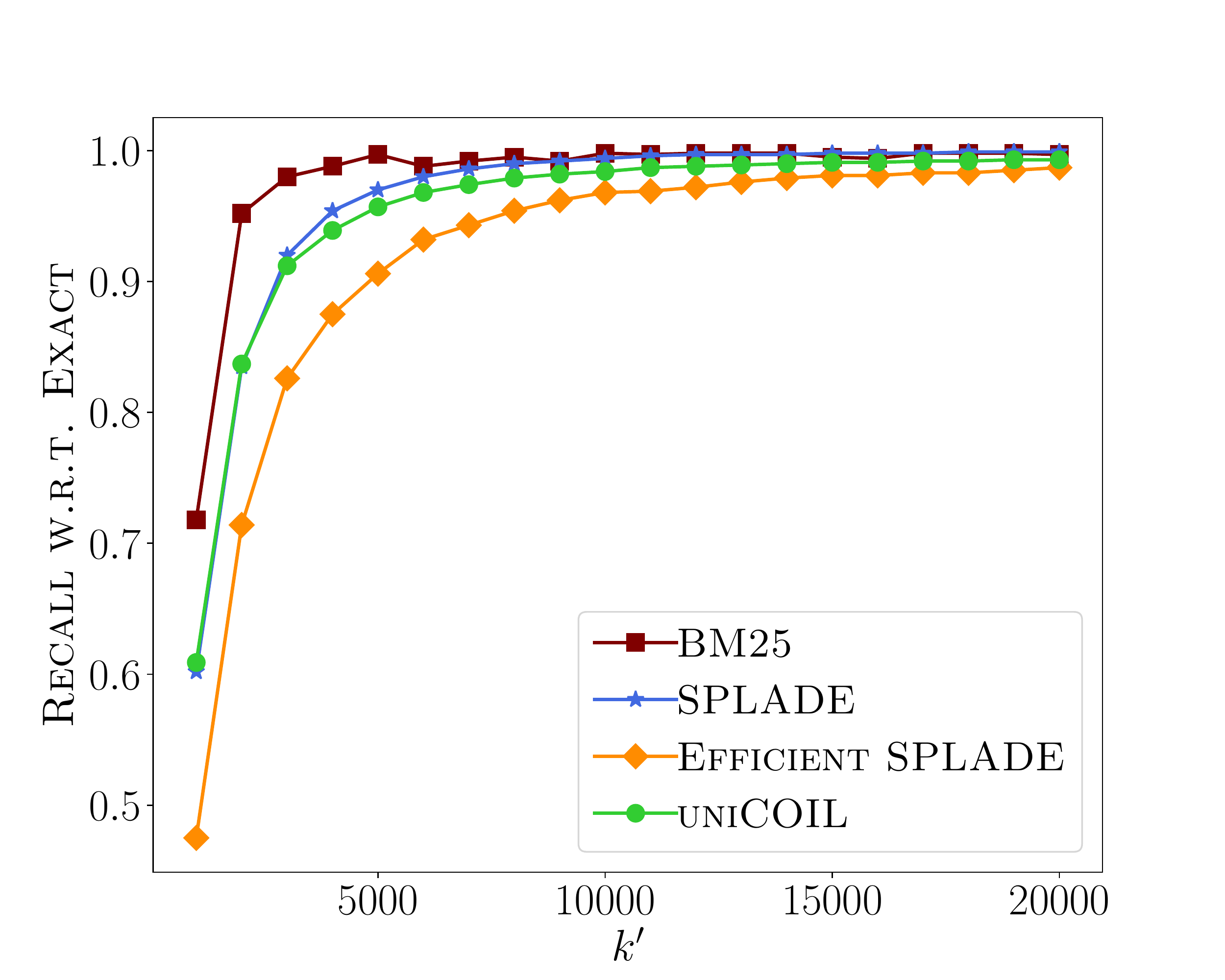}
}}

\caption{The effect of retrieving $k^\prime$ vectors in Algorithm~\ref{algorithm:retrieval:ranking} on retrieval accuracy in terms of recall with respect to exact retrieval. In the experiments leading to the above figure, we use a value of $m$ that is roughly $25\%$ and $75\%$ of the average $\psi_d$ in the datasets ($10$/$30$ for BM25, $30$/$90$ for SPLADE, $45$/$135$ for Efficient SPLADE, and $15$/$45$ for uniCOIL).}
\label{figure:evaluation:kprime-sweep}
\end{center}
\end{figure}

\subsubsection{Effect of $k^\prime$}
The experiments so far used a fixed value for the intermediate number of candidates $k^\prime$. We ask now if and to what degree changing this hyperparameter affects the trade-offs and shapes the interplay between the three factors. To study this effect, we choose a single value of $m$ for each dataset and measure retrieval accuracy as $k^\prime$ grows from $1,000$ to $20,000$. We expect retrieval accuracy to increase as more candidates are re-ranked by Algorithm~\ref{algorithm:retrieval:ranking}.

As we show in Figure~\ref{figure:evaluation:kprime-sweep}, this is indeed the effect we observe on our vector collections. In this figure, we retrieve $k^\prime$ vectors for two values of $m$ that are $25\%$ and $75\%$ of the $\psi_d$ of a dataset. As $k^\prime$ increases, so does retrieval accuracy. We must note, however, that a larger $k^\prime$ adversely affects overall latency as more documents must be re-ranked using exact scores, which results in a larger number of fetches from the raw vector storage. However, this effect can be amortized over multiple processors in the $\sinnamon{}^\parallel$ variant.

\begin{figure}[th!]
\begin{center}
\centerline{
\subfloat[BM25]{
\includegraphics[width=0.45\linewidth]{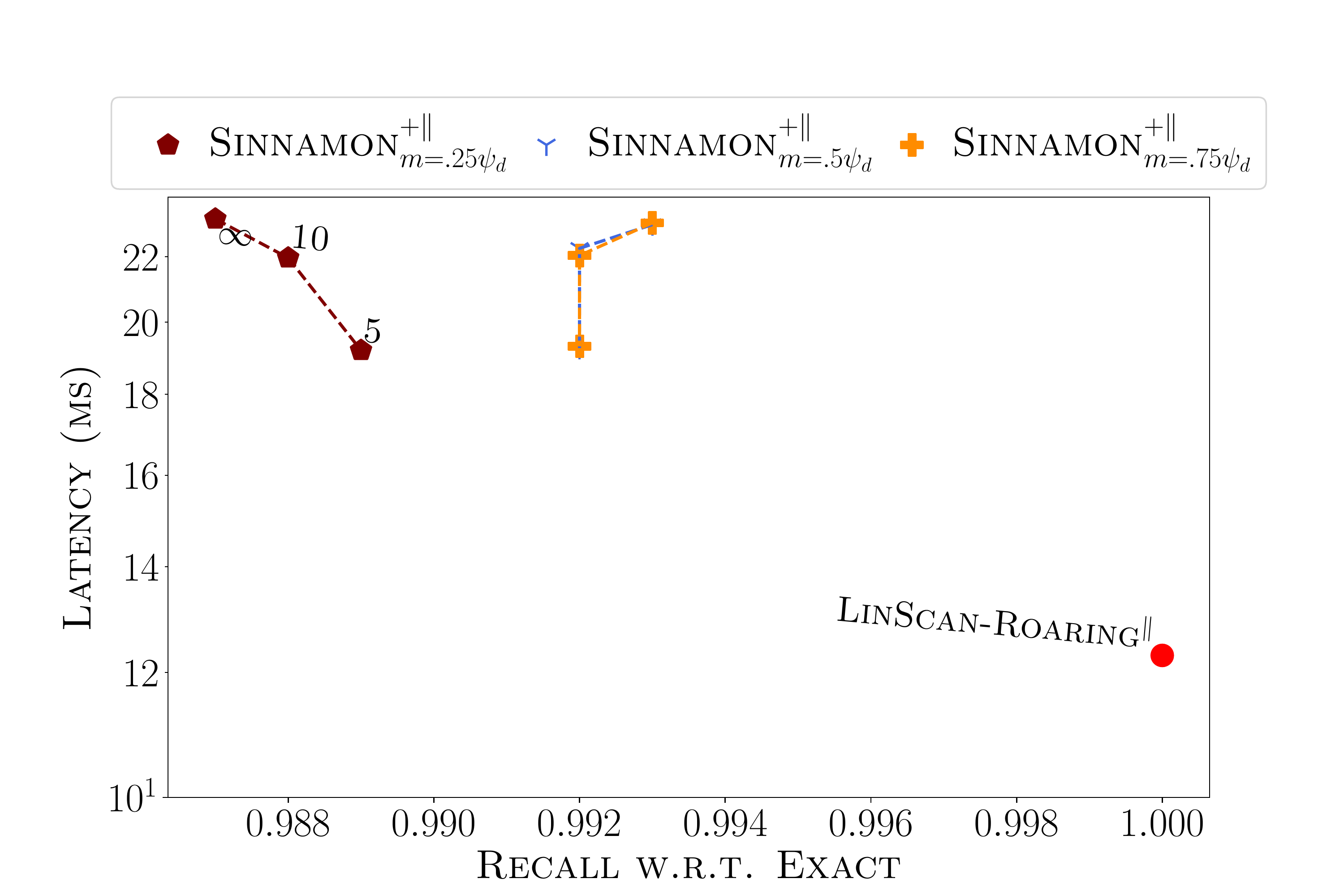}}
\subfloat[SPLADE]{
\includegraphics[width=0.45\linewidth]{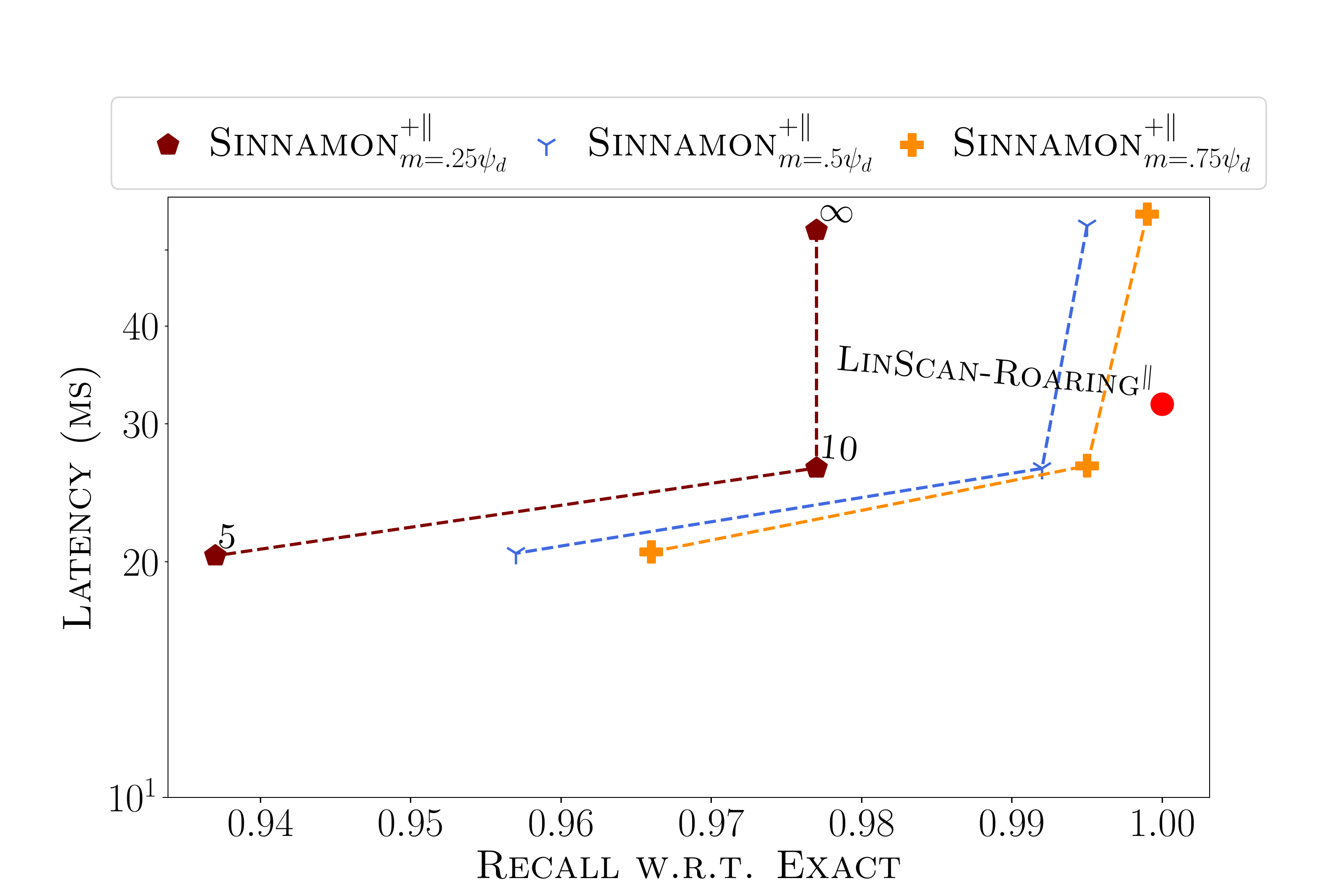}}
}

\centerline{
\subfloat[Efficient SPLADE]{
\includegraphics[width=0.45\linewidth]{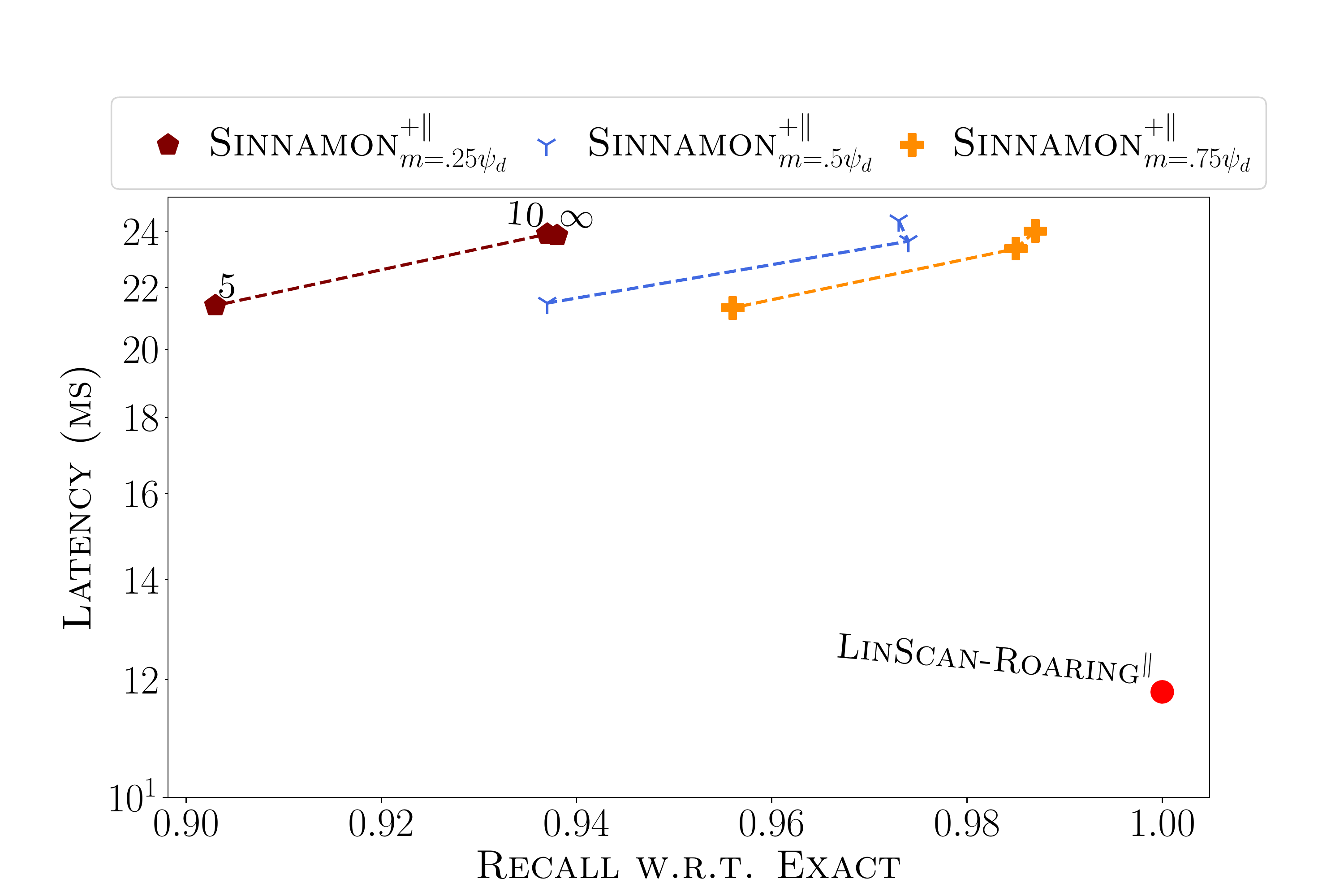}}
\subfloat[uniCOIL]{
\includegraphics[width=0.45\linewidth]{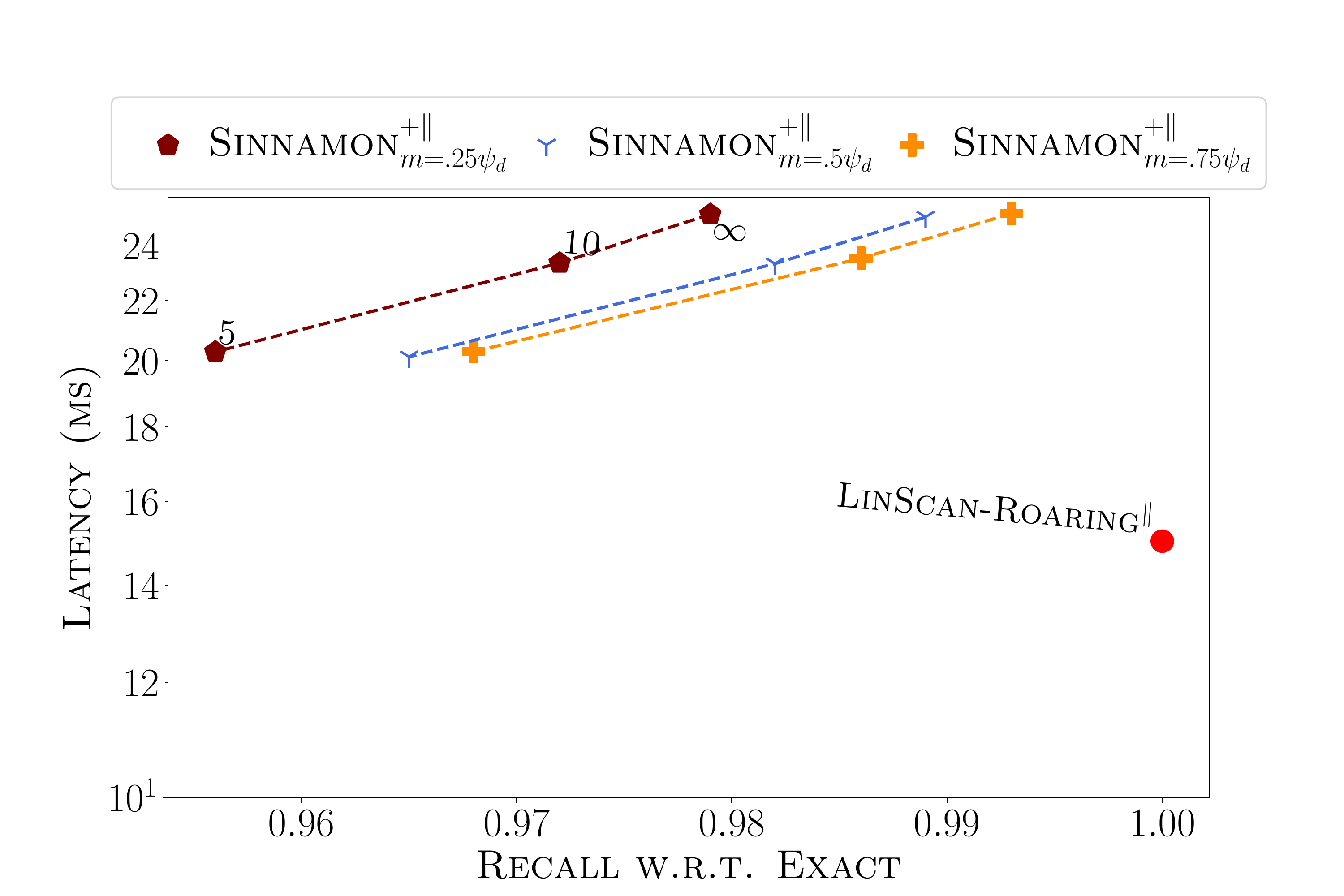}}
}

\caption{Trade-offs on the Intel processor between latency and retrieval accuracy for various vector collections by retrieving $k^\prime=20,000$. Shapes (and colors) distinguish between different configurations of \sinnamon{}, and points on a line represent different time budgets $T$ (in milliseconds).}
\label{figure:evaluation:msmarco-passage-v1-icelake-20k}
\end{center}
\end{figure}

Having made the observations above, we repeat the experiments presented earlier in this section for $k^\prime=20,000$ and visualize the trade-offs between latency and retrieval accuracy once more---index size remains the same as in Figure~\ref{figure:evaluation:msmarco-passage-v1-icelake}, so we leave it out. This time, we limit the experiments to $\sinnamon{}^\parallel$ as it is a more practical choice than a mono-CPU variant, following the note above. These results are presented in Figure~\ref{figure:evaluation:msmarco-passage-v1-icelake-20k} for the Intel processor with M1 results presented in Appendix~\ref{appendix:retrieval-m1-20k}.

While retrieval accuracy appears to improve across the board, including the anytime variants, with an increase in $k^\prime$, \sinnamon{}'s latency remains the same or degrades comparative to the $k^\prime=5,000$ setup before. This increase in cost is primarily driven by the exact score computation as expected. We observe that \sinnamon{}'s anytime variants perform faster than $\roaringlinscan{}^\parallel$ and with a high accuracy on the SPLADE collection. This is an interesting result if one notes that the differentiating factor between SPLADE and the other collections is SPLADE's relatively larger $\psi_q$, hinting that \sinnamon{} may indeed be a competitive algorithm when query vectors have a large number of non-zero entries. Finally, we note that the increase in quality as a result of using a stricter time budget in the BM25 plot is statistically insignificant and can be attributed to noise.

\subsection{Insertions and Deletions}
We have argued that \sinnamon{} is an online algorithm in the sense that it is capable of indexing new documents and deleting existing documents with a low-enough latency that the algorithm can function robustly in a streaming setting, where vectors change rapidly. In this section, we put that property to the test and examine the response time of the algorithm to updates to its index.

\begin{figure}[th!]
\begin{center}
\centerline{
\subfloat[Insertions]{
\includegraphics[width=0.45\linewidth]{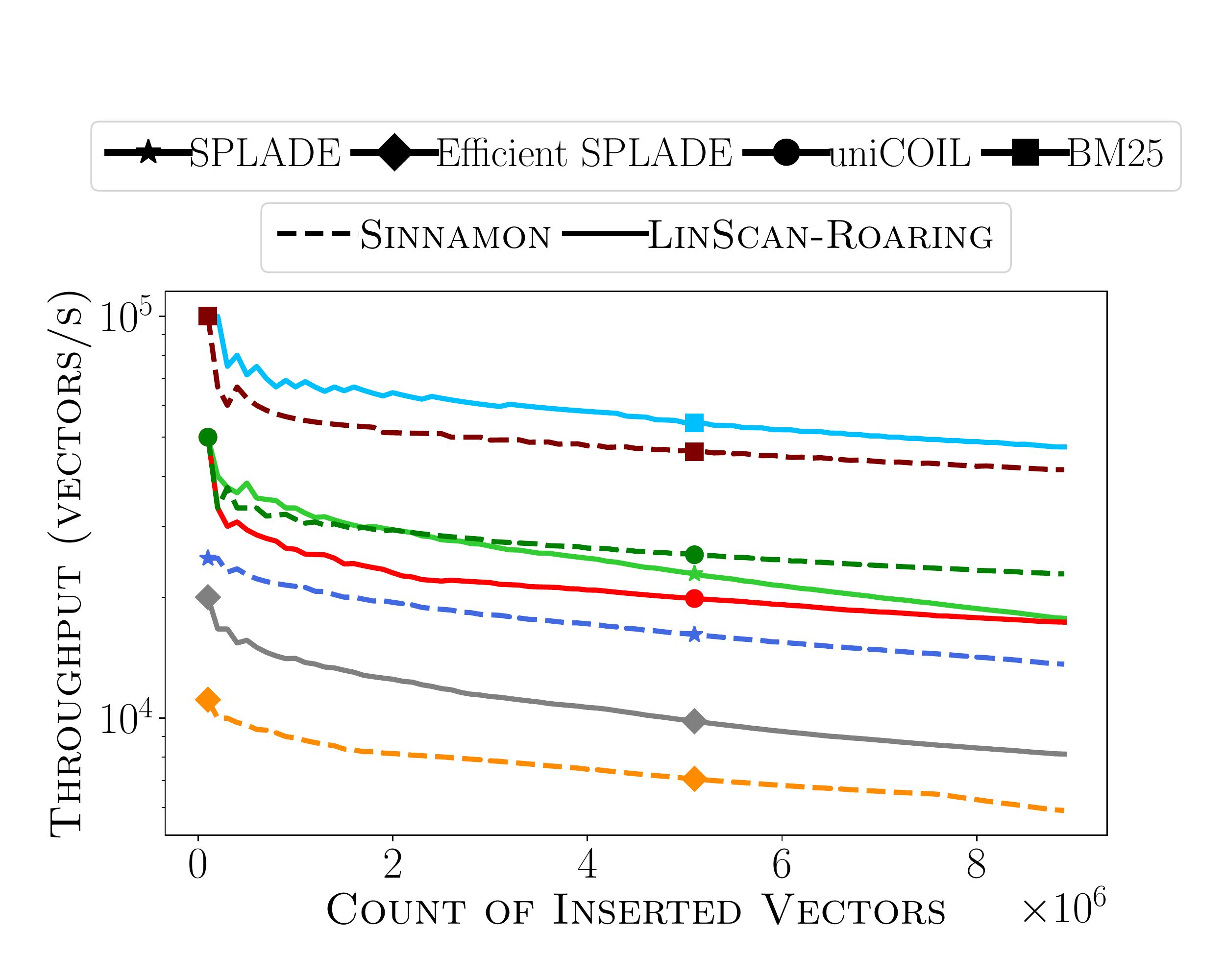}\label{figure:evaluation:msmarco-passage-v1-icelake-benchmark:insertions}}
\subfloat[Deletions]{
\includegraphics[width=0.45\linewidth]{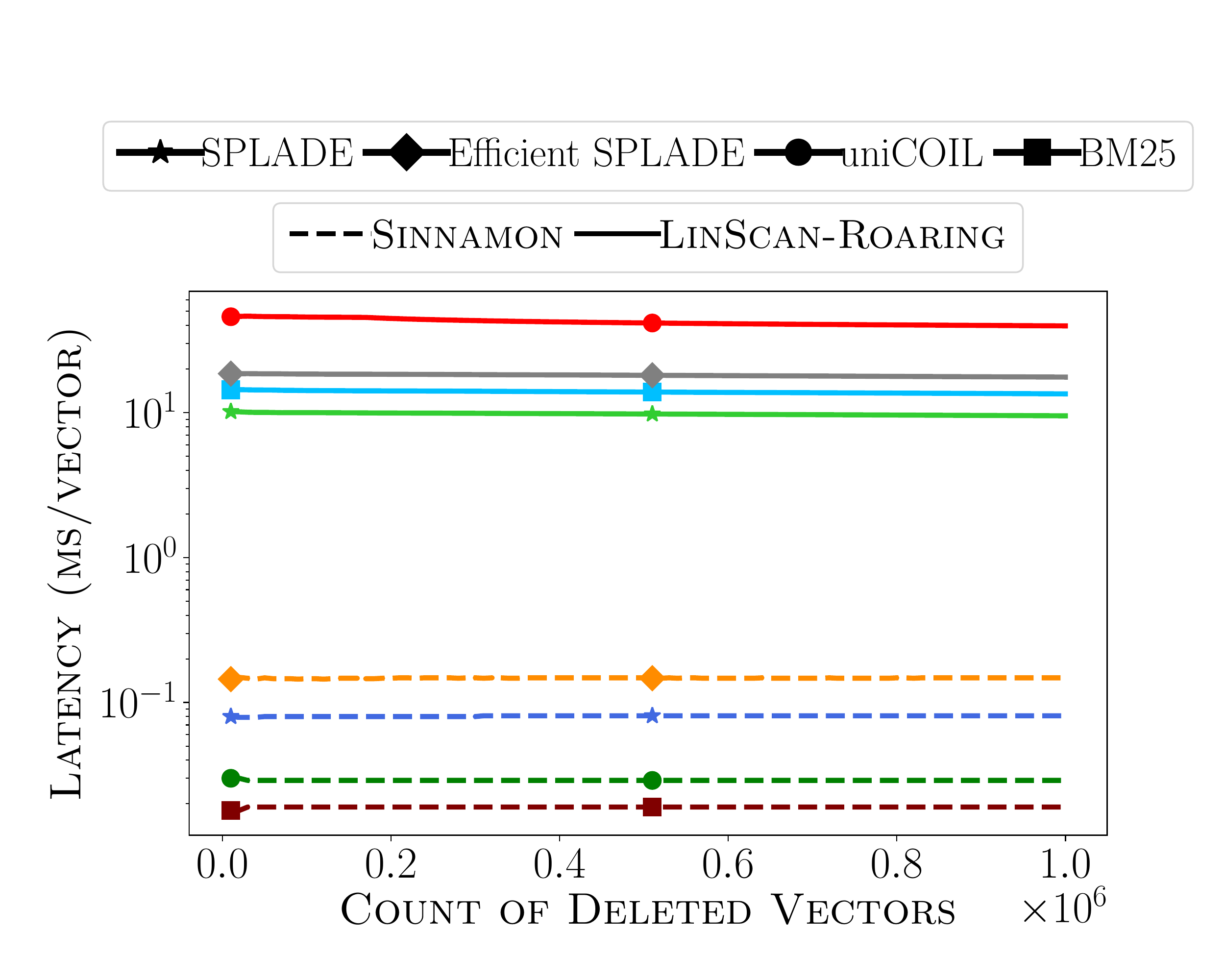}\label{figure:evaluation:msmarco-passage-v1-icelake-benchmark:deletions}}
}
\caption{Indexing throughput (vectors per second) and deletion latency (milliseconds per vector) as a function of the size of the index on the Intel processor.}
\label{figure:evaluation:msmarco-passage-v1-icelake-benchmark}
\end{center}
\end{figure}

To study this, we design the following experiment. We shuffle the vectors in a collection randomly and insert them into the index sequentially using a single thread. At the end of every insertion, we measure and record the elapsed time. We repeat this experiment $10$ times and measure mean throughput throughout the life of the index---that is, from when it is empty to the moment all vectors have been indexed. The results from these experiments are shown in Figure~\subref*{figure:evaluation:msmarco-passage-v1-icelake-benchmark:insertions} for the Intel processor with M1 results illustrated in Appendix~\ref{appendix:benchmark-m1}.

In this figure, we also include \roaringlinscan{} as a strong baseline: Because \roaringlinscan{} simply adds one posting to an inverted list for every non-zero coordinate in a document, it represents an insertion skyline. We observe that, as the index grows the insertion throughput decreases. We also note that \sinnamon{}'s indexing procedure is often slower than insertions in \roaringlinscan{} on every vector collection, but that its throughput---is nonetheless large enough to make the algorithm suitable for indexing large volumes of data online in most applications.

In addition to indexing, we also measure deletion latency in Figure~\subref*{figure:evaluation:msmarco-passage-v1-icelake-benchmark:deletions} for various collections on the Intel processor. To produce this figure, we start with an index containing all vectors in a collection and subsequently delete $1$ million randomly-selected vectors from the index. We measure the latency of each deletion in milliseconds and report the means of $10$ trials in this figure. Both algorithms show a stable latency across datasets, with a negligible speed-up as more vectors are deleted. But we observe that \sinnamon{} deletes vectors with a response time that is about an order of magnitude faster than that in \roaringlinscan{}.

\subsection{Real Vectors}
\label{section:evaluation:synthetic}
As we stated in the beginning of this section, all datasets used in our experiments so far consist of \emph{non-negative} sparse vectors. While we claim theoretically that \sinnamon{} handles real-valued vectors just as well as non-negative vectors, we have not yet demonstrated evidence to support that claim in practice. In this section, we do just that. But because, to the best of our knowledge, no benchmark dataset or embedding model exists that produces real-valued sparse vectors, we resort to generating synthetic data for the experiments in this section.

To that end, we fix the average number of non-zero coordinates $\psi_d = \psi_q = \psi$, and the dimensionality of the space $n$ as hyperparameters. To generate a document or query vector, we decide if a coordinate is non-zero with probability $\psi / n$. If it is, we draw its value from the standard normal distribution. We repeat this process to generate five million vectors to form the document set, and $1,000$ vectors to form the query set. We then run a limited set of experiments with a modified variant of WAND that can handle real values, \roaringlinscan{}, and \sinnamon{} to retrieve the top $k=1,000$ documents, where $k^\prime= 20,000$ in \sinnamon{}. We experiment with two configurations: a) $\psi=100$ and $n=10,000$, a dataset we call $G_{100}$; and b) $\psi=200$ and $n=32,000$, a dataset we call $G_{200}$.

Table~\ref{table:evaluation:real-vectors} summarizes the results of our experiments. We observe trends that are consistent with the findings in the preceding sections. Notably, WAND scales poorly when $\psi_q$ is large. \roaringlinscan{} achieves much better latency with a reasonable memory footprint. Finally, \sinnamon{} with $2m=75\% \psi_d$ reduces memory usage, with anytime variants achieving better latency at the cost of accuracy. Overall, these results confirm that \sinnamon{} on real-valued vector datasets offers similar trade-offs as $\sinnamon{}^+$ does on non-negative vectors.

\begin{table*}[t]
\caption{Comparison of different retrieval algorithms on a dataset of $5$ million vectors with $\psi$ non-zero coordinates, drawn from the standard normal distribution. In $G_{100}$, $\psi_q = \psi_d = \psi=100$ and dimensionality $n=10,000$. In $G_{200}$, $\psi=200$ and dimensionality $n=32,000$.}
\label{table:evaluation:real-vectors}
\begin{center}
\begin{sc}
\begin{tabular}{cc|ccc}
& Algorithm & Index Size (GB) & Query Lateny (ms) & Recall w.r.t Exact \\
\toprule
\midrule
\parbox[t]{2mm}{\multirow{4}{*}{\rotatebox[origin=c]{90}{$G_{100}$}}}
& WAND & 5.2 & 2,437 & 1 \\
& \roaringlinscan{}{} & 2.0 & 151 & 1 \\
& \sinnamon{}$_{2m=74,T=40}$ & 1.7 & 87 & 0.97 \\
& \sinnamon{}$_{2m=74,T=20}$ & 1.7 & 66 & 0.89 \\
\midrule
\parbox[t]{2mm}{\multirow{4}{*}{\rotatebox[origin=c]{90}{$G_{200}$}}}
& WAND & 8.4 & 4,738 & 1 \\
& \roaringlinscan{}{} & 4.0 & 182 & 1 \\
& \sinnamon{}$_{2m=150,T=40}$ & 3.5 & 107 & 0.92 \\
& \sinnamon{}$_{2m=150,T=20}$ & 3.5 & 86 & 0.83 \\
\bottomrule
\end{tabular}
\end{sc}
\end{center}
\end{table*}

\subsection{Effect of Parallelism}
We close this section by touching on the more secondary question of how our implementation of \sinnamon{}$^\parallel$ scales with respect to the number of threads. In particular, we measure the speed-up for a given number of CPU cores on the Intel platform, where we define speed-up as the ratio between the latency of a single-threaded \sinnamon{} to its parallel execution. For completeness, we run these experiments on the synthetic vector datasets containing real-valued vectors.

We show the query latency and speed-up figures in Table~\ref{table:evaluation:speed-up}. From these limited experiments, it appears that the speed-up is sub-linear in the number of CPU cores. We attribute the rather large waste to our language of choice: We implemented the parallel version in Rust where memory de-referencing and sharing between processes introduces run-time overhead to the parallel program. In particular, because it proved challenging to bypass the memory safety mechanisms inherent in Rust, we are only able to chunk up the entire sketch matrix, rather than chunking up just the inverted list. As a result, the load-balancing between threads is highly non-uniform, leading to sub-optimal execution time. We are confident that, with better engineering, this gap can be closed and a near-linear growth can be attained.

\begin{table*}[t]
\caption{Query latency (in milliseconds) with speed-up in parentheses (i.e., ratio between single-threaded latency \sinnamon{}$_{T=\infty}$ with \sinnamon{}$_{T=\infty}^\parallel$) on the Intel platform with the indicated number of cores. Experiments are conducted on the synthetic datasets $G_{100}$ and $G_{200}$ of Table~\ref{table:evaluation:real-vectors}.}
\label{table:evaluation:speed-up}
\begin{center}
\begin{sc}
\begin{tabular}{c|cccc}
Dataset & 1 & 2 & 4 & 8 \\
\toprule
\midrule
$G_{100}$ & 174 & 95 (1.83$\times$) & 59 (2.95$\times$) & 40 (4.35$\times$) \\
$G_{200}$ & 253 & 139 (1.82$\times$) & 86 (2.94$\times$) & 54 (4.68$\times$) \\
\bottomrule
\end{tabular}
\end{sc}
\end{center}
\end{table*}

\section{Discussion and Conclusion}
\label{section:discussion}

We presented a new take on one of the oldest problems in information retrieval by generalizing its domain: How do we efficiently and effectively find $k$ nearest neighbors with respect to inner product over sparse \emph{real} vectors in a high-dimensional space, all in a \emph{streaming} system. As we noted in our introduction to the work, this problem is well-understood when we impose strict constraints on the structure of the space such as non-negativity of vectors, Zipfian distribution of values, higher sparsity rate in query vectors, and near-stationarity of the vector collection. But when these assumptions are violated, we showed that existing methods from the information retrieval literature suffer a substantial performance degradation.

We went back to the drawing board and started with a na\"ive algorithm that proved surprisingly robust and powerful. \linscan{}, as we call it, creates an inverted index populated with pairs of document identifiers and coordinate values. Retrieval is trivial by way of a coordinate-at-a-time algorithm.

\linscan{} is simple to implement and provides an exact solution to our problem. But what if one does not need an exact solution and can do with an approximate top-$k$ set? It turns out that by accepting a degree of error in the solution, we can design an algorithm that can be flexible in the space and time it requires to solve the problem. We proposed \sinnamon{} as an attempt to realize that idea and explained how it inherently offers levers to trade off memory for accuracy for latency.

Our prose hinted at a few research directions we are keen to explore going forward. One is the question of approximate inverted indexes: Just as vector values can be approximated with a compact sketch with some quantifiable error, so can, we hypothesize, the inverted lists. We wish to study if noisy inverted lists can reduce the overall memory footprint possibly at the expense of accuracy.

Staying with memory usage, we also deem it worthy of a thorough empirical investigation to understand the effect of other forms of compression on the inverted lists as well as the sketch matrix, especially in the case of an offline algorithm where it is safe to assume that the vector dataset remains almost stationary. Does knowing the data distribution \emph{a priori} help us quantize vectors or document sketches to make the overall structure more compact? Can we design a sketch with a smaller error probability if we relaxed the streaming requirement?

As an example of what is possible, we take the synthetic datasets of Section~\ref{section:evaluation:synthetic} and assume that the dataset is stationary. We then use a different compression scheme to encode inverted lists in an offline manner. On the $G_{100}$ dataset, switching from Roaring to PForDelta~\cite{pfordelta} reduces the overall size of the index from $1.7$ GB to $1.3$ GB---a reduction of approximately $23\%$. On the $G_{200}$ dataset, we observe a similar reduction in index size.

We believe another effort that may lead to a lower memory consumption is by taking advantage of $h$, the number of random mappings in \sinnamon{}---a hyperparameter that we left largely unexplored in our empirical evaluation of the algorithm. In particular, we showed that by using more random mappings, we may see an increase in the probability of error, but where the probability mass shifts closer to $0$. This suggests that, in applications or environments where consuming less memory is preferred over achieving the lowest latency, it makes more sense to use $2$ or more random mappings. What number of random mappings is appropriate can be determined by the theoretical results presented in this work for a particular data distribution.

Memory and latency aside, we believe our theoretical analysis of the sketch and the retrieval algorithm raise a few interesting questions that we wish to investigate. For example, we developed a good understanding of the magnitude of error and its impact on the final retrieval quality. Can we now use this knowledge to decide what value of $k^\prime$ can give us a particular retrieval accuracy? Can that prediction be done on a per-query basis? Our initial thoughts are that that is indeed possible, but materializing it and quantifying its impact on latency needs further experiments.

We can even take a step beyond sparse vectors and, in what may be unexpected, generalize the top-$k$ retrieval problem to, simply, vectors. More specifically, vectors that may have a ``dense'' part and a ``sparse'' part. When is it sufficient, theoretically and practically, to break up such vectors into distinct collections of their dense part and their sparse part, and solve the two problems separately before re-ranking the top candidates from the two solution sets? When does it make sense to \emph{jointly}, with a single index, solve the top-$k$ retrieval problem over such hybrid vectors? And how? Can these results over hybrid vectors help us design better matrix multiplication algorithms for sparse, or dense-sparse matrices? These and similar questions will be the focus of our future research.

%\begin{acks}
%\end{acks}

\bibliographystyle{ACM-Reference-Format}
\bibliography{main}

\appendix
\section{Proof of Equation~(\ref{equation:analysis:prob-error-gaussian})}
\label{appendix:proof:prob-error-gaussian}

We begin by proving the special case where $h=1$.

\begin{lemma}
Suppose a sparse random vector $X \in \mathbb{R}^n$ is encoded into an $m$-dimensional sketch with \sinnamon{} using Algorithm~\ref{algorithm:indexing} with a single random mapping ($h=1$). Suppose the probability that a coordinate is active, $p_i$, is equal to $p$ for all coordinates. Suppose further that the value of active coordinates $X_i$s are drawn from $\text{Gaussian}(\mu=0, \sigma=1)$. The probability that the upper-bound sketch overestimates the value of $X_i$ upon decoding as $\overline{X}_i$ is the following quantity:
\begin{equation*}
    \mathbb{P}\big[ \overline{X}_i > X_i \big] \approx 1 - \frac{m}{(n - 1)p} (1 - e^{-\frac{(n - 1)p}{m}} ).
\end{equation*}
\label{appendix:lemma:h1}
\end{lemma}
\begin{proof}
    The result from Equation~(\ref{equation:analysis:prob-of-error}) tells us that:
    \begin{equation*}
        \mathbb{P}\big[ \overline{X}_i > X_i \big] \approx \int {\big[ 1 - e^{-\frac{h}{m} (1 - \Phi(\alpha)) (n - 1)p} \big]^h d\mathbb{P}(\alpha)} \overset{h=1}{=} \int {\big[ 1 - e^{-\frac{(1 - \Phi(\alpha)) (n - 1) p}{m}} \big] d\mathbb{P}(\alpha)}.
    \end{equation*}

    Given that $X_i$s are drawn from a Gaussian distribution, and using the approximation above, we can rewrite the probability of error as:
    \begin{equation*}
        \mathbb{P}\big[ \overline{X}_i > X_i \big] \approx \frac{1}{\sqrt{2\pi}} \int_{-\infty}^{\infty} {\big[ 1 - e^{-\frac{(n - 1)p}{m} (1 - \Phi(\alpha))} \big] e^{-\frac{\alpha^2}{2}} d\alpha}.
    \end{equation*}
    We now break up the right hand side into the following three sums, replacing $(n - 1)p/m$ with $\beta$ for brevity:
    \begin{align}
        \mathbb{P}\big[ \overline{X}_i > X_i \big] &\approx
        \int_{-\infty}^{\infty} \frac{1}{\sqrt{2\pi}} e^{-\frac{\alpha^2}{2}} d\alpha \label{equation:appendix:proof:prob-error-gaussian:one} \\
        & - \int_{-\infty}^{0} \frac{1}{\sqrt{2\pi}} e^{-\beta (1 - \Phi(\alpha))} e^{-\frac{\alpha^2}{2}} d\alpha \label{equation:appendix:proof:prob-error-gaussian:negative} \\
        & - \int_{0}^{\infty} \frac{1}{\sqrt{2\pi}} e^{-\beta (1 - \Phi(\alpha))} e^{-\frac{\alpha^2}{2}} d\alpha \label{equation:appendix:proof:prob-error-gaussian:positive}.
    \end{align}
    The sum in~(\ref{equation:appendix:proof:prob-error-gaussian:one}) is equal to the quantity $1$. Let us turn to~(\ref{equation:appendix:proof:prob-error-gaussian:positive}) first. We have that:
    \begin{equation}
        1 - \Phi(\alpha) \overset{\alpha > 0}{=} \frac{1}{2} - \underbrace{\int_{0}^{\alpha} \frac{1}{\sqrt{2\pi}} e^{-\frac{t^2}{2}}}_{\lambda(\alpha)} dt.
    \end{equation}
    As a result, we can write:
    \begin{align}
        \int_{0}^{\infty} \frac{1}{\sqrt{2\pi}} e^{-\beta (1 - \Phi(\alpha))} e^{-\frac{\alpha^2}{2}} d\alpha &= \int_{0}^{\infty} \frac{1}{\sqrt{2\pi}} e^{-\beta (\frac{1}{2} - \lambda(\alpha))} e^{-\frac{\alpha^2}{2}} d\alpha \nonumber \\
        &= e^{-\frac{\beta}{2}} \int_{0}^{\infty} \frac{1}{\sqrt{2\pi}} e^{\beta \lambda(\alpha)} e^{-\frac{\alpha^2}{2}} d\alpha \nonumber \\
        &= \frac{1}{\beta} e^{-\frac{\beta}{2}} e^{\beta\lambda(\alpha)} \big|_{0}^{\infty} = \frac{1}{\beta} e^{-\frac{\beta}{2}} \big( e^\frac{\beta}{2} - 1 \big) \nonumber \\
        &= \frac{1}{\beta} \big( 1 - e^{-\frac{\beta}{2}} \big)
    \end{align}
    By similar reasoning, and noting that:
    \begin{equation}
        1 - \Phi(\alpha) \overset{\alpha < 0}{=} \frac{1}{2} + \underbrace{\int_{\alpha}^{0} \frac{1}{\sqrt{2\pi}} e^{-\frac{t^2}{2}}}_{-\lambda(\alpha)} dt,
    \end{equation}
    we arrive at:
    \begin{equation}
        \int_{-\infty}^{0} \frac{1}{\sqrt{2\pi}} e^{-\beta (1 - \Phi(\alpha))} e^{-\frac{\alpha^2}{2}} d\alpha = \frac{1}{\beta} e^{-\frac{\beta}{2}} \big( 1 - e^{-\frac{\beta}{2}} \big)
    \end{equation}

    Plugging the results above into Equations~(\ref{equation:appendix:proof:prob-error-gaussian:one}),~(\ref{equation:appendix:proof:prob-error-gaussian:negative}), and~(\ref{equation:appendix:proof:prob-error-gaussian:positive}) results in:
    \begin{align*}
        \mathbb{P}\big[ \overline{X}_i > X_i \big] &\approx 1 - \frac{1}{\beta} \big( 1 - e^{-\frac{\beta}{2}} \big) - \frac{1}{\beta} e^{-\frac{\beta}{2}} \big( 1 - e^{-\frac{\beta}{2}} \big) \\
        &= 1 - \frac{1}{\beta} \big( 1 - e^{-\frac{\beta}{2}} \big) \big( 1 + e^{-\frac{\beta}{2}} \big) \\
        &= 1 - \frac{m}{(n-1)p} \big( 1 - e^{-\frac{(n - 1)p}{m}} \big),
    \end{align*}
    which completes the proof.
\end{proof}

Given the result above, the solution for the general case of $h > 0$ is straightforward to obtain. We show this result in the lemma below.

\begin{lemma}
Under similar conditions as Lemma~\ref{appendix:lemma:h1} but assuming $h \geq 1$, the probability that the upper-bound sketch overestimates an active coordinate $X_i$ upon decoding it as $\overline{X}_i$ is the following quantity:
\begin{equation*}
    \mathbb{P}\big[ \overline{X}_i > X_i \big] \approx 1 + \sum_{k=1}^{h} {h \choose k} (-1)^k \frac{m}{kh(n - 1)p} (1 - e^{-\frac{kh(n - 1)p}{m}} ).
\end{equation*}
\end{lemma}
\begin{proof}

    Using the binomial theorem, we have that:
    \begin{equation*}
        \mathbb{P}\big[ \overline{X}_i > X_i \big] \approx \int {\big[ 1 - e^{-\frac{h}{m} (1 - \Phi(\alpha)) (n - 1) p } \big]^h d\mathbb{P}(\alpha)} = \sum_{k=0}^{h} { {h \choose k} \int \big( -e^{-\frac{h}{m} (1 - \Phi(\alpha)) (n - 1) p } \big)^{k} d\mathbb{P}(\alpha)}.
    \end{equation*}
    We rewrite the expression above for Gaussian variables to arrive at:
    \begin{equation*}
        \mathbb{P}\big[ \overline{X}_i > X_i \big] \approx \frac{1}{\sqrt{2\pi}} \sum_{k=0}^{h} {h \choose k} \int_{-\infty}^{\infty} {\big(- e^{-\frac{h(n - 1)p}{m} (1 - \Phi(\alpha))} \big)^{k} e^{-\frac{\alpha^2}{2}} d\alpha}.
    \end{equation*}
    Following the proof of the previous lemma, we can expand the right hand side as follows:
    \begin{align*}
        \mathbb{P}\big[ \overline{X}_i > X_i \big] &\approx 1 + \frac{1}{\sqrt{2\pi}} \sum_{k=1}^{h} {h \choose k} (-1)^k \int_{-\infty}^{\infty} {e^{-\frac{kh(n - 1)p}{m} (1 - \Phi(\alpha))} e^{-\frac{\alpha^2}{2}} d\alpha} \\
        &= 1 + \sum_{k=1}^{h} {h \choose k} (-1)^k \frac{m}{kh(n - 1)p} (1 - e^{-\frac{kh(n - 1)p}{m}} ),
    \end{align*}
    which completes the proof.
\end{proof}

\section{Proof of Equation~(\ref{equation:analysis:dist-error-gaussian})}
\label{appendix:proof:dist-error-gaussian}

\begin{lemma}
Under the conditions of Corollary~\ref{corollary:analysis:dist-error-gaussian}, the cumulative distribution function of $\overline{Z}_i$ can then be approximated as follows:
\begin{equation*}
    \mathbb{P}[\overline{Z}_i \leq \delta] \approx 1 - \big[ 1 - e^{-\frac{h (n - 1)p}{m} (1 - \Phi(\delta)) } \big]^h,
\end{equation*}
where $\Phi(\cdot)$ is the CDF of a zero-mean Gaussian with standard deviation $\sigma\sqrt{2}$.
\end{lemma}
\begin{proof}

Recall from Equation~(\ref{equation:analysis:dist-error}) the following result:
\begin{equation*}
    \mathbb{P}[\overline{Z}_i \leq \delta] \approx 1 - \int \big[ 1 - e^{-\frac{h}{m} (1 - \Phi(\alpha + \delta)) (n - 1)p } \big]^h d \mathbb{P}(\alpha).
\end{equation*}

When the active values of a vector are drawn from a Gaussian distribution, then the pairwise difference between any two coordinates itself has a Gaussian distribution with standard deviation $\sqrt{\sigma^2 + \sigma^2}=\sigma \sqrt{2}$. As such, we may estimate $1 - \Phi(\alpha + \delta)$ by considering the probability that a pair of coordinates (one of which having value $\alpha$) has a difference greater than $\delta$: $\mathbb{P}[X_i - X_j > \delta]$. With that idea, we may thus write:
\begin{equation*}
    1 - \Phi(\alpha + \delta) = 1 - \Phi^\prime(\delta),
\end{equation*}
where $\Phi^\prime(\cdot)$ is the CDF of a zero-mean Gaussian distribution with standard deviation $\sigma\sqrt{2}$.

Putting everything together, we can write:
\begin{align*}
    \mathbb{P}[\overline{Z}_i \leq \delta] &\approx 1 - \int \big[ 1 - e^{-\frac{h(n - 1)p}{m}(1 - \Phi^\prime(\delta))} \big]^h d \mathbb{P}(\alpha) \\
    &= 1 - \big[ 1 - e^{-\frac{h(n - 1)p}{m}(1 - \Phi^\prime(\delta))} \big]^h,
\end{align*}
proving our claim.
\end{proof}

\section{Analysis on Remaining Vector Datasets}
\label{appendix:analysis-empirical}

\begin{figure}[h]
\begin{center}
\centerline{
\subfloat[Cumulative distribution of sketching error for $h=1$ (left) and $h=2$ (right)]{
\includegraphics[width=0.27\linewidth]{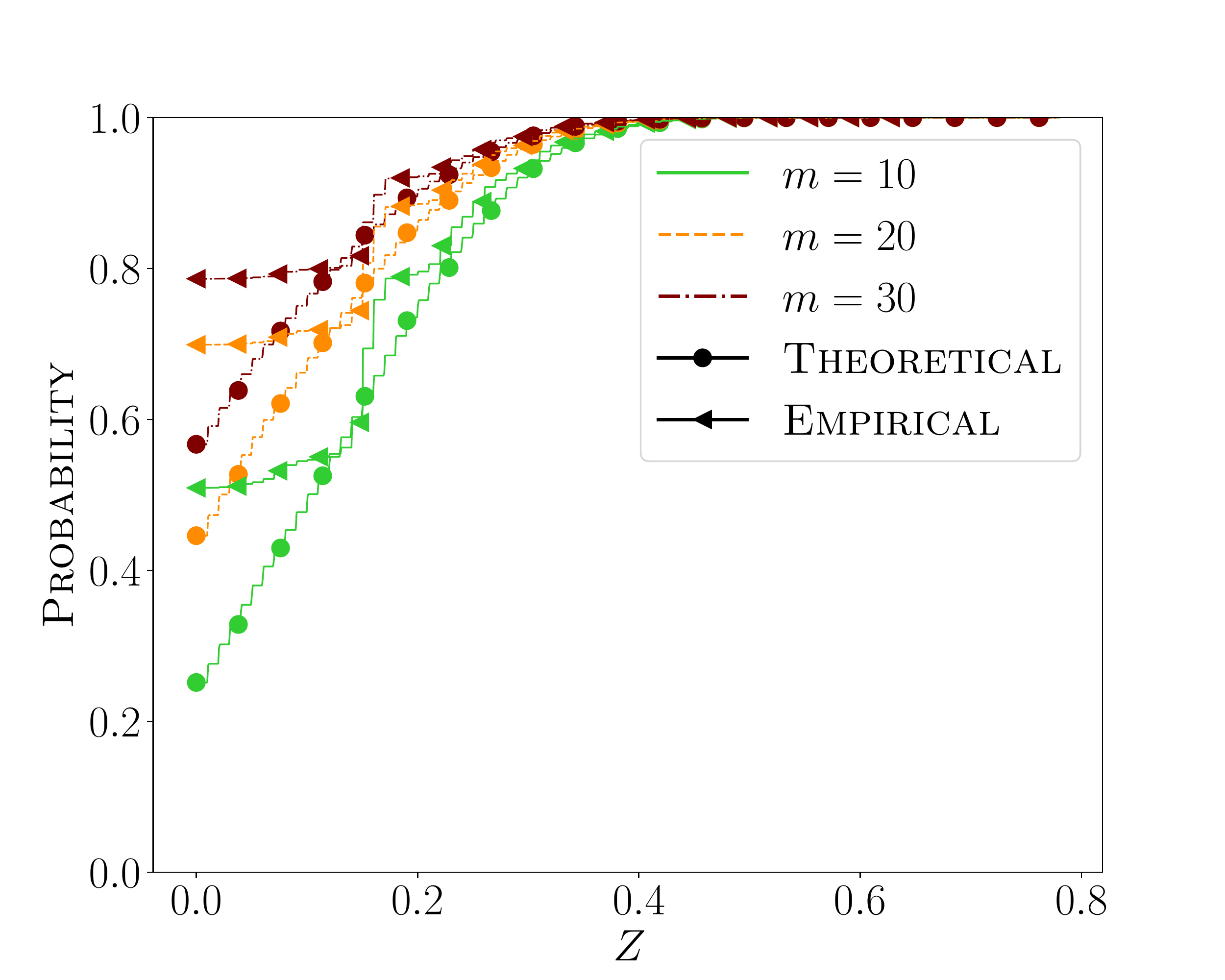}
\includegraphics[width=0.27\linewidth]{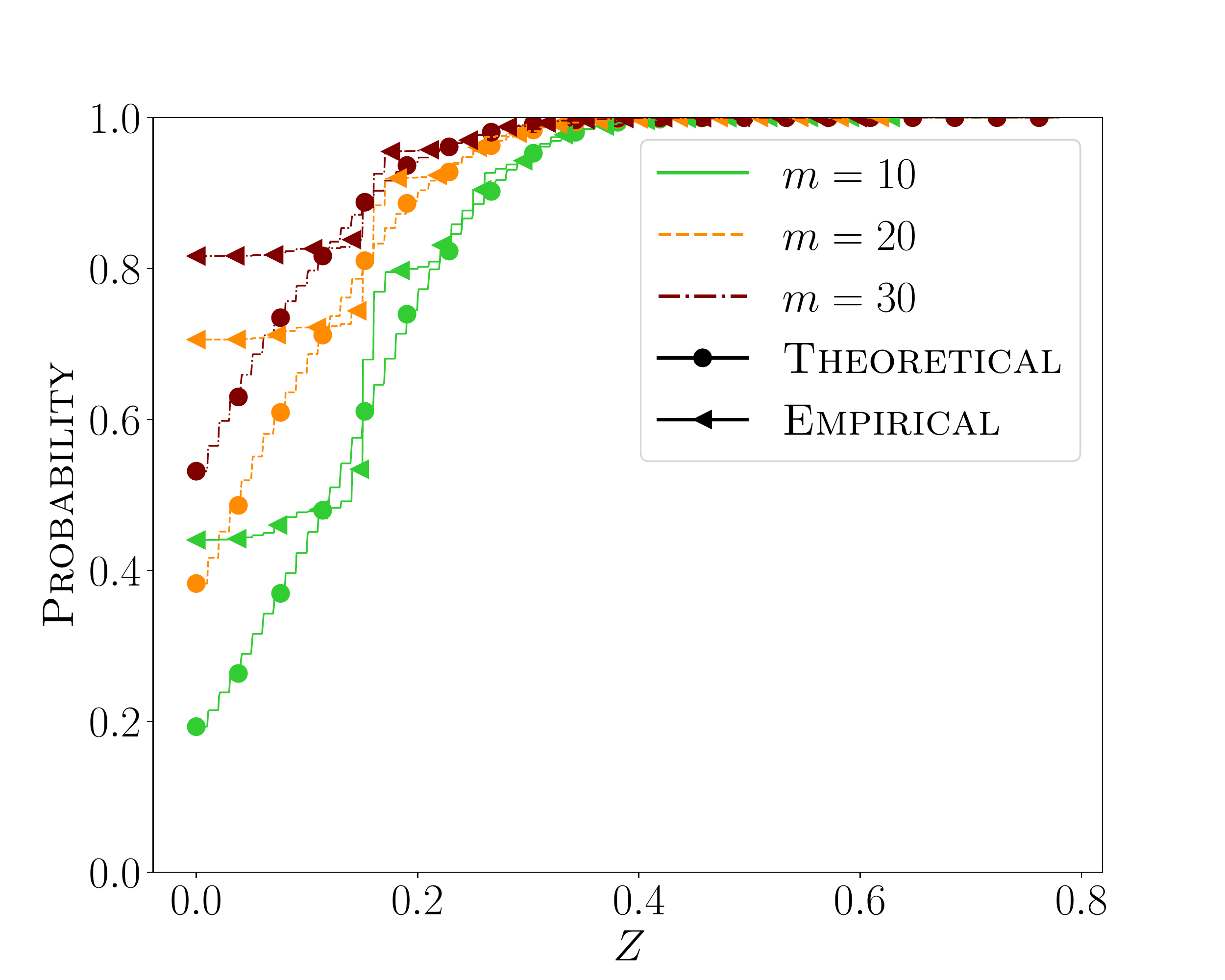}}
}
\centerline{
\subfloat[Inner product error for $h=1$ (left) and $h=2$ (right)]{
\includegraphics[width=0.27\linewidth]{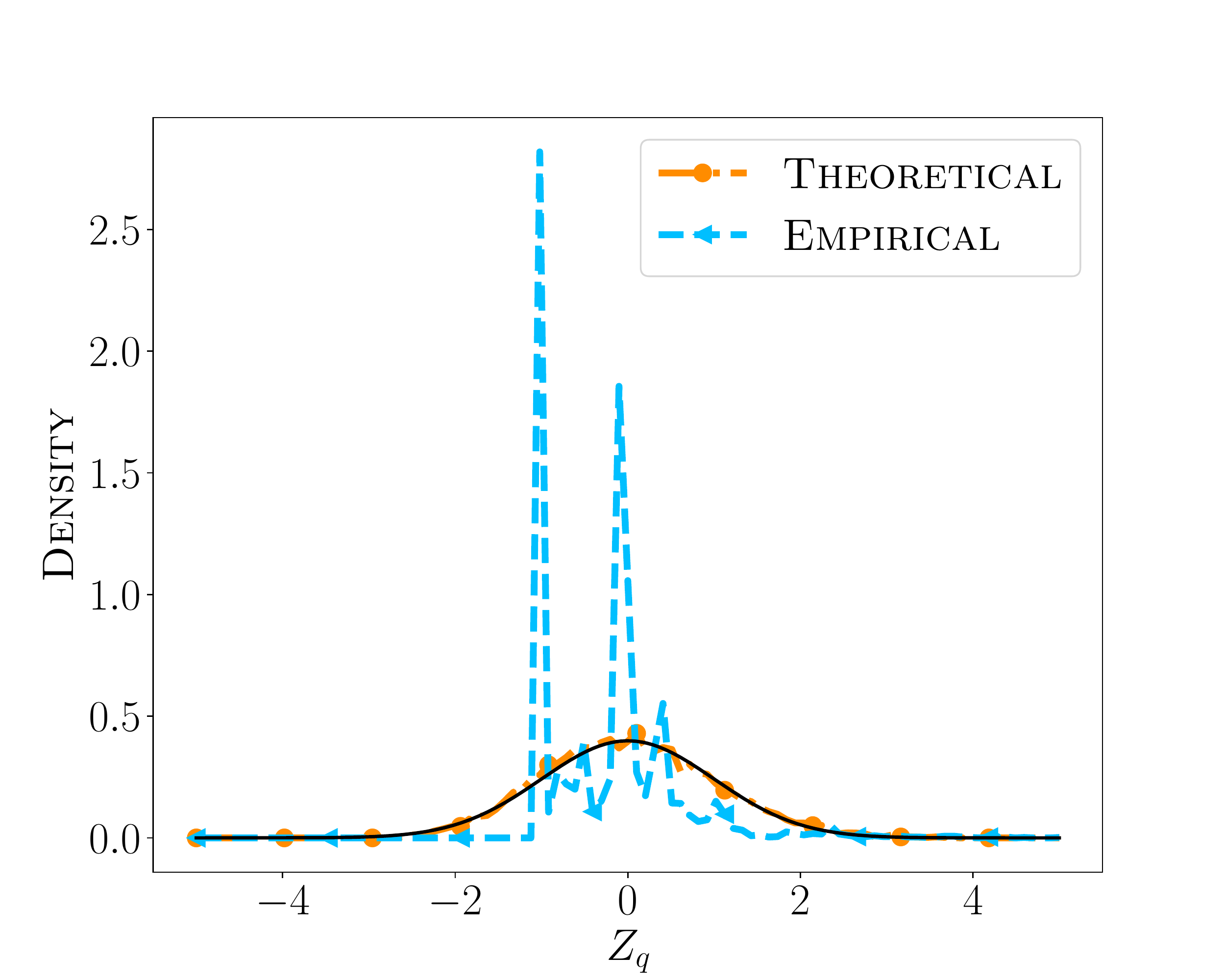}
\includegraphics[width=0.27\linewidth]{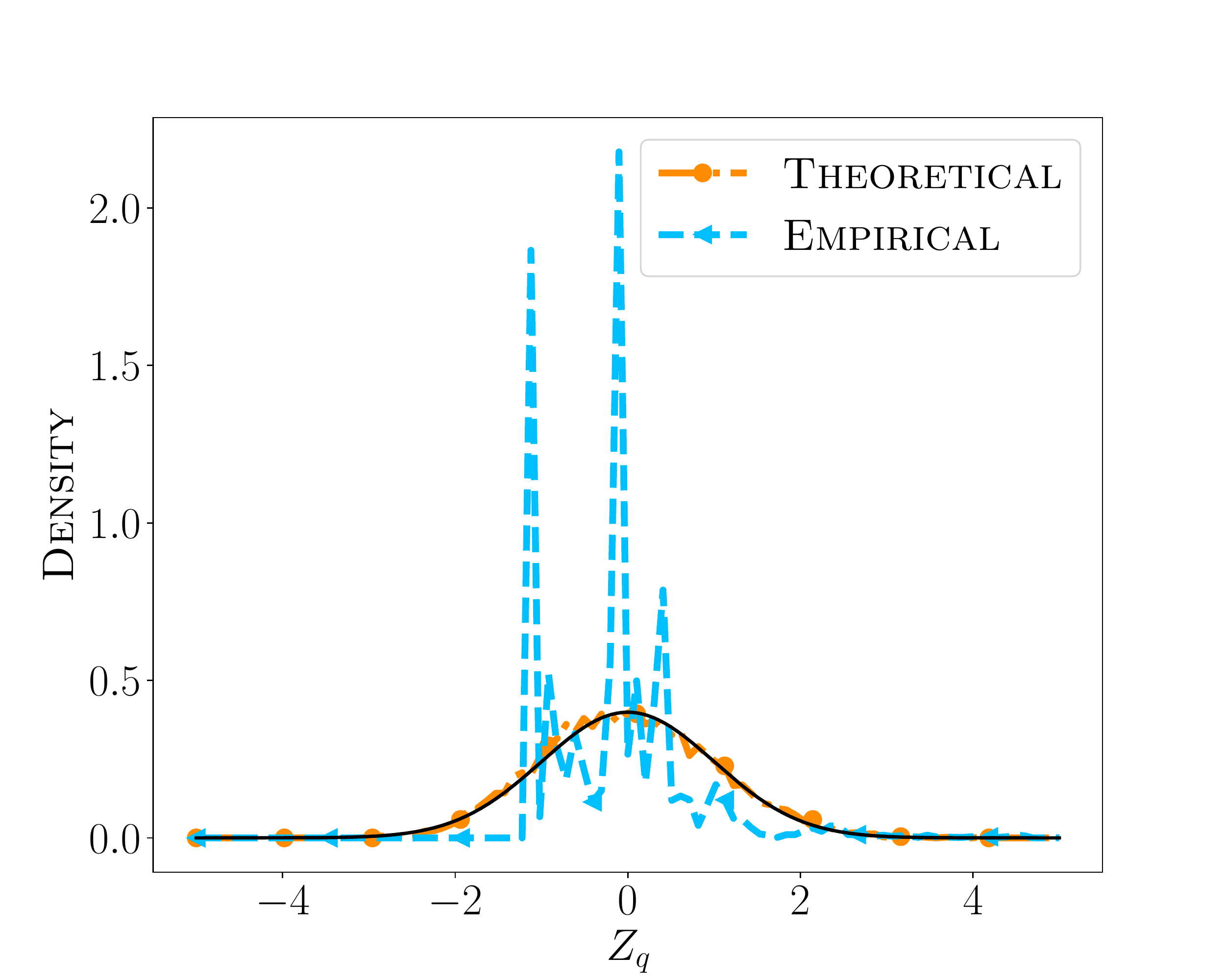}}
}

\caption{Sketching and inner product error distributions for the BM25 dataset.}
\label{figure:evaluation:analysis:bm25}
\end{center}
\end{figure}

\begin{figure}[h]
\begin{center}
\centerline{
\subfloat[Cumulative distribution of sketching error for $h=1$ (left) and $h=2$ (right)]{
\includegraphics[width=0.27\linewidth]{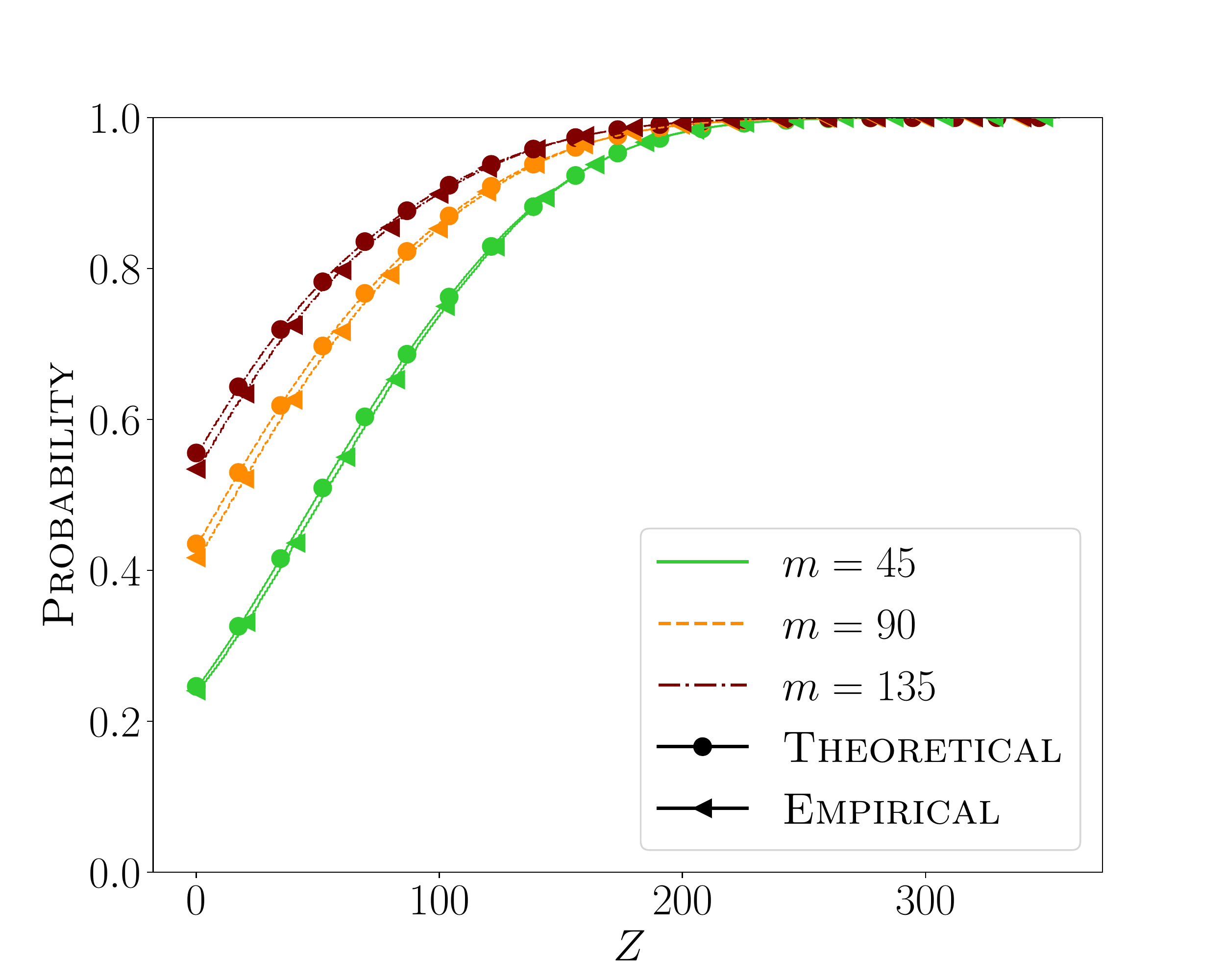}
\includegraphics[width=0.27\linewidth]{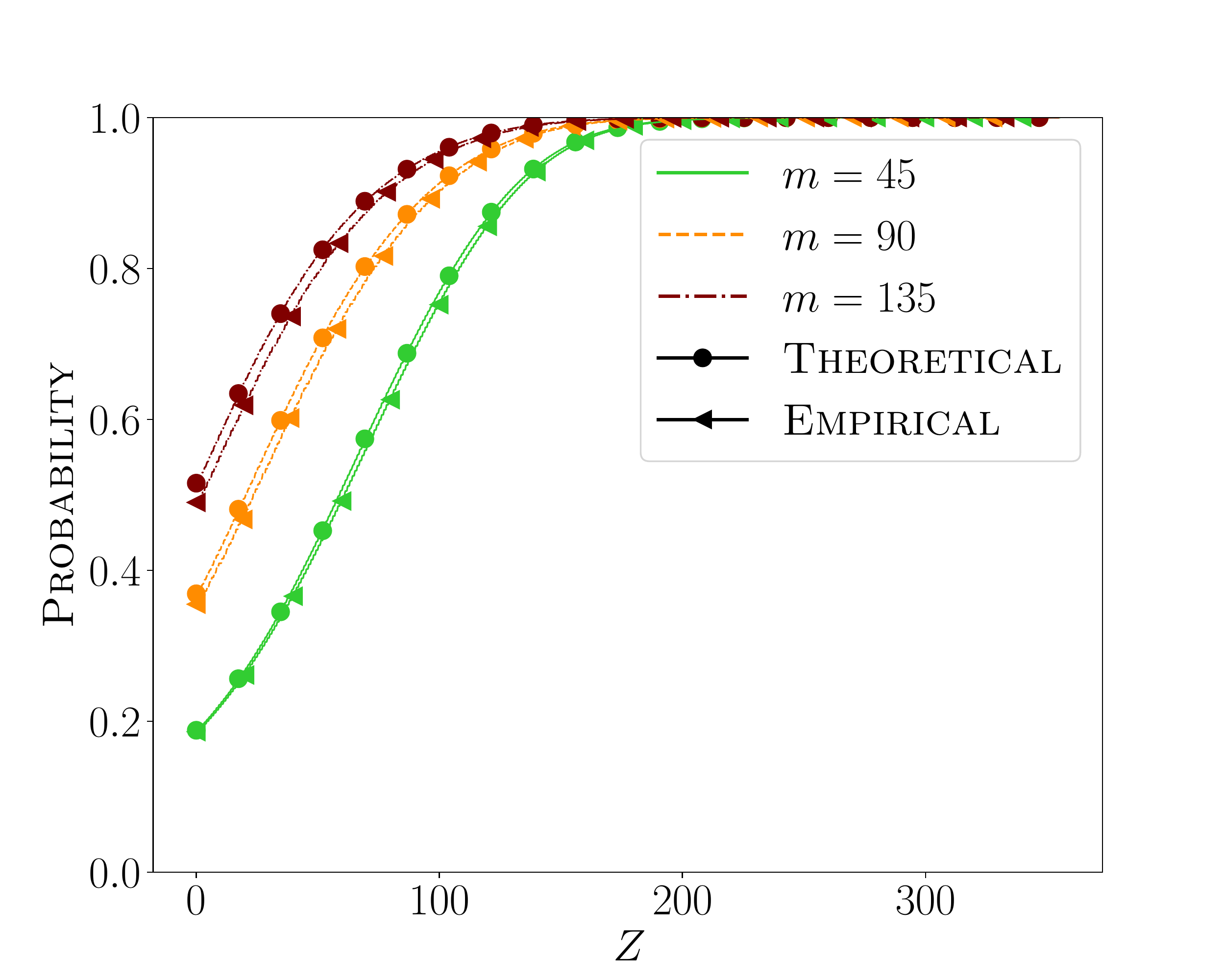}}
}
\centerline{
\subfloat[Inner product error for $h=1$ (left) and $h=2$ (right)]{
\includegraphics[width=0.27\linewidth]{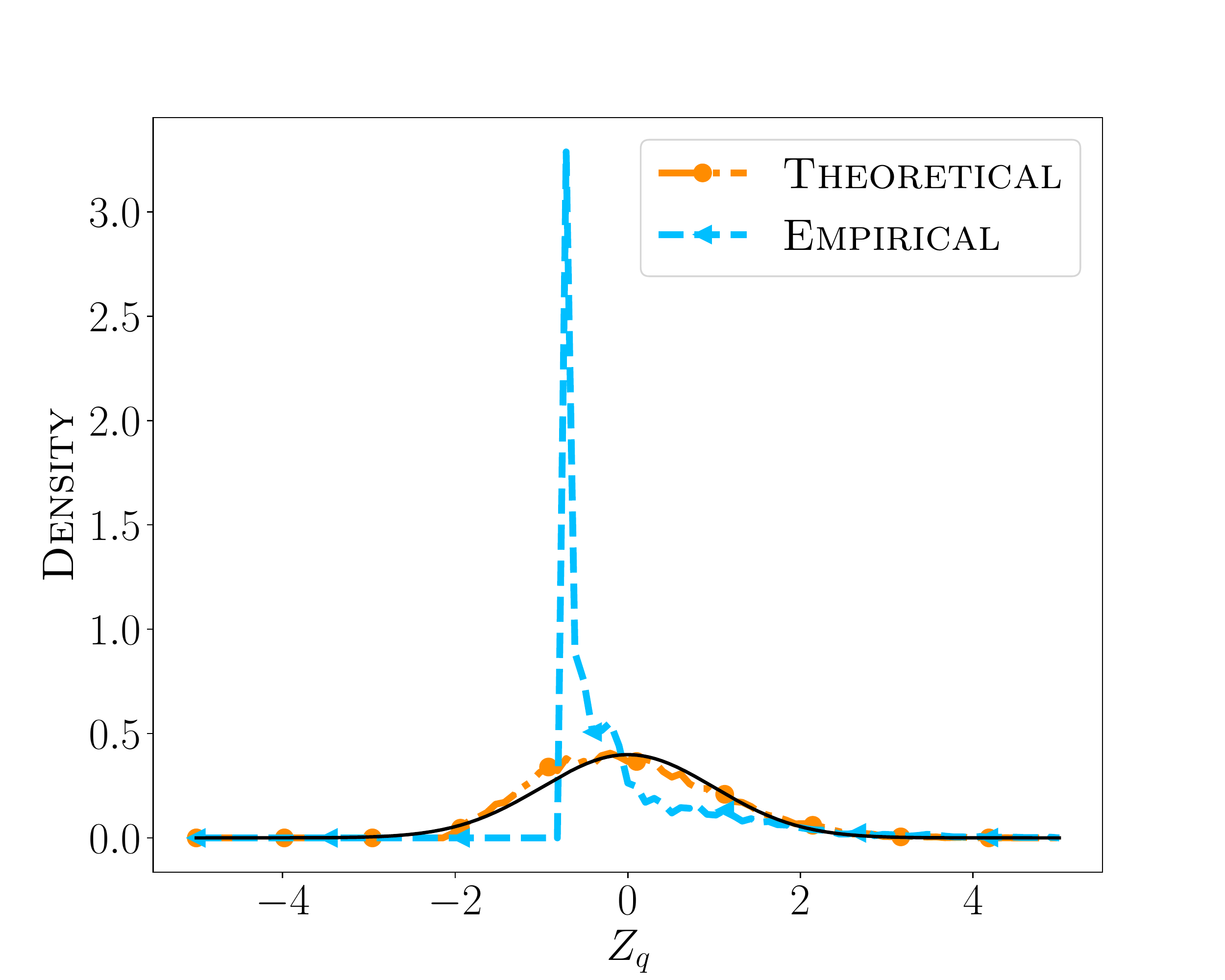}
\includegraphics[width=0.27\linewidth]{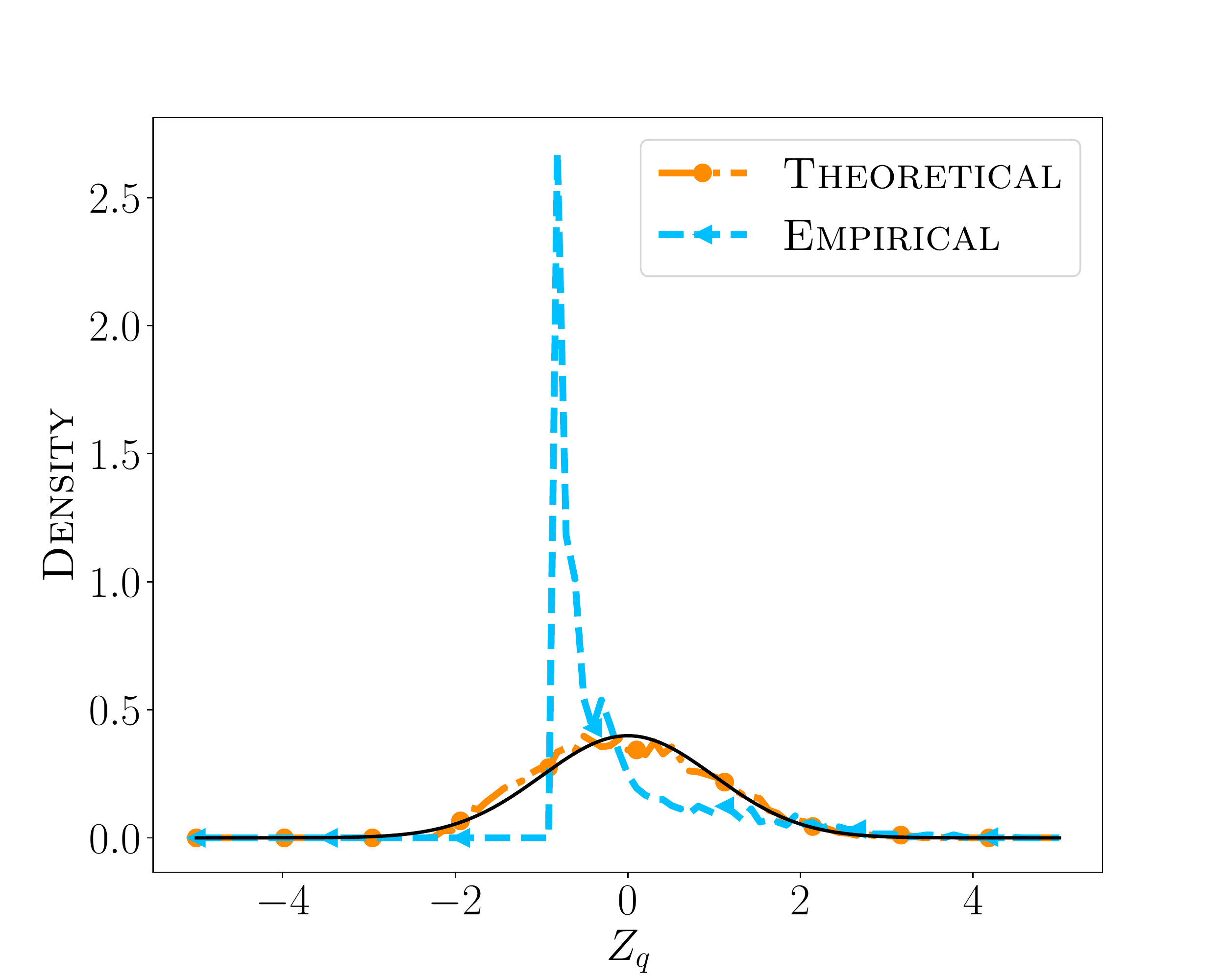}}
}

\caption{Sketching and inner product error distributions for the Efficient SPLADE dataset.}
\label{figure:evaluation:analysis:splade-efficient}
\end{center}
\end{figure}

\begin{figure}[!h]
\begin{center}
\centerline{
\subfloat[Cumulative distribution of sketching error for $h=1$ (left) and $h=2$ (right)]{
\includegraphics[width=0.3\linewidth]{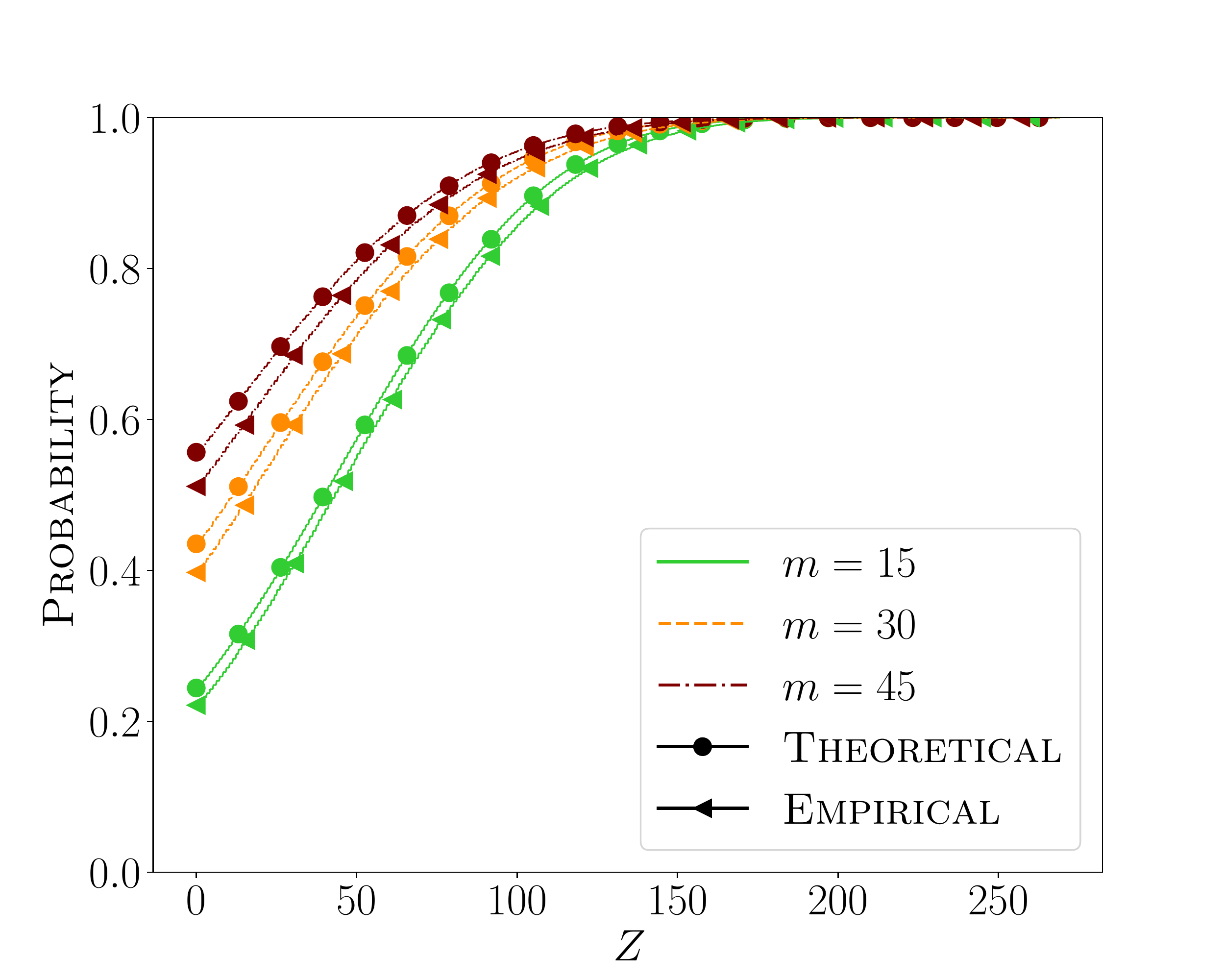}
\includegraphics[width=0.3\linewidth]{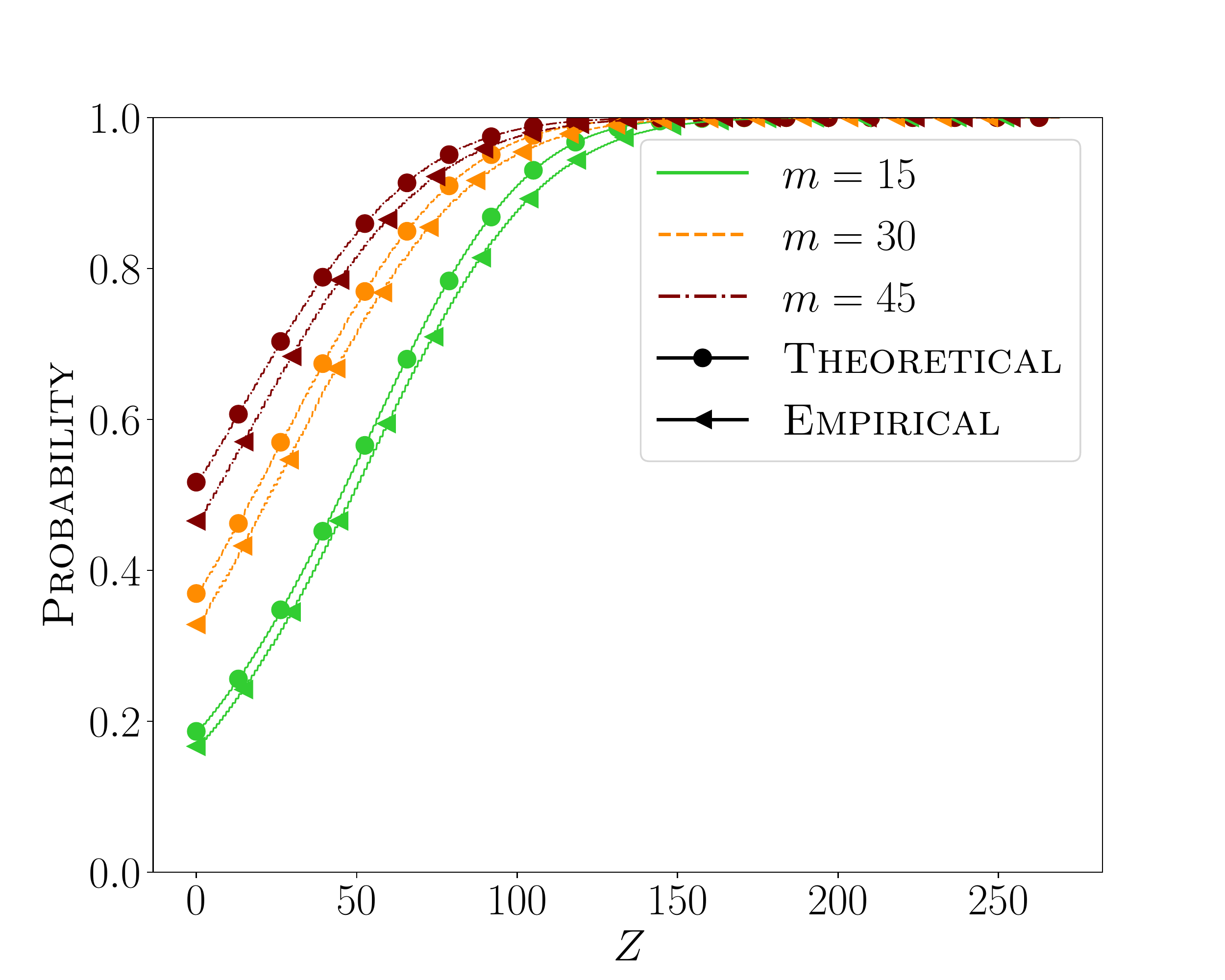}}
}
\centerline{
\subfloat[Inner product error for $h=1$ (left) and $h=2$ (right)]{
\includegraphics[width=0.3\linewidth]{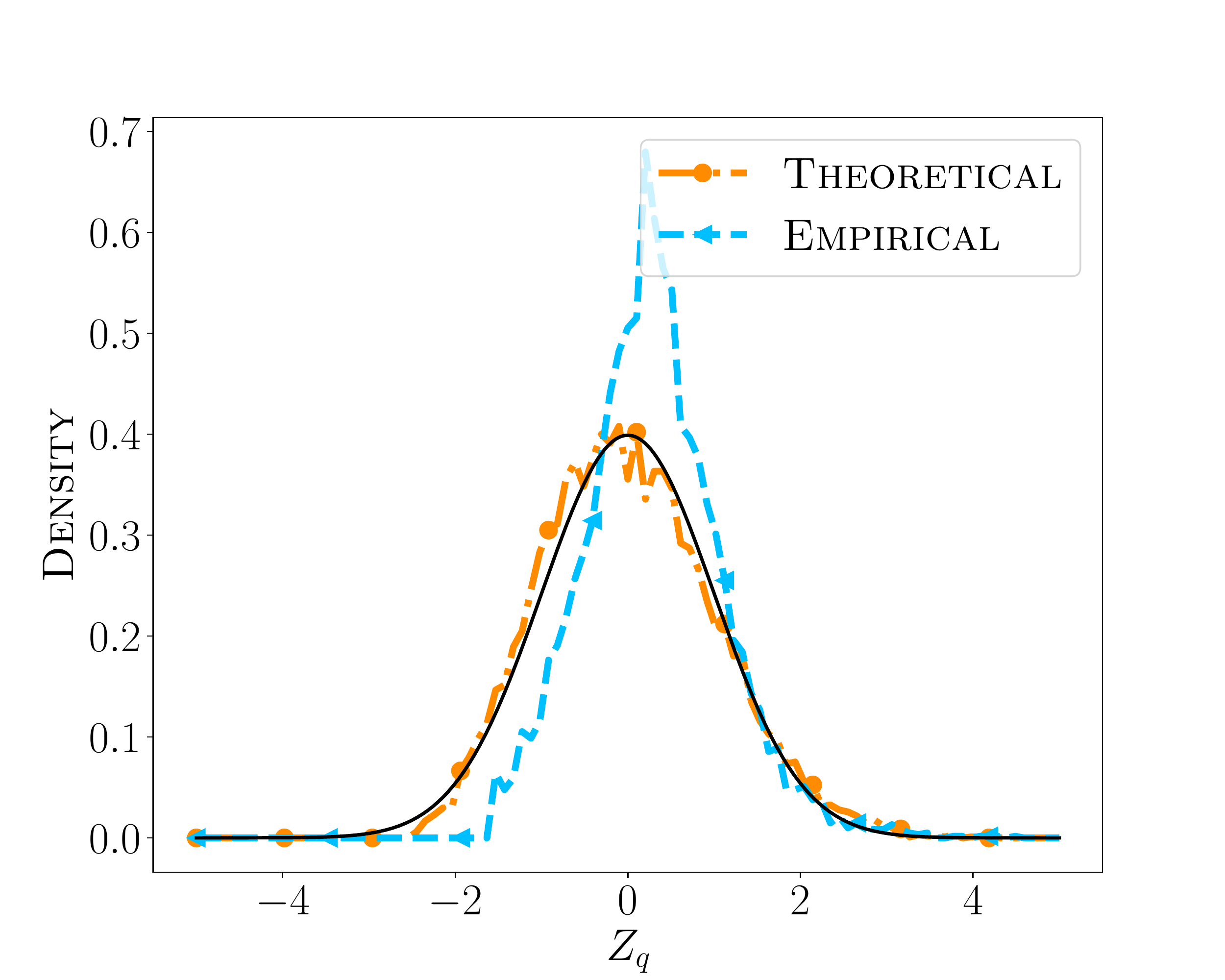}
\includegraphics[width=0.3\linewidth]{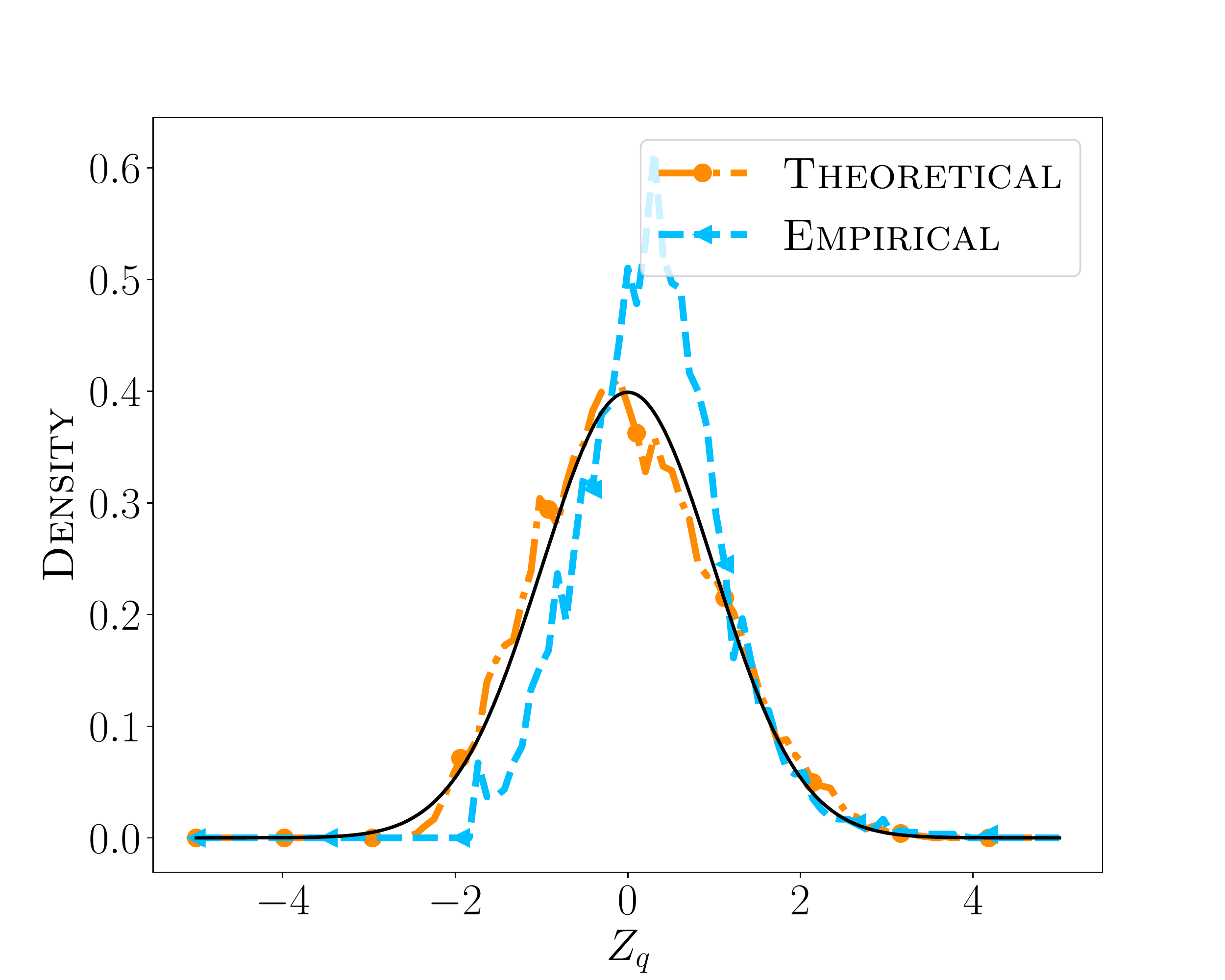}}
}

\caption{Sketching and inner product error distributions for the uniCOIL dataset.}
\label{figure:evaluation:analysis:unicoil}
\end{center}
\end{figure}

\FloatBarrier
\section{Trade-offs in Retrieval on Apple M1}
\label{appendix:retrieval-m1}

\begin{figure}[!ht]
\begin{center}
\centerline{
\subfloat[BM25]{
\includegraphics[width=0.4\linewidth]{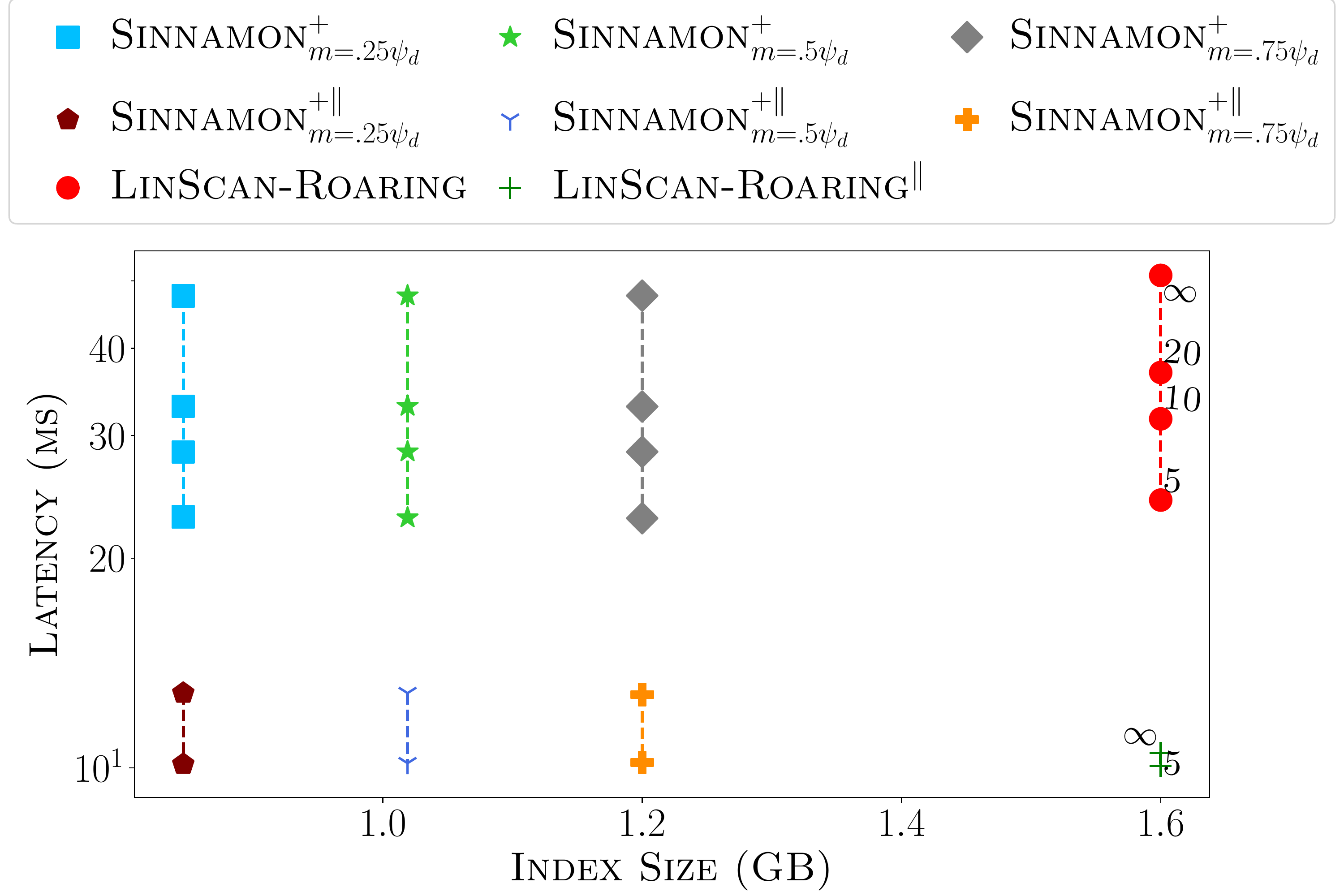}
\includegraphics[width=0.4\linewidth]{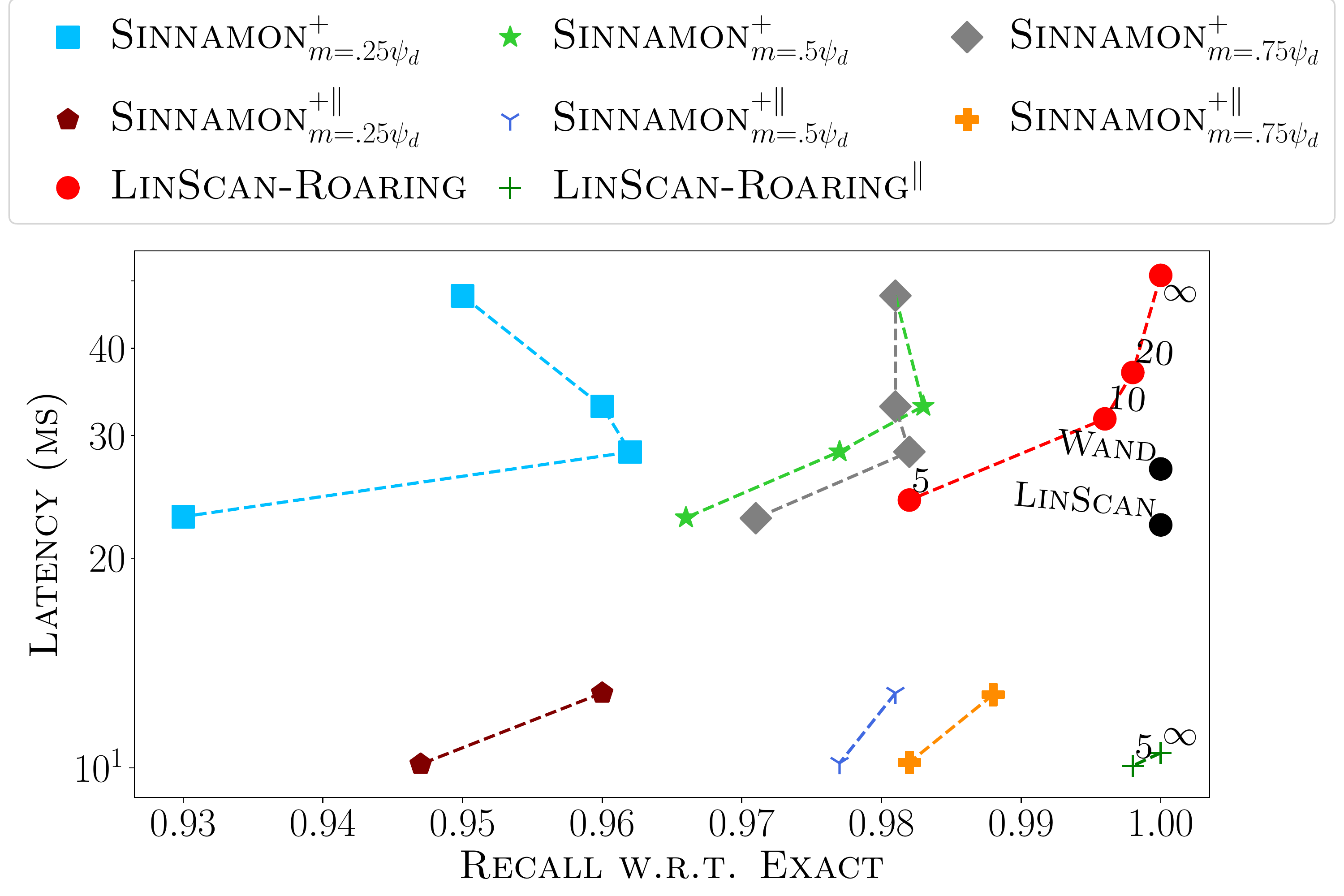}}
}
% \end{center}
% \end{figure}
% \begin{figure}[!ht]
% \begin{center}
% \ContinuedFloat
\centerline{
\subfloat[SPLADE]{
\includegraphics[trim={0 0 0 5.3cm},clip,width=0.4\linewidth]{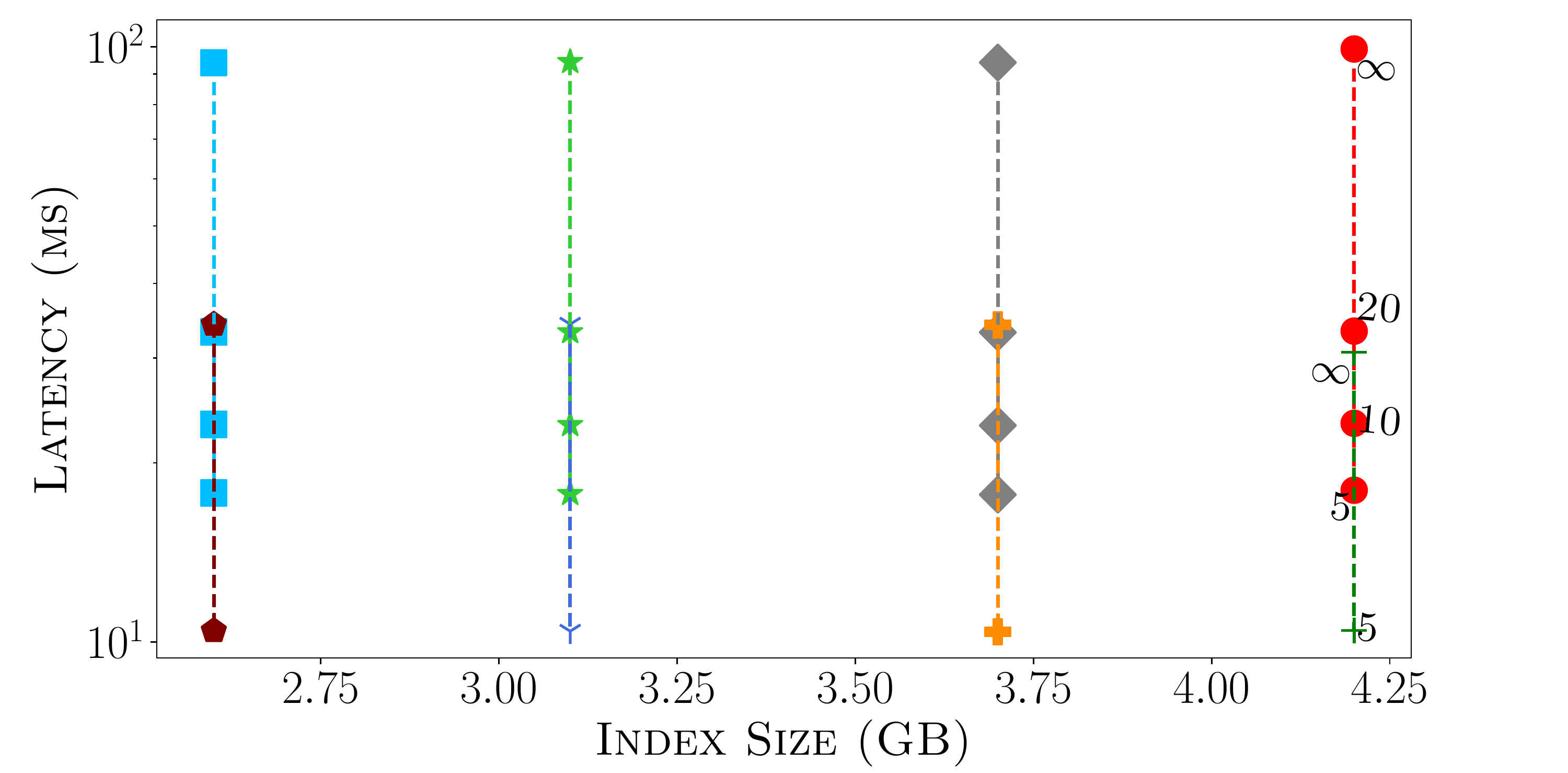}
\includegraphics[trim={0 0 0 5.3cm},clip,width=0.4\linewidth]{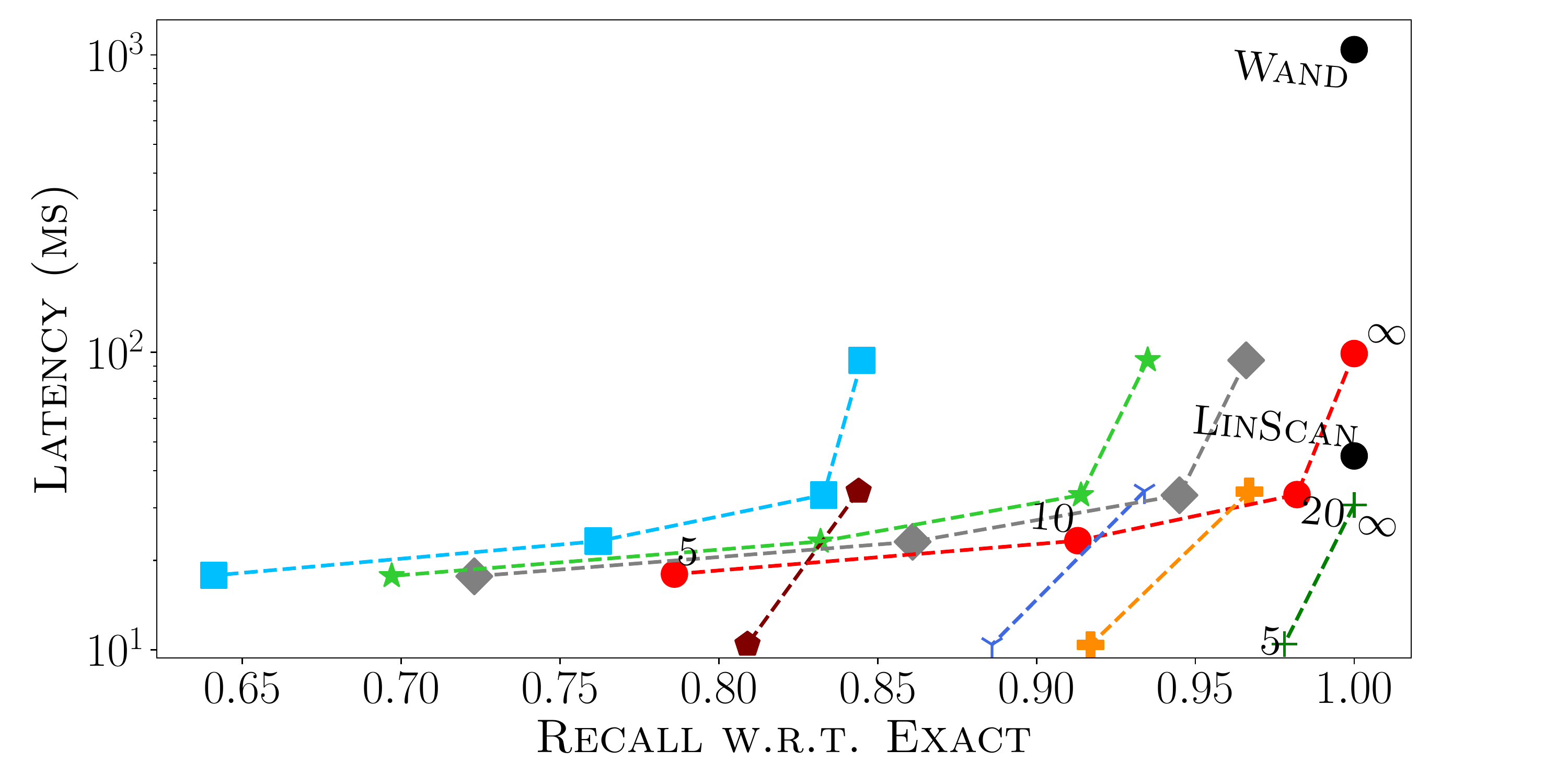}}
}
% \end{center}
% \end{figure}
% \begin{figure}[!ht]
% \begin{center}
% \ContinuedFloat
\centerline{
\subfloat[Efficient SPLADE]{
\includegraphics[trim={0 0 0 5.3cm},clip,width=0.4\linewidth]{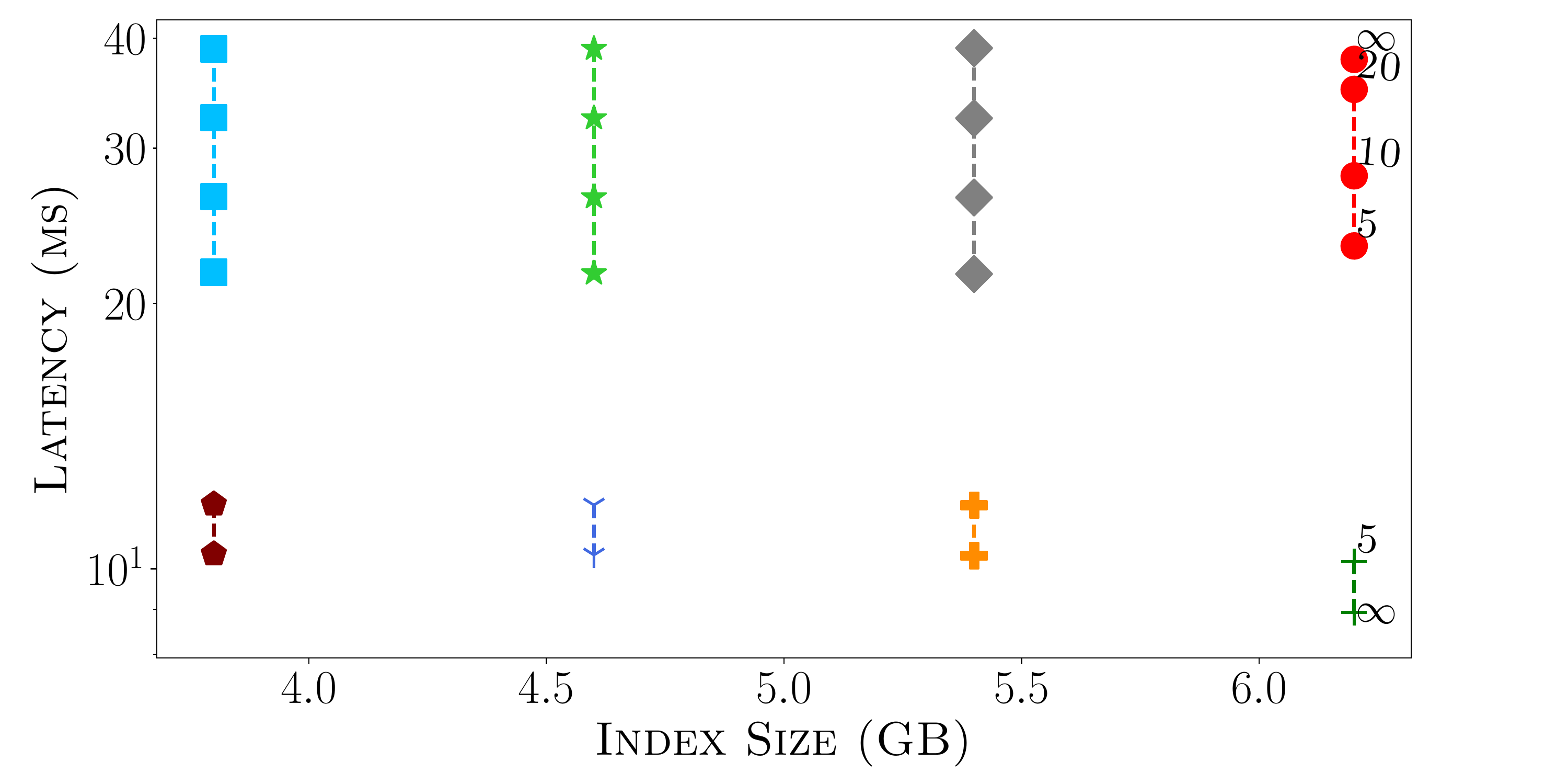}
\includegraphics[trim={0 0 0 5.3cm},clip,width=0.4\linewidth]{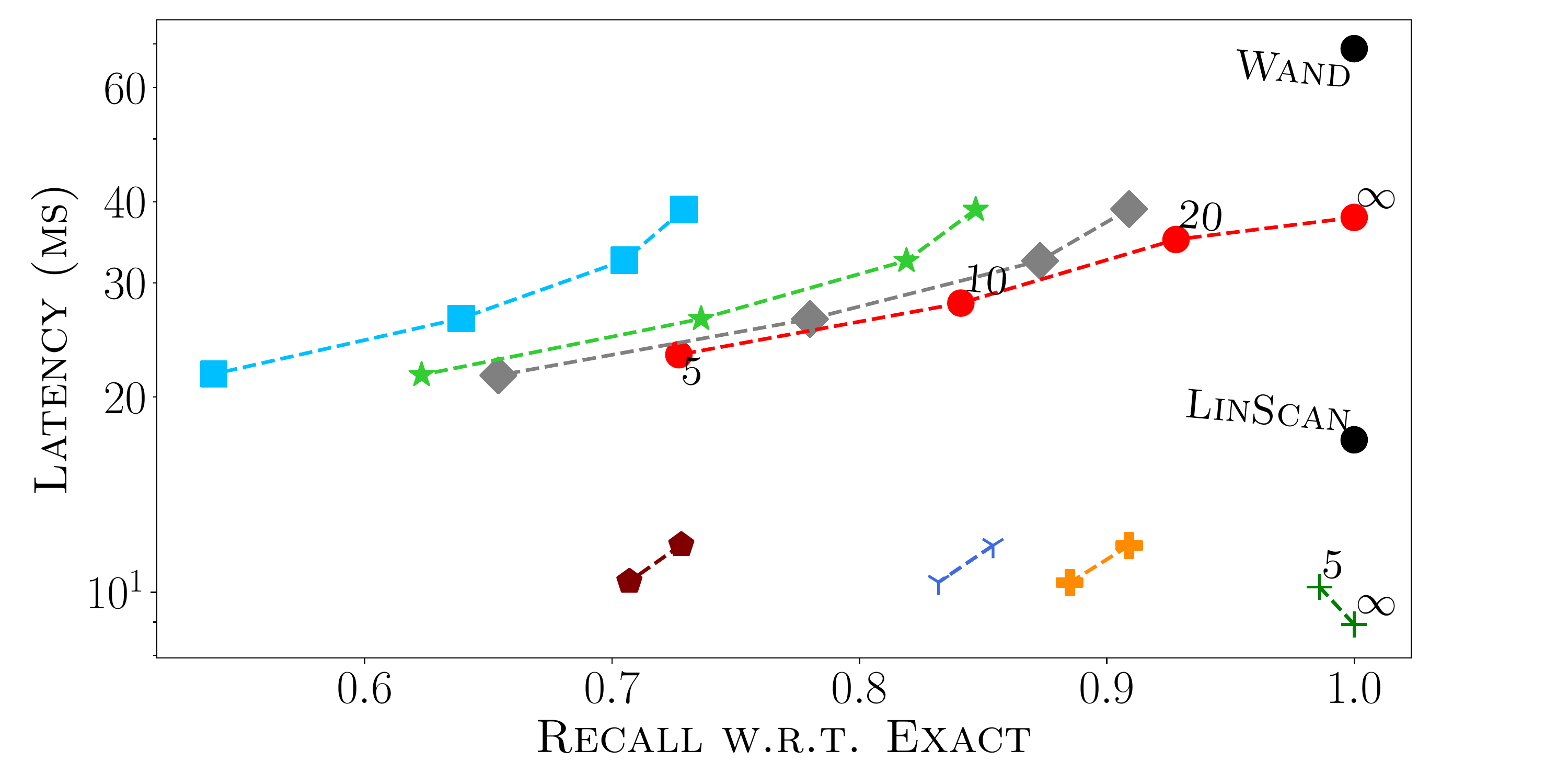}}
}
% \end{center}
% \end{figure}
% \begin{figure}[!ht]
% \begin{center}
% \ContinuedFloat
\centerline{
\subfloat[uniCOIL]{
\includegraphics[trim={0 0 0 5.3cm},clip,width=0.4\linewidth]{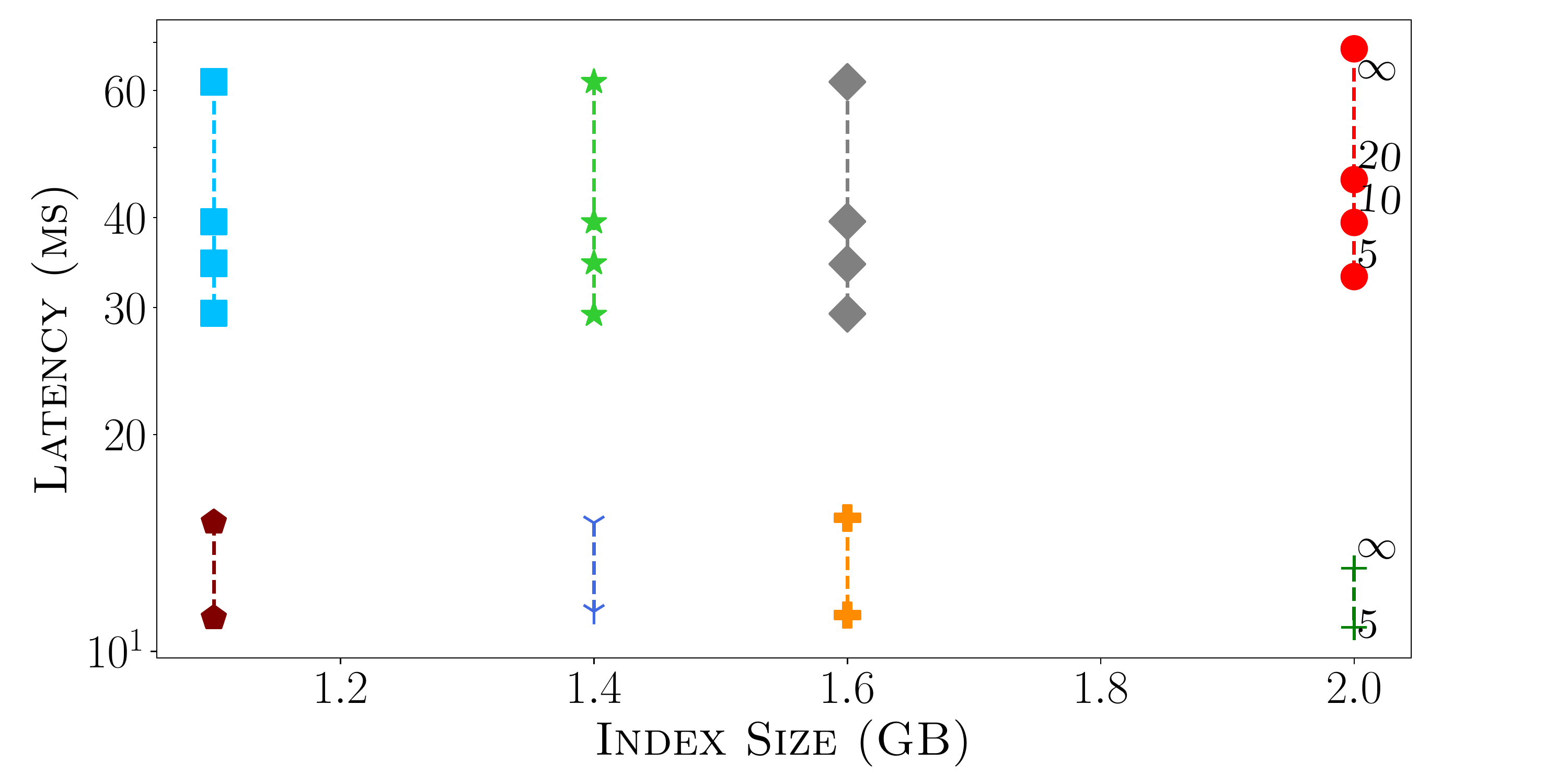}
\includegraphics[trim={0 0 0 5.3cm},clip,width=0.4\linewidth]{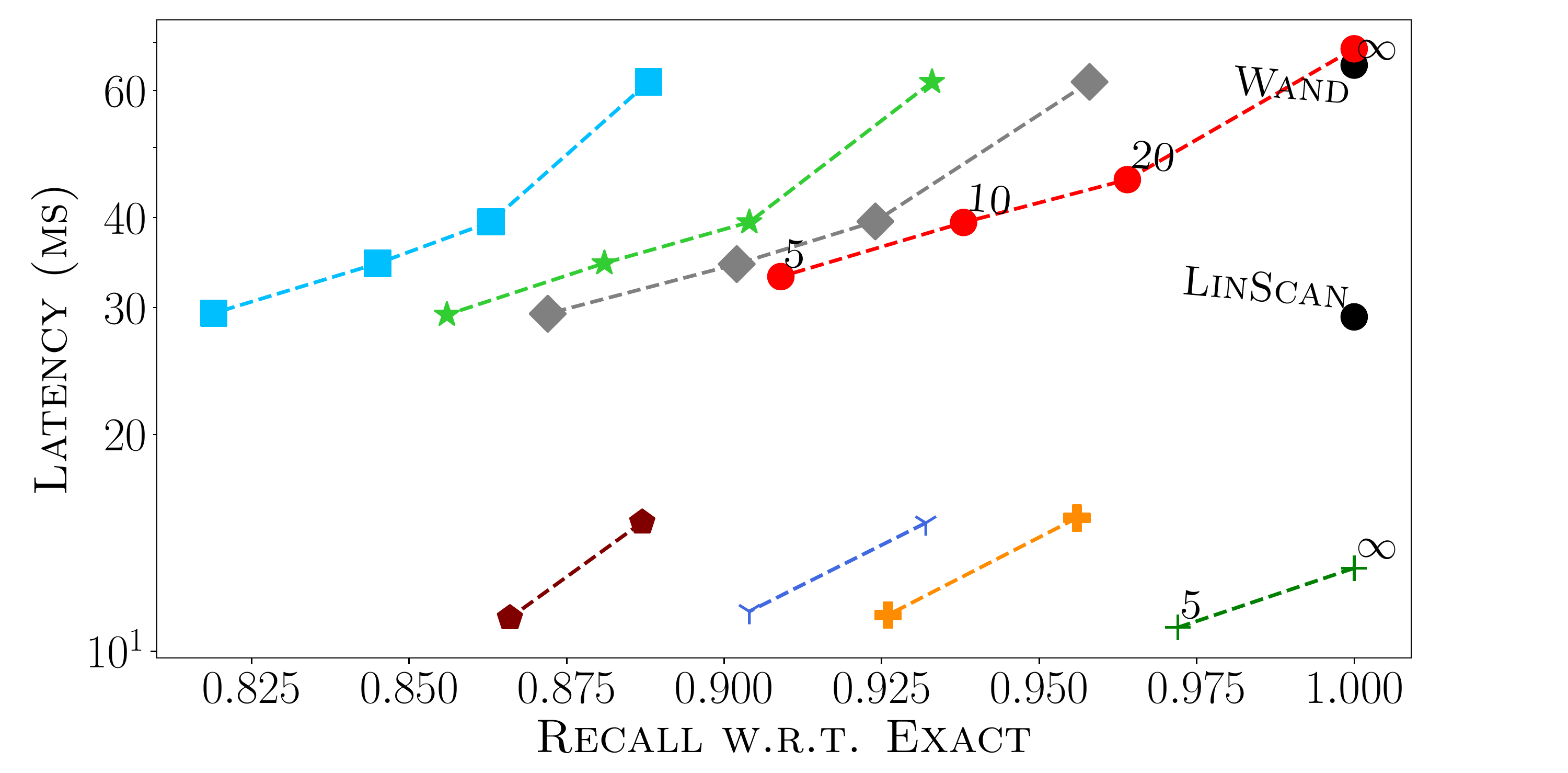}}
}

\caption{Trade-offs on the Apple M1 processor between latency and memory (left column), and latency and accuracy (right column) for various vector collections (rows). Shapes (and colors) distinguish between different configurations of \sinnamon{}, and points on a line represent different time budgets $T$ (in milliseconds).}
\label{figure:evaluation:msmarco-passage-v1-m1}
\end{center}
\end{figure}
\FloatBarrier

\section{Trade-offs Concerning Task Metrics on Apple M1}
\label{appendix:retrieval-end-metric-m1}

\begin{figure}[!ht]
\begin{center}
\centerline{
\subfloat[BM25]{
\includegraphics[width=0.4\linewidth]{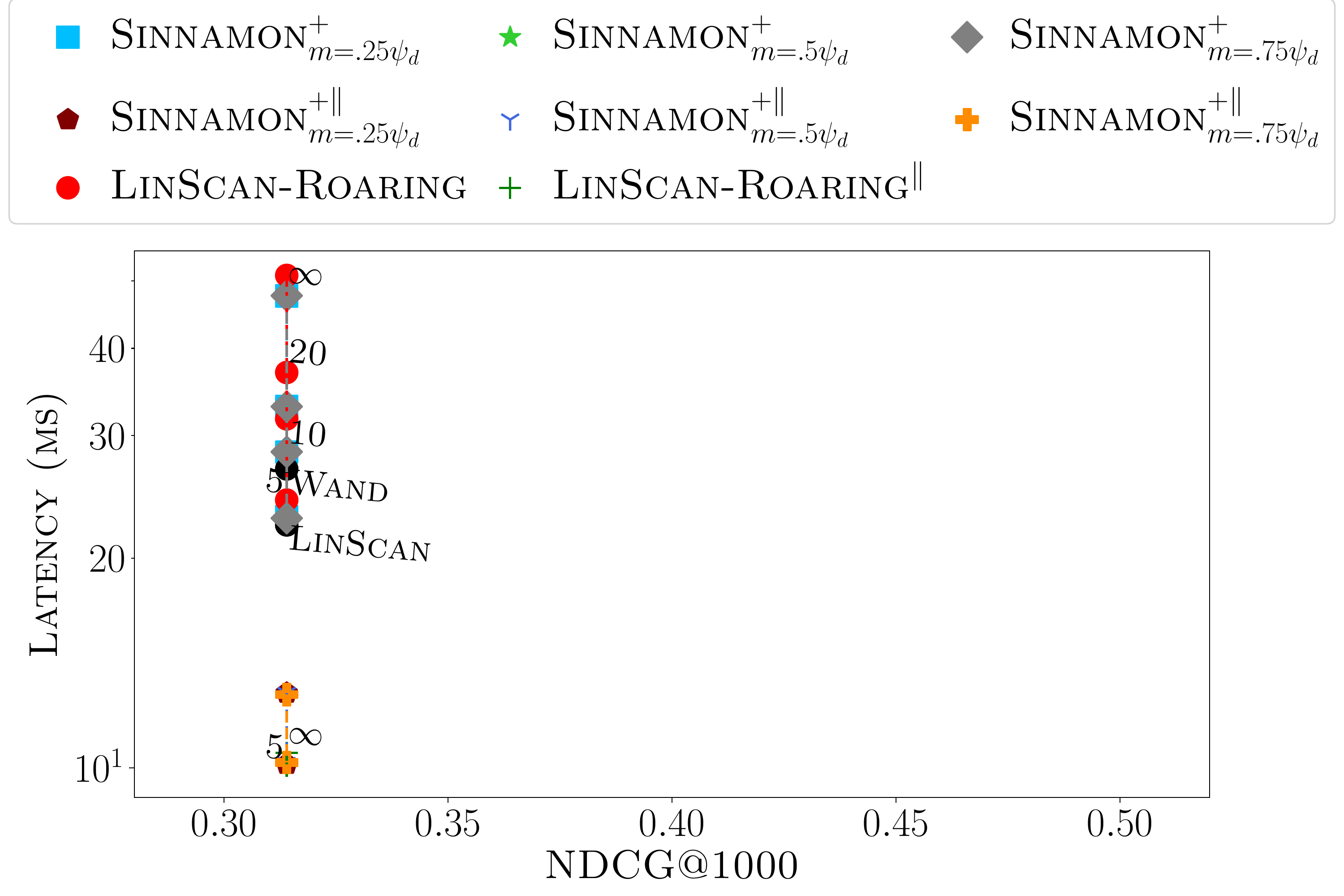}
\includegraphics[width=0.4\linewidth]{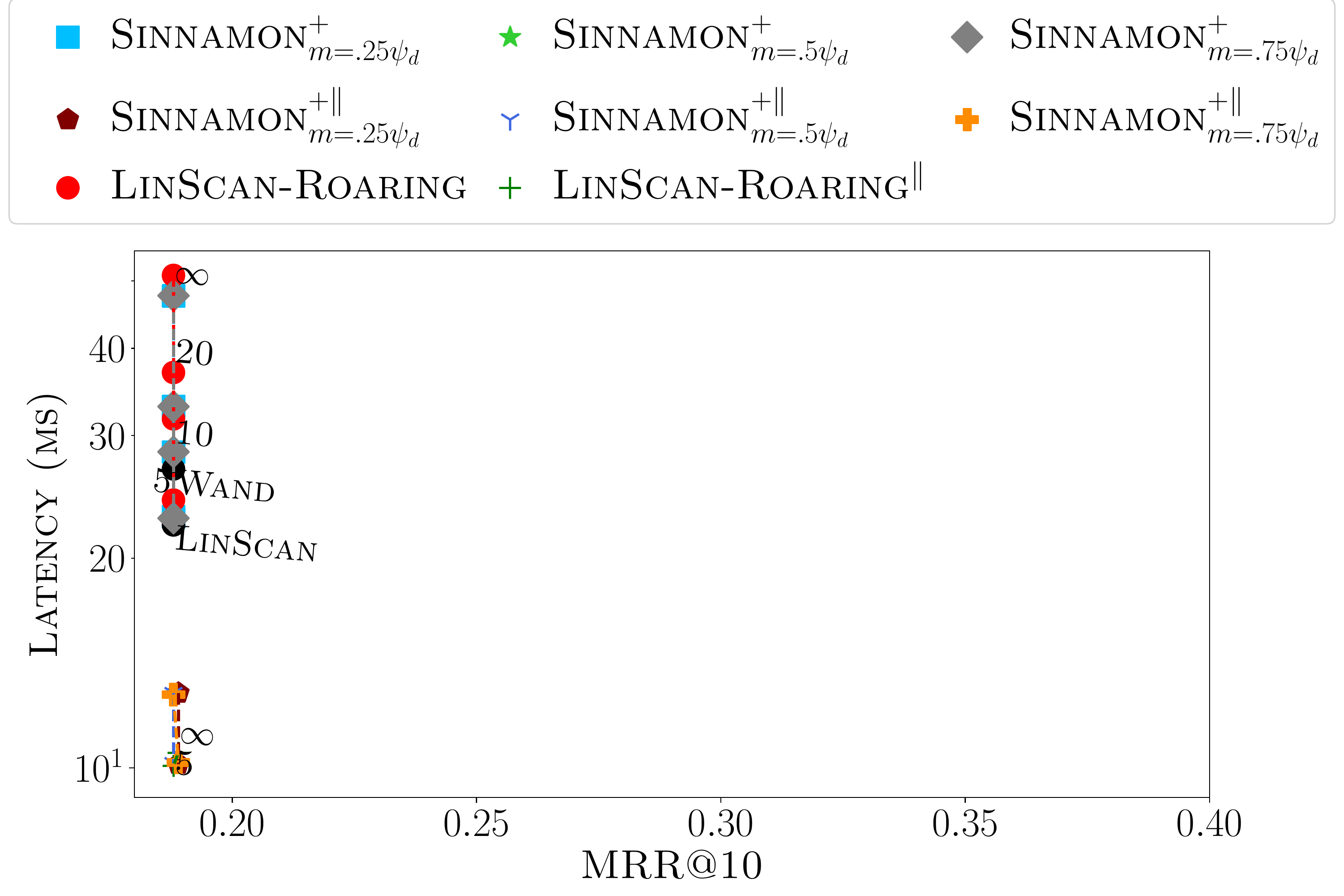}}
}

\centerline{
\subfloat[SPLADE]{
\includegraphics[trim={0 0 0 5.3cm},clip,width=0.4\linewidth]{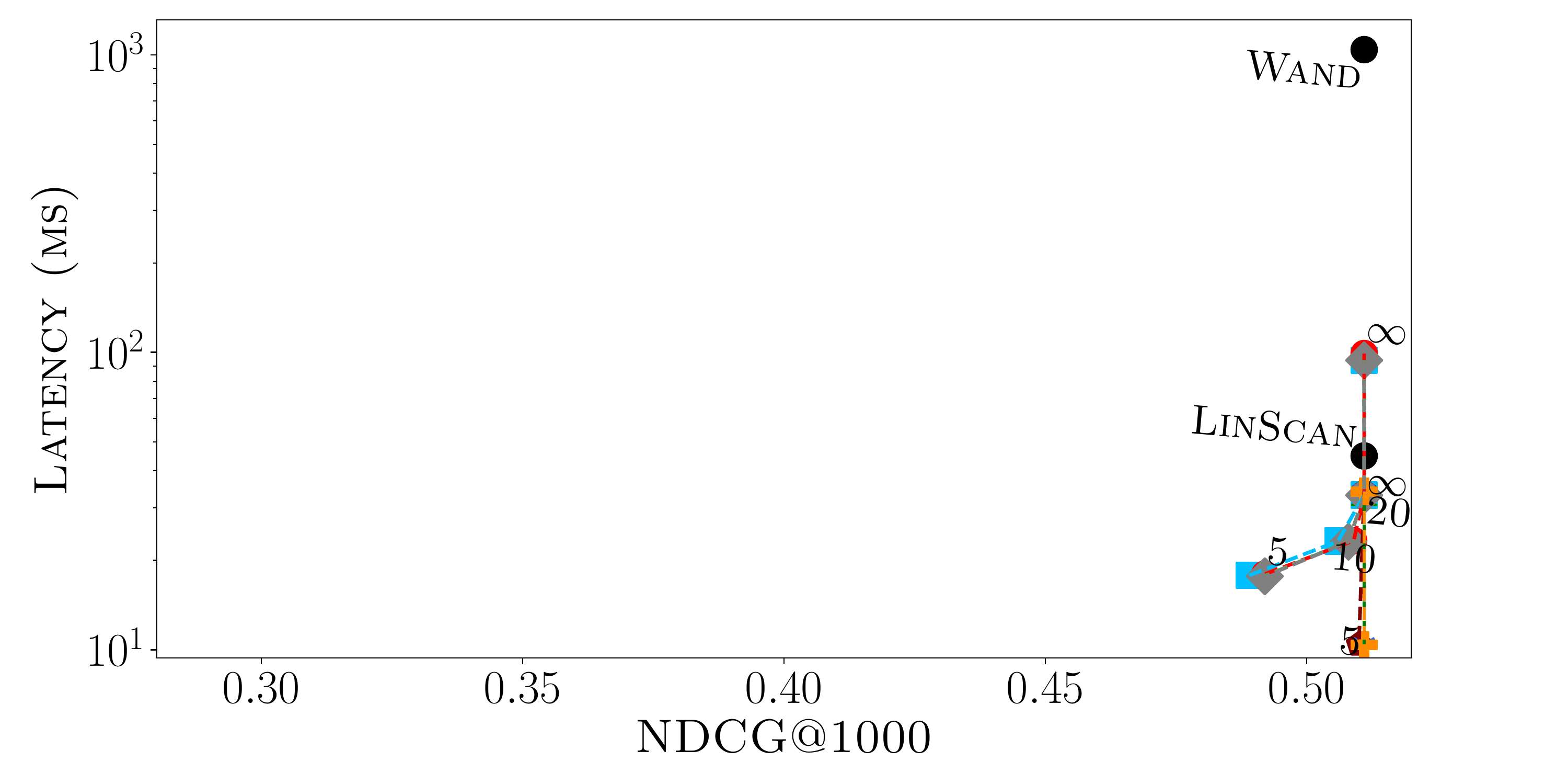}
\includegraphics[trim={0 0 0 5.3cm},clip,width=0.4\linewidth]{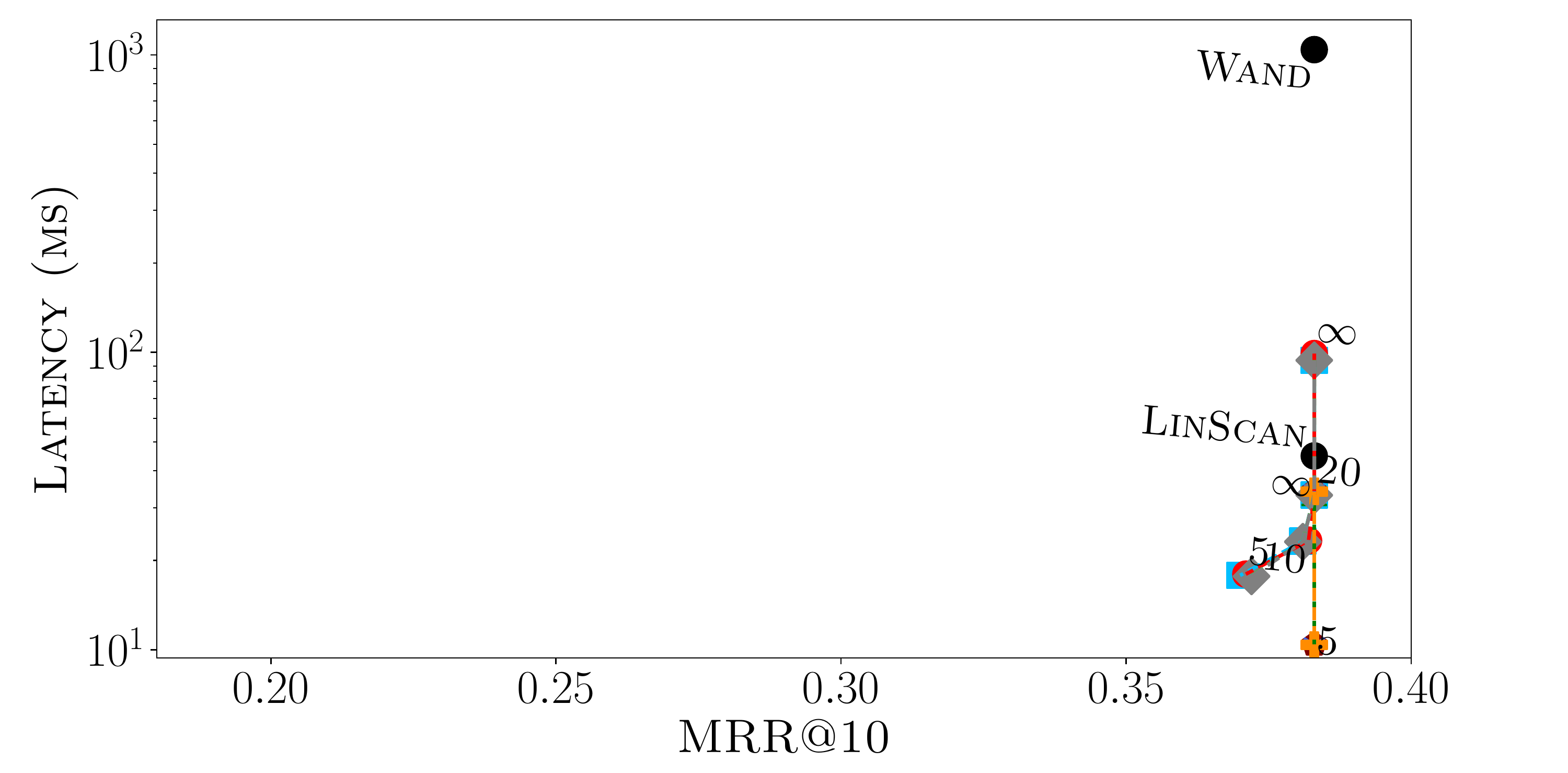}}
}

\centerline{
\subfloat[Efficient SPLADE]{
\includegraphics[trim={0 0 0 5.3cm},clip,width=0.4\linewidth]{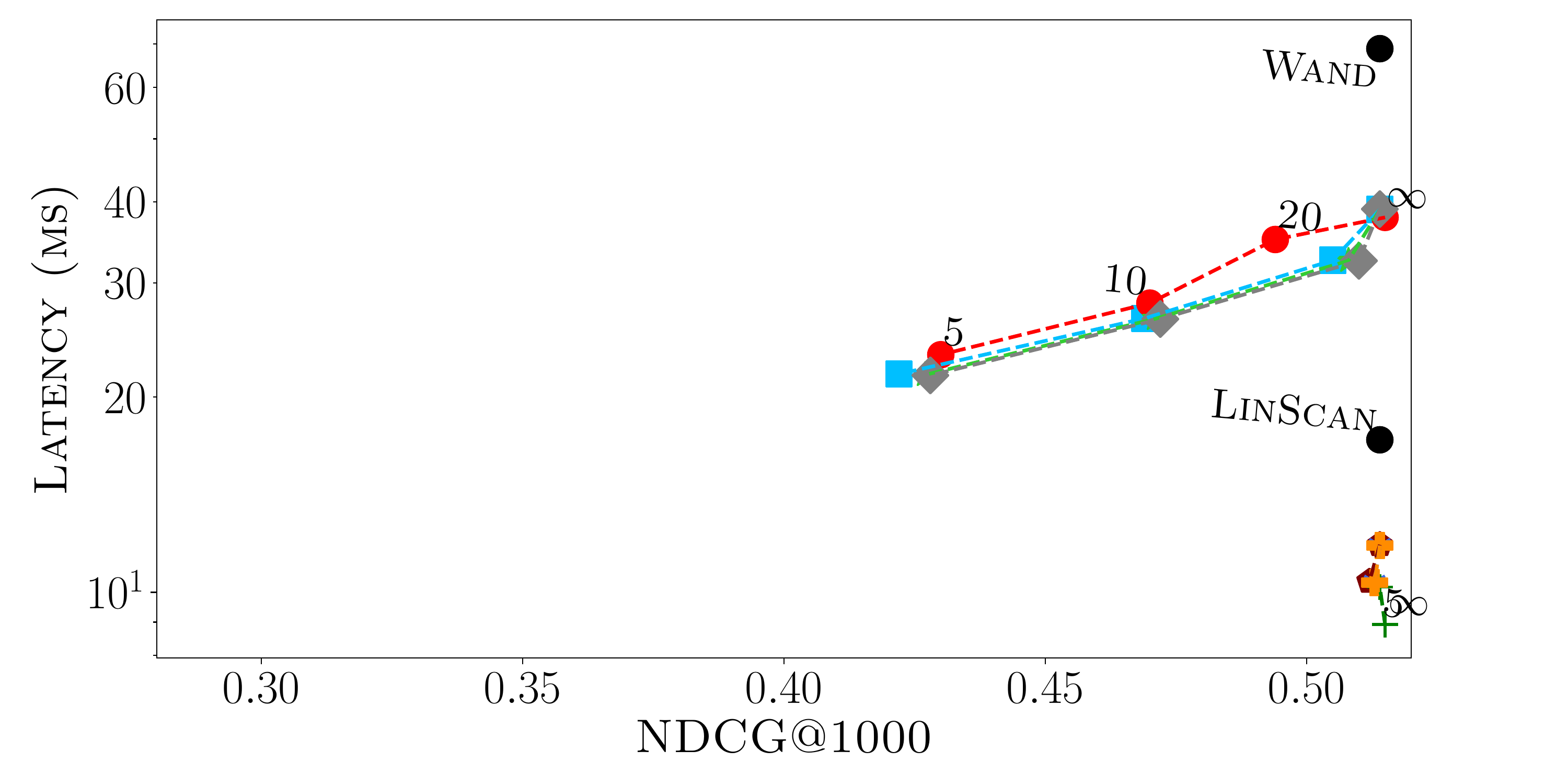}
\includegraphics[trim={0 0 0 5.3cm},clip,width=0.4\linewidth]{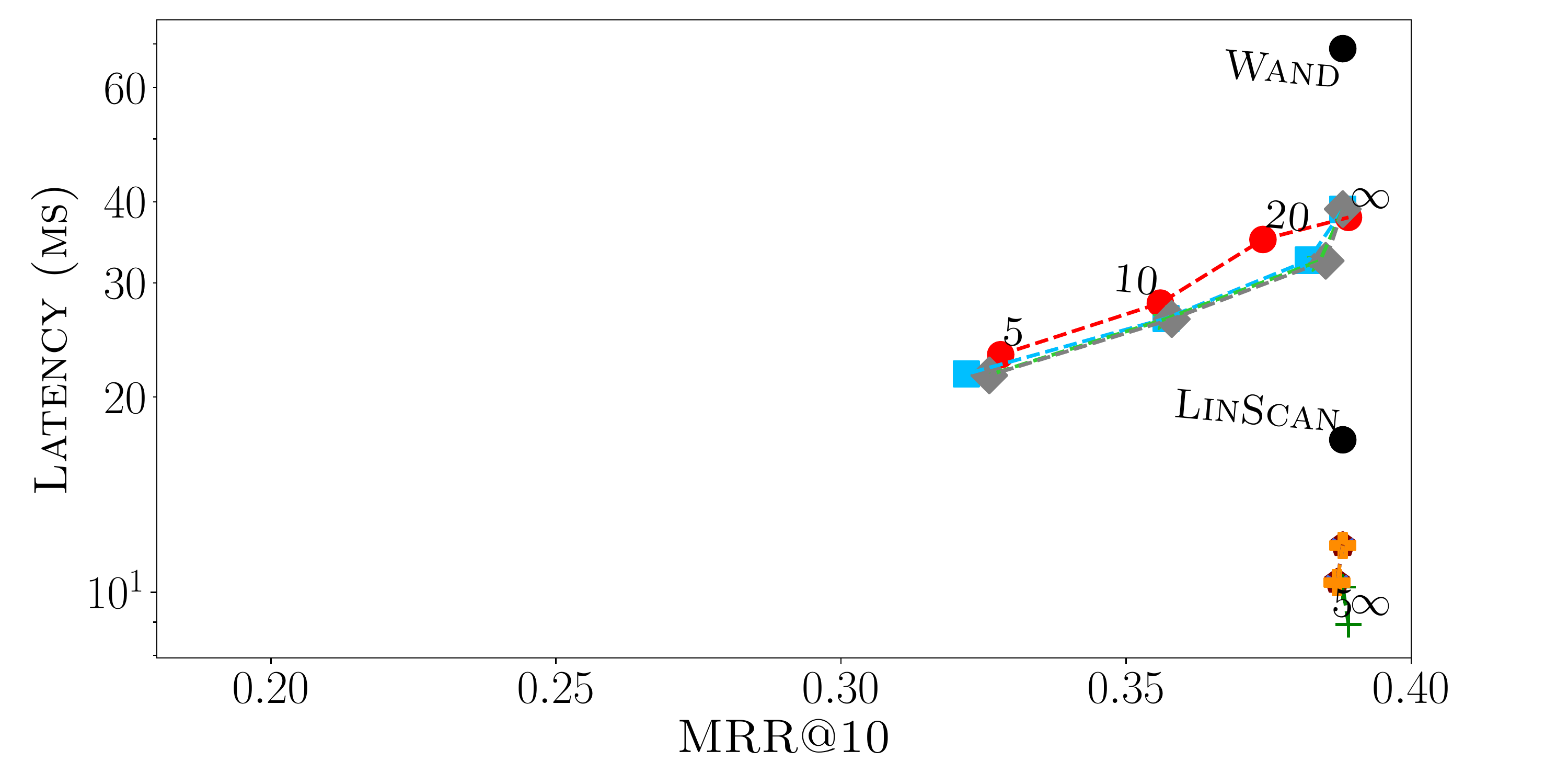}}
}

\centerline{
\subfloat[uniCOIL]{
\includegraphics[trim={0 0 0 5.3cm},clip,width=0.4\linewidth]{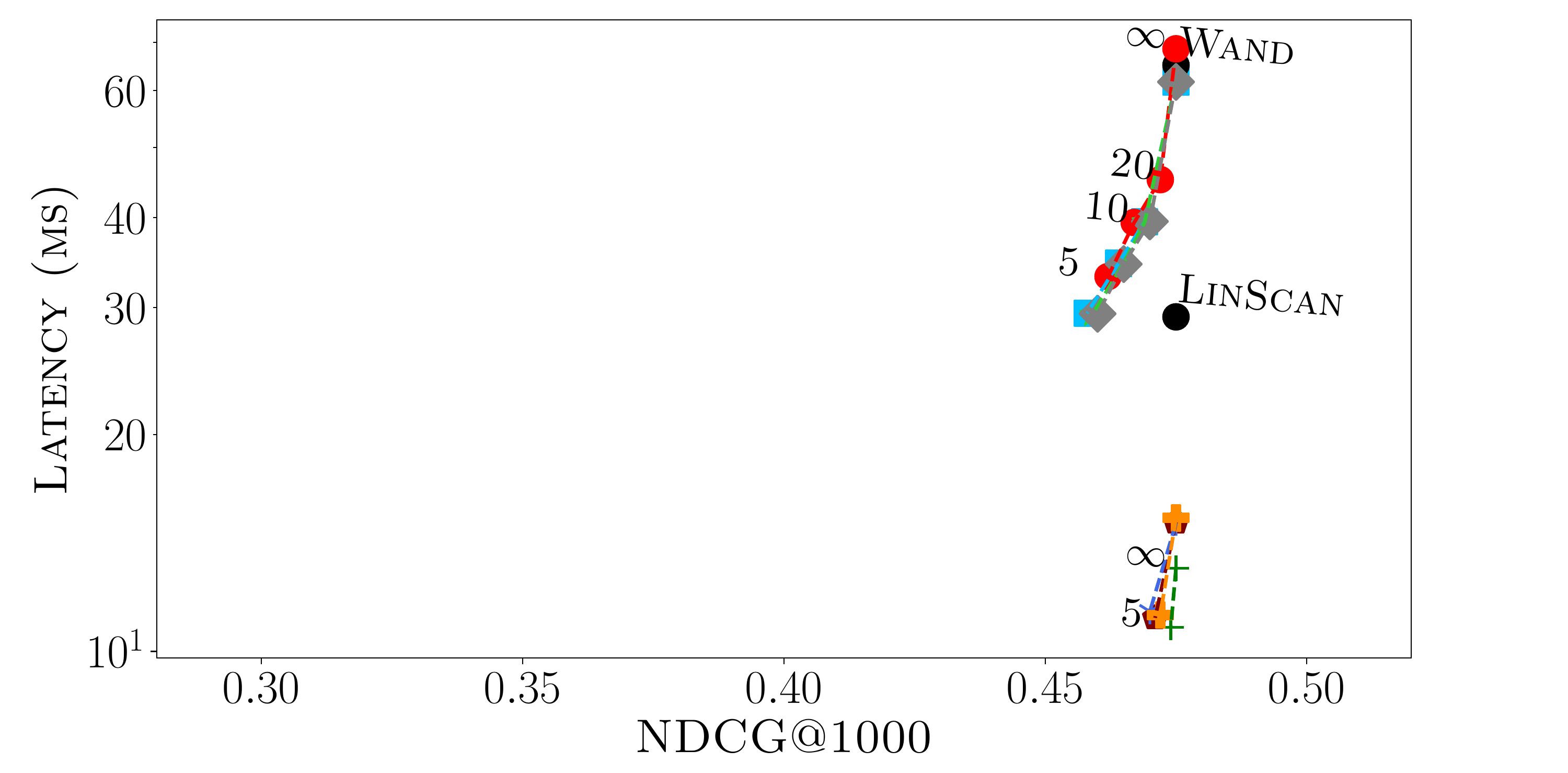}
\includegraphics[trim={0 0 0 5.3cm},clip,width=0.4\linewidth]{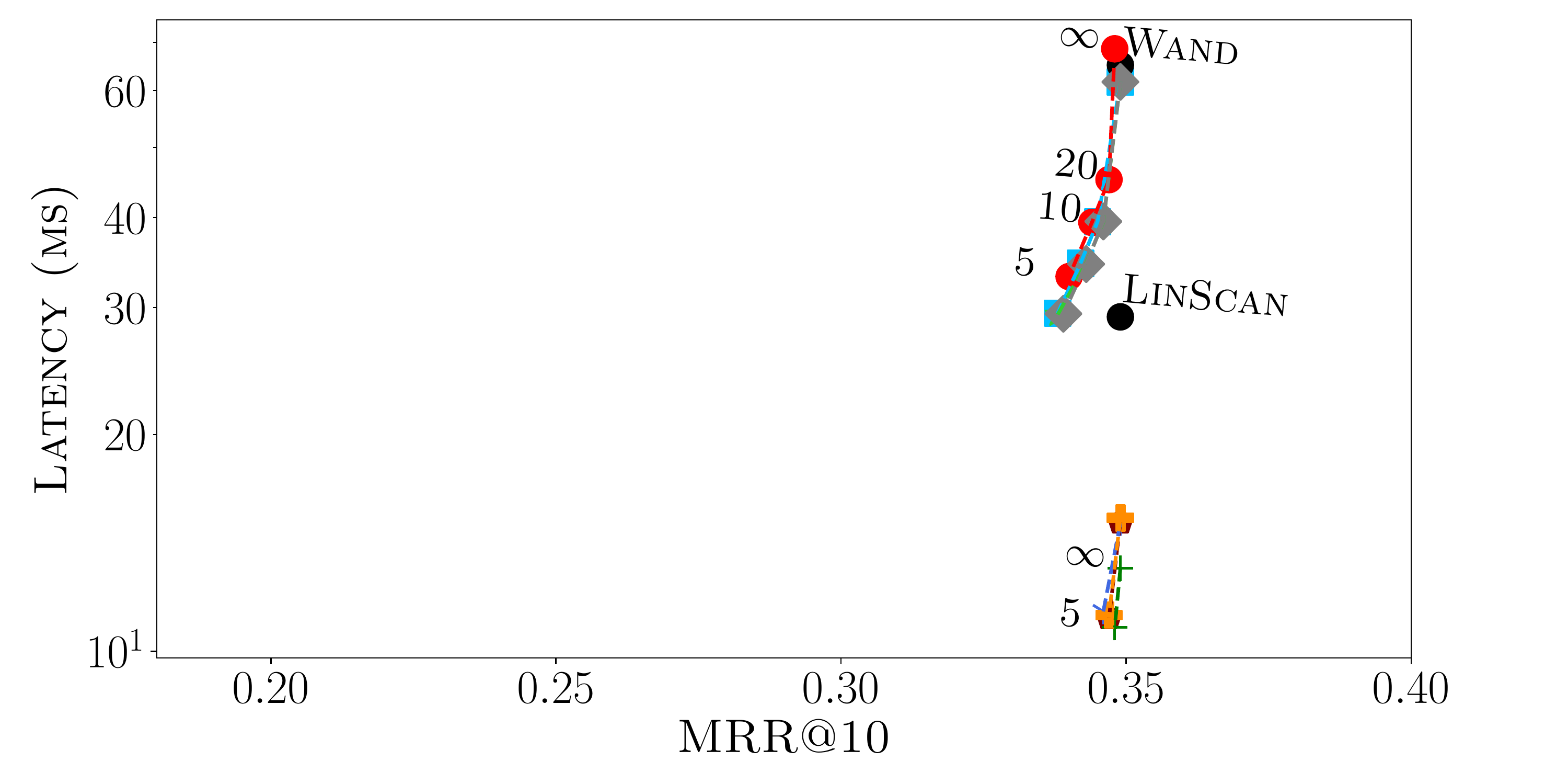}}
}

\caption{Trade-offs on the M1 processor between latency and NDCG$@1000$ (left column), and latency and MRR$@10$ (right column) for various vector collections (rows). As before, shapes (and colors) distinguish between different configurations of \sinnamon{}, and points on a line represent different time budgets $T$ (in milliseconds).}
\label{figure:evaluation:msmarco-passage-v1-m1-end-metric}
\end{center}
\end{figure}
\FloatBarrier
\section{Trade-offs in Retrieval on Apple M1 with $k^\prime=20,000$}
\label{appendix:retrieval-m1-20k}

\begin{figure}[h!t]
\begin{center}
\centerline{
\subfloat[BM25]{
\includegraphics[width=0.4\linewidth]{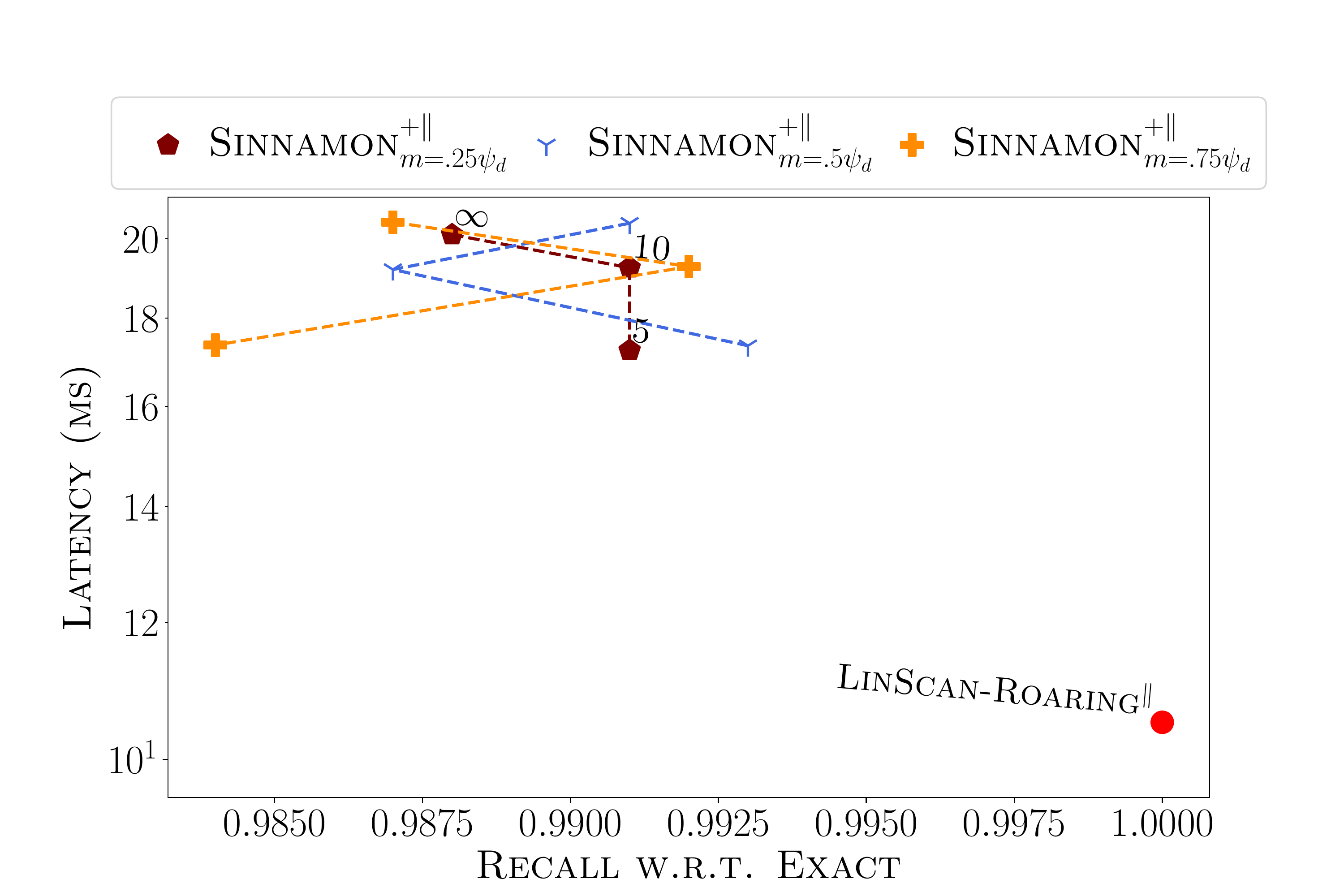}}
\subfloat[SPLADE]{
\includegraphics[width=0.4\linewidth]{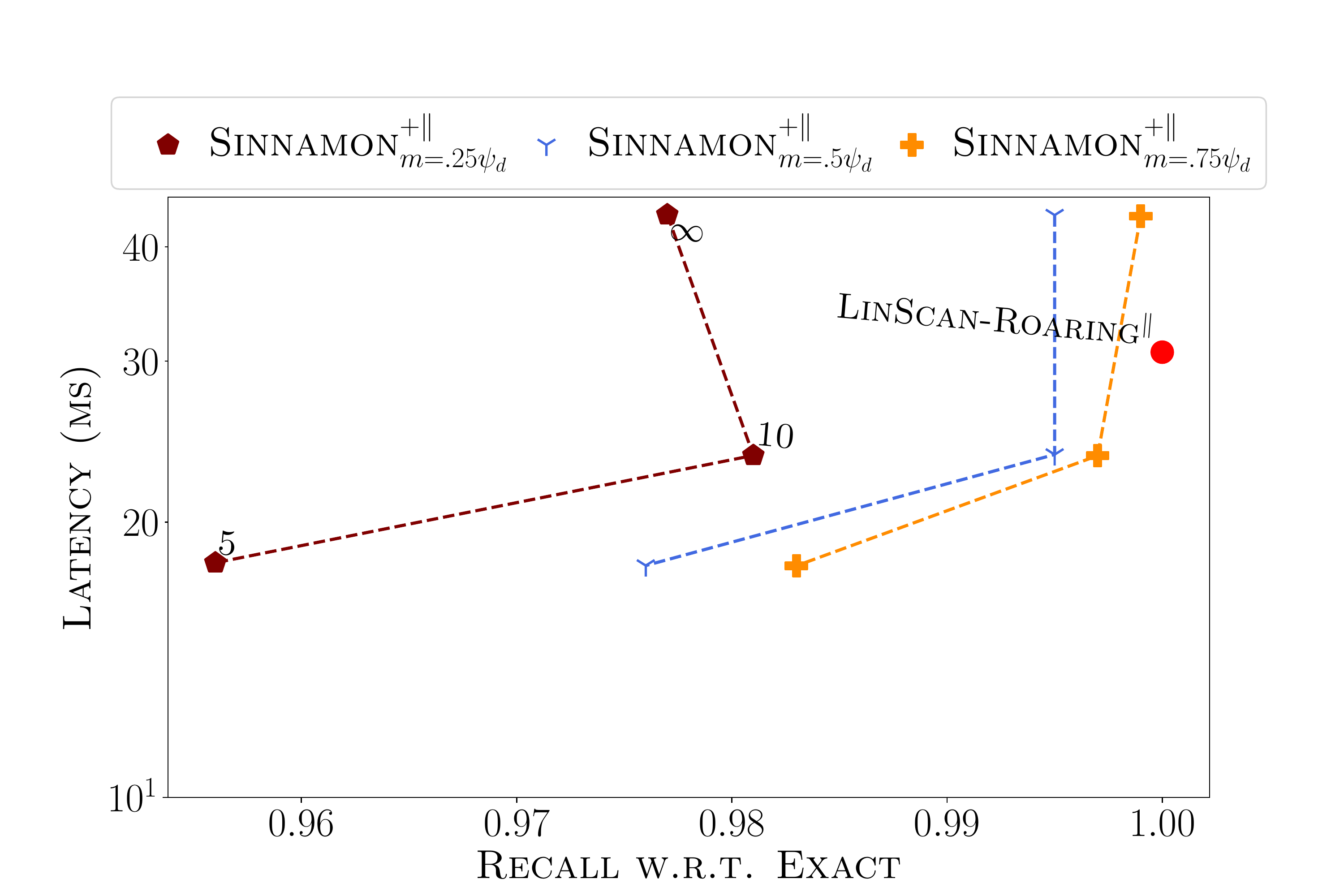}}
}

% \end{center}
% \end{figure}
% \begin{figure}[!ht]
% \begin{center}
% \ContinuedFloat

\centerline{
\subfloat[Efficient SPLADE]{
\includegraphics[width=0.4\linewidth]{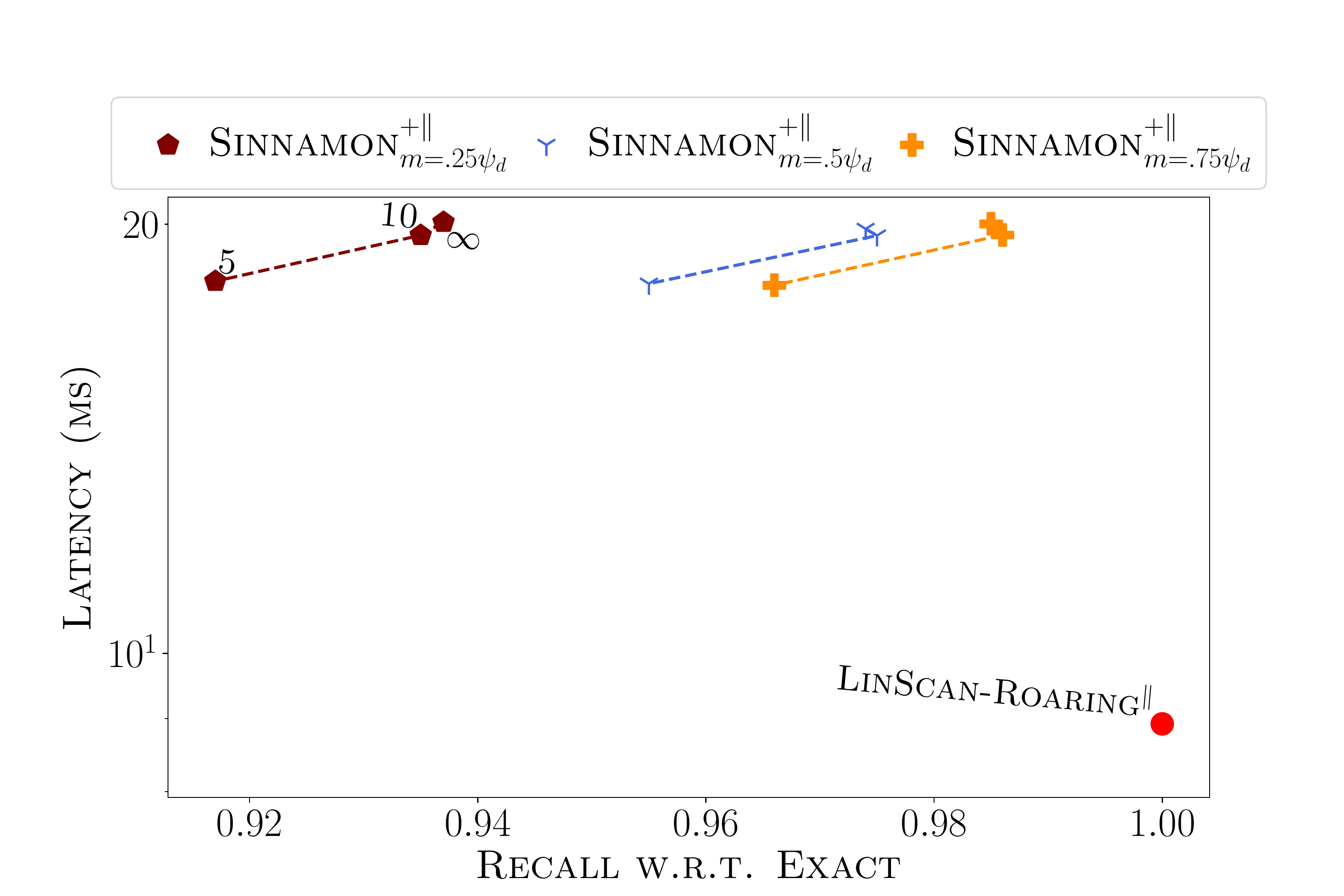}}
\subfloat[uniCOIL]{
\includegraphics[width=0.4\linewidth]{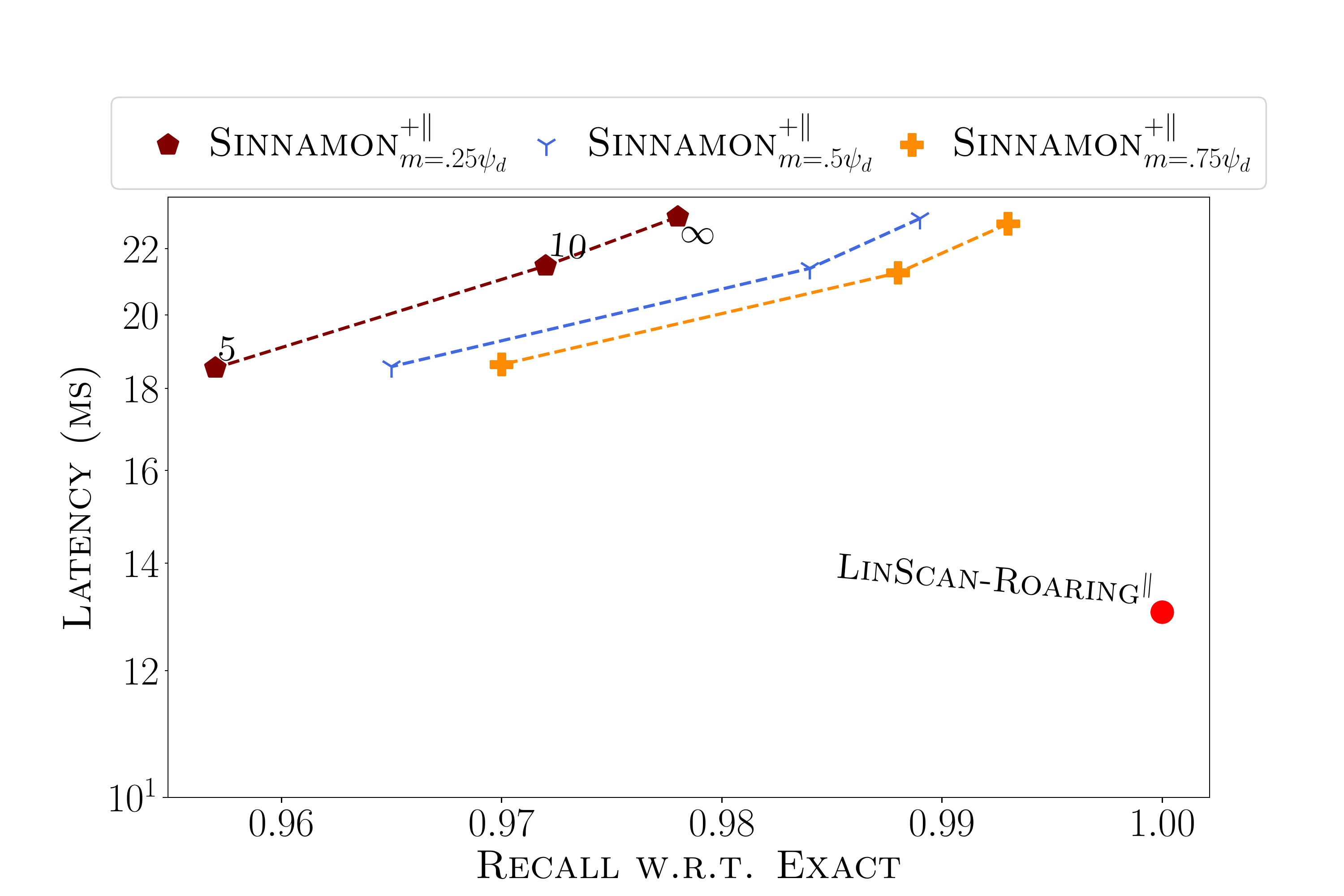}}
}

\caption{Trade-offs on the M1 processor between latency and retrieval accuracy for various vector collections by retrieving $k^\prime=20,000$. Shapes (and colors) distinguish between different configurations of \sinnamon{}, and points on a line represent different time budgets $T$ (in milliseconds).}
\label{figure:evaluation:msmarco-passage-v1-m1-20k}
\end{center}
\end{figure}

\FloatBarrier
\section{Insertions and Deletions on Apple M1}
\label{appendix:benchmark-m1}

\begin{figure}[!ht]
\begin{center}
\centerline{
\subfloat[Insertions]{
\includegraphics[width=0.42\linewidth]{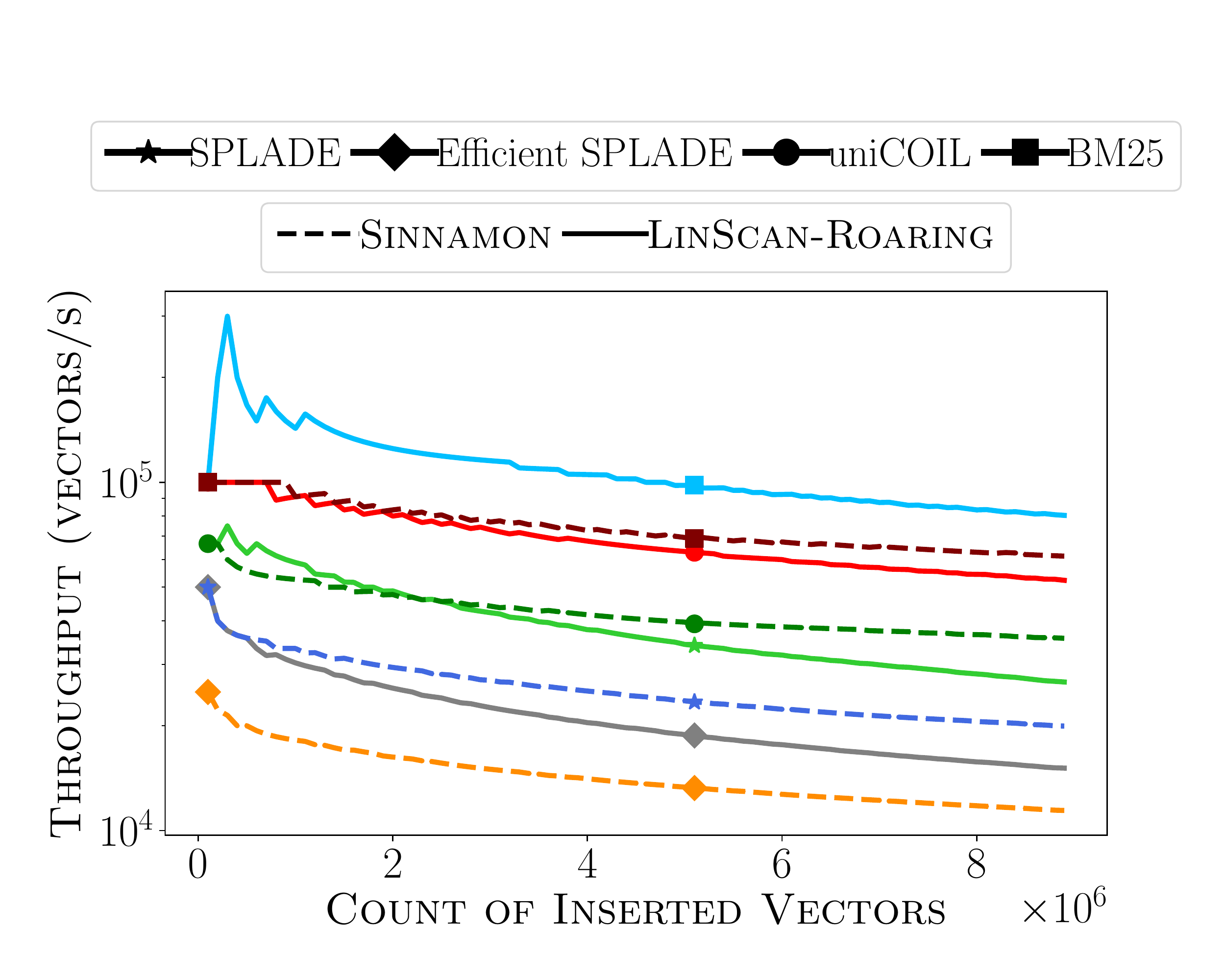}\label{figure:evaluation:msmarco-passage-v1-m1-benchmark:insertions}}
\subfloat[Deletions]{
\includegraphics[width=0.42\linewidth]{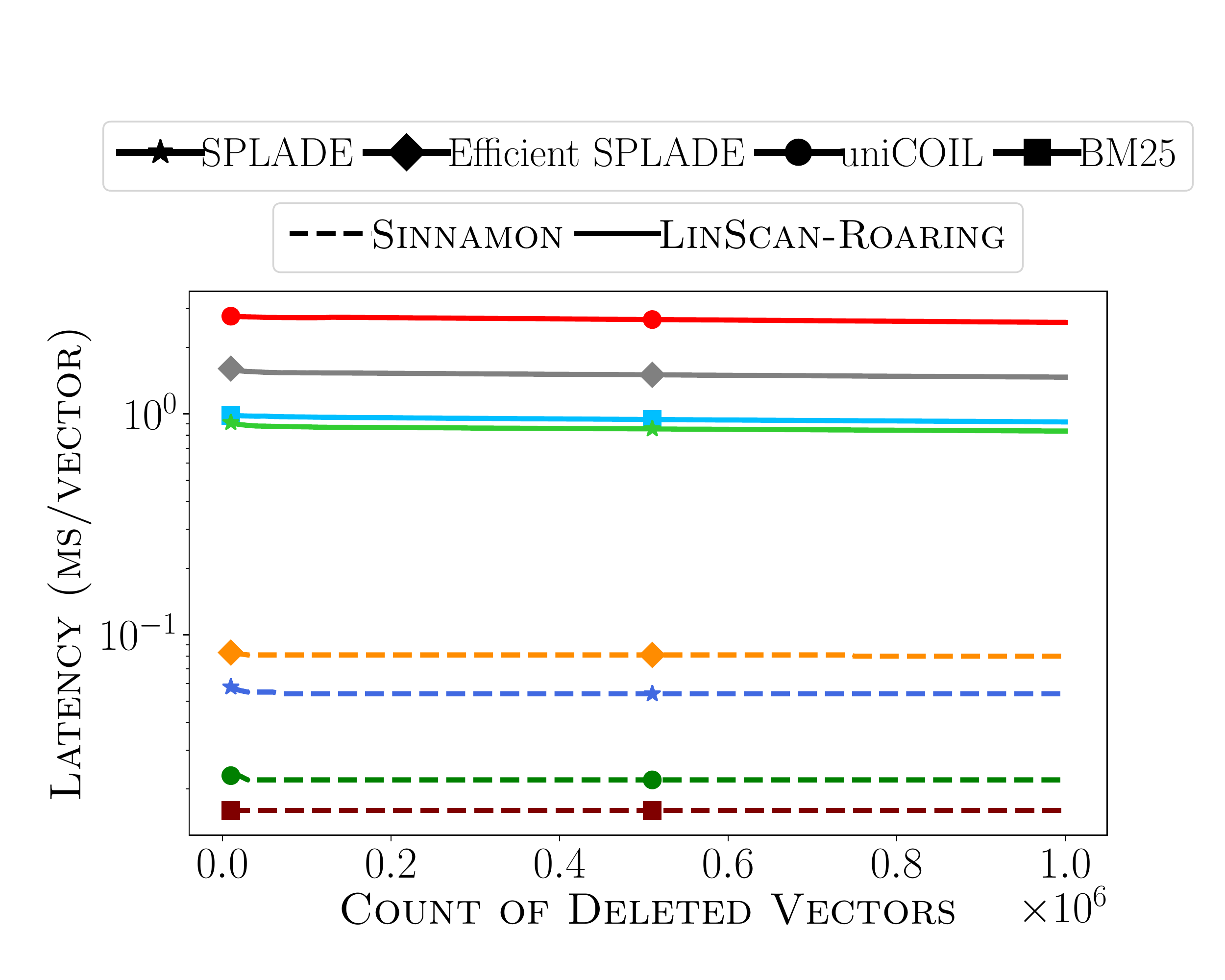}\label{figure:evaluation:msmarco-passage-v1-m1-benchmark:deletions}}
}
\caption{Indexing throughput (vectors per second) and deletion latency (milliseconds per vector) as a function of the size of the index on the M1 processor.}
\label{figure:evaluation:msmarco-passage-v1-m1-benchmark}
\end{center}
\end{figure}

\end{document}